\pdfoutput=1

\documentclass[prb,superscriptaddress,amsmath,amssymb,onecolumn]{revtex4-2}

\usepackage{graphicx,bm,amsmath}
\usepackage[usenames,dvipsnames]{xcolor}
\usepackage[colorlinks,bookmarks=false,citecolor=NavyBlue,linkcolor=Red,urlcolor=blue]{hyperref}

\usepackage[bbgreekl]{mathbbol}
\usepackage{xcolor}
\usepackage[normalem]{ulem}
\usepackage{algorithm}
\usepackage{braket}
\usepackage{algpseudocode}
\usepackage{mathtools}
\usepackage{comment}
\usepackage{amsthm}
\usepackage{enumerate}
\usepackage{enumitem}
\usepackage{tikz}
\usetikzlibrary{arrows,calc,decorations.markings,math,arrows.meta,shapes.geometric}
\usepackage{subfigure}
\usetikzlibrary{calc}
\usepackage{physics}
\usepackage{bm}
\newcommand{\kket}[1]{\lvert #1 \rangle\!\rangle}
\newcommand{\bbra}[1]{\langle\!\langle #1 \rvert}
\newcommand{\bbrakket}[2]{\langle\!\langle #1 \vert #2 \rangle\!\rangle}

\def\doi{http://dx.doi.org/}

\newcommand{\be}{\begin{equation}}
\newcommand{\ee}{\end{equation}}
\newcommand{\bec}{\begin{equation*}}
\newcommand{\eec}{\end{equation*}}
\newcommand{\bea}{\begin{eqnarray}}
\newcommand{\eea}{\end{eqnarray}}

\usepackage{enumitem}
\usepackage{calc}

\usepackage{tikz}
\usetikzlibrary{quantikz2}
\usepackage{booktabs}
\usepackage{pifont}

\usepackage{dsfont}
\newcommand{\G}{\mathcal{G}}
\newcommand{\M}{\mathcal{M}}
\newcommand{\F}{\mathcal{F}}
\newcommand{\nGMax}{\Lambda_\text{d}}
\newcommand{\nGF}{\F}
\newcommand{\nGR}{\M_R}
\newcommand{\nGLambda}{\M_\Lambda}

\newtheorem{thm}{Theorem}
\newtheorem*{thm_intro}{Theorem}
\newtheorem*{cor_intro}{Corollary}

\newtheorem{prop}{Proposition}
\newtheorem{cor}{Corollary}
\newtheorem{lem}{Lemma}
\newtheorem{defn}{Definition}

\newcommand{\titleinfo}{Fermionic non-Gaussianity via Bell sampling: monotones and efficient quantum algorithms}
\date{\today}

\begin{document}
\newcommand{\TUM}{\affiliation{Technical University of Munich, TUM School of Natural Sciences, Physics Department, James-Franck-Straße 1,
85748 Garching, Germany}}
\newcommand{\MCQST}{\affiliation{Munich Center for Quantum Science and Technology (MCQST), Schellingstraße 4, 80799 M{\"u}nchen, Germany}}

\title{\titleinfo}
\author{Poetri Sonya Tarabunga} 
\TUM \MCQST
\email{poetri.tarabunga@tum.de}

\begin{abstract}
Fermionic non-Gaussianity is an essential resource for unlocking the full computational power of fermionic quantum platforms. In this work we develop monotones and efficient quantum algorithms for fermionic non-Gaussianity, all built on the eigenvalue structure of the operator $\Lambda = \sum_{j=1}^{2n}\gamma_j\otimes\gamma_j$ defined on two copies of an $n$-mode fermionic state, accessible via Bell sampling. In particular, we introduce the \emph{bridge degree} of even pure states, a novel non-Gaussianity monotone defined as the largest eigenvalue of $\Lambda$ whose eigenspace is populated by two copies of the state. Our key technical result is that the bridge degree is non-increasing under post-selected Gaussian protocols, which yields no-go theorems for Gaussian conversion beyond the reach of previously known monotones and shows that the resource theory of fermionic non-Gaussianity is irreversible in the exact-conversion setting. Beyond this, the bridge degree exhibits several further features: it (i) is easy to compute, (ii) is efficiently witnessed through Bell sampling, (iii) lower-bounds the non-Gaussian gate complexity of state preparation, (iv) controls the non-Gaussian gate complexity of producing quantum state designs, and (v) naturally extends to mixed states via the Choi--Jamio{\l}kowski isomorphism. We further develop an approximate variant together with an efficiently measurable lower bound, yielding an experimentally certifiable lower bound on the non-Gaussian cost of approximately preparing any state, based directly on Bell-sampling data. Finally, the same eigenvalue structure underlies two Bell-sampling-based algorithmic primitives, both with polynomial sample complexity: a two-copy Gaussianity test with perfect completeness, optimal among two-copy tests sharing this property, and a test for the state $2$-design property of matchgate-invariant ensembles.
\end{abstract}

\maketitle

\tableofcontents

\section{Introduction}
\label{sec:intro}

Quantum computers hold the promise of solving problems beyond the reach of any classical computer --- a notion known as \emph{quantum advantage}~\cite{preskill2012quantum,harrow2017quantum}. Establishing this advantage requires preparing quantum states that no classical computer can efficiently simulate. Identifying the boundary between classical simulability and quantum advantage is therefore a central challenge, both theoretically and practically. A natural strategy is to study families of quantum-computation models that admit efficient classical simulation and to characterize the additional \emph{resources} required to promote them to universality~\cite{Gour2025,chitambar2019quantum}. Quantifying these resources is essential for evaluating the computational potential of any quantum information-processing platform.

Within fermionic quantum computation, the natively free operations are the fermionic Gaussian unitaries~\cite{knill2001fermionic} --- also called matchgates, or fermionic linear optics --- which act linearly on the Majorana operators and are generated by Hamiltonians quadratic in fermion operators. They are efficiently simulable on a classical computer~\cite{valiant2001quantum,terhal2002classical,jozsa2008matchgates,bravyi2005lagrangian,brod2016efficient,surace2022fgs}, and the states they generate --- the fermionic Gaussian states (FGSs) --- form an important class of quantum states ubiquitous across physics, including BCS superconducting states~\cite{Altland2010}, mean-field Hartree--Fock states~\cite{Hartree1928}, and the ground states of Kitaev's honeycomb lattice model~\cite{kitaev2006anyon}. Promoting matchgate circuits to universal quantum computation requires non-Gaussian resources, which can be supplied by injecting non-Gaussian magic states~\cite{hebenstreit2019all}. Quantifying how much non-Gaussianity a state contains is therefore both a structural question about fermionic complexity and a practical question about the cost of fermionic computation. This question is of high experimental relevance: fermionic degrees of freedom are now directly controllable across platforms such as ultracold atoms in optical lattices~\cite{Koepsell2021,Brown2019,Hartke2023,Xu2025,vijayan2020timeresolved,jordens2008mott} and digital quantum simulators of fermionic systems~\cite{GonzlezCuadra2023,Chien2026,Bojovi2026}.

The importance of non-Gaussian resources has prompted a growing body of work proposing monotones for the resource theory of fermionic non-Gaussianity~\cite{dias2024classical,cudby2023gaussian,reardonSmith2024improved,tarabunga2026,gottlieb2006properties,gottlieb2014correlations,lyu2024fermionicgaussiantesting}, chief among them the Gaussian extent, the Gaussian rank, the Gaussian fidelity~\cite{dias2024classical,cudby2023gaussian,reardonSmith2024improved}, and the relative entropy of non-Gaussianity~\cite{tarabunga2026,gottlieb2006properties,gottlieb2014correlations,lyu2024fermionicgaussiantesting}. However, every existing monotone suffers from at least one of two shortcomings: it is either (i) not efficiently computable, requiring optimization over Gaussian decompositions, or (ii) not known to carry an operational interpretation. The Gaussian rank, for instance, directly controls classical simulability but is defined through an optimal Gaussian decomposition that is not known to be efficiently obtainable; the relative entropy of non-Gaussianity, by contrast, is easy to compute from the covariance matrix but lacks an operational meaning tied to state complexity. A range of further quantifiers have been proposed~\cite{turner2017optimal,pachos2018quantifying,pachos2022quantifying,meichanetzidis2018free,coffman2025measuringnongaussianmagicfermions,sierant2025fermionicmagicresourcesquantum}, but they are not known to be monotones in the resource-theoretic sense. As a result, a fermionic non-Gaussianity monotone that is simultaneously computable and operationally meaningful is still lacking --- in stark contrast to the situation for other prominent resource theories such as entanglement~\cite{vidal2000entanglement}, coherence~\cite{baumgratz2014quantifying}, non-stabilizerness~\cite{bittel2026operational}, and bosonic non-Gaussianity~\cite{mele2025symplecticrank}.

We address this gap by introducing a new quantity, the \emph{bridge degree}. The central object is the \emph{bridge operator} $\Lambda$, originally introduced by Bravyi~\cite{bravyi2005lagrangian}~\footnote{The name bridge operator was coined later in Ref.~\cite{sierant2026theorymatchgatecommutant}.} as a Gaussianity witness on two copies of a fermionic state: $\Lambda\ket{\psi}^{\otimes 2}=0$ if and only if $\psi$ is an FGS. We elevate the role of $\Lambda$ from a binary test to a functional of the state by analyzing its full eigenvalue distribution on $\ket{\psi}^{\otimes 2}$. Crucially, this distribution is directly sampled by \emph{Bell sampling} --- a protocol that measures two copies of the state in the Bell basis, implementable as a depth-$1$ Clifford circuit on qubit platforms or as a single matchgate beam-splitter layer on fermionic platforms. Bell sampling is a powerful tool for quantum information tasks~\cite{hangleiter2024bell,huang2022quantum,huang2021information,haug2023scalable,montanaro2017learning,gross2021schur,grewal2024efficient,king2024triplyefficient,nielsen2011quantum}, often providing substantial advantage over single-copy measurements~\cite{huang2022quantum,huang2021information,chen2022exponential}, and has been demonstrated experimentally~\cite{bluvstein2024logical,huang2022quantum,haug2023scalable,islam2015measuring,Kaufman2016,Linke2018,haug2026efficient} as well as employed in numerical simulations of quantum systems~\cite{tarabunga2025bell,tarabunga2025efficient,lami2023nonstabilizerness,haug2023stabilizer,tarabunga2023many,collura2025quantummagicfermionic}.

The bridge degree of an even pure state $\psi$ is defined as the largest (in magnitude) eigenvalue of $\Lambda$ whose eigenspace is occupied by $\ket{\psi}^{\otimes 2}$, and is therefore both easy to read off from a state description and easy to witness experimentally via Bell sampling. Our central technical result is that the bridge degree is non-increasing under pure-state post-selected Gaussian protocols, which immediately yields lower bounds on the non-Gaussian cost of state preparation. As a structural consequence, we use it to establish the irreversibility of the resource theory of fermionic non-Gaussianity in the exact-conversion setting. As an operational consequence, the bridge degree lower-bounds the non-Gaussian resources required to generate quantum state designs --- ensembles of pure states that reproduce the first few statistical moments of the Haar measure.

We further introduce an approximate variant of the bridge degree, which retains monotonicity under Gaussian protocols and is lower-bounded by an efficiently measurable quantity --- the \emph{bridge fidelity} --- accessible directly from Bell sampling. We then extend the framework to mixed states via a Choi-state construction. The resulting \emph{mixed-state bridge degree} inherits the structural properties of its pure-state counterpart and is a monotone under post-selected Gaussian operations.

Finally, building on the same eigenvalue structure of the bridge operator, we develop two concrete algorithmic primitives, both efficient and experimentally accessible. First, we construct a sample-efficient \emph{Gaussianity test} with perfect completeness, optimal among all two-copy tests sharing this property, together with a tolerant variant that distinguishes states close to Gaussian from states far from it. This is directly relevant for verification: it allows an experimentalist to certify, from measurements alone, whether the state produced by a quantum device is close to or far from Gaussian. Second, for \emph{matchgate-invariant} ensembles, we show that the state $2$-design property is efficiently certifiable from polynomially many Bell samples. State designs underpin many protocols --- including randomized benchmarking~\cite{knill2008randomized}, classical shadows~\cite{huang2020predicting}, quantum cryptography~\cite{ji2018pseudorandom,aaronson2012quantum,kretschmer2021quantum}, and random circuit sampling~\cite{boixo2018characterizing, bouland2018complexity} --- so certifying that a candidate ensemble achieves the design property is a prerequisite for trusting them. Our result also addresses the recently introduced complexity-theoretic question of design certification~\cite{nakata2025computational}, for which no efficient quantum algorithm is known in general: matchgate-invariant ensembles provide a structurally natural class on which the $2$-design property becomes efficiently testable.

\begin{table*}[t]
\centering
\renewcommand{\arraystretch}{1.2}
\begin{tabular}{@{}l@{\hspace{2em}}llc@{}}
\toprule
\textbf{Quantity} & \textbf{Symbol} & \textbf{Monotone under} & \textbf{Computable} \\
\midrule
Bridge degree (Sec.~\ref{sec:maxlambda}) & $\nGMax(\psi)$ & pure-state post-sel.\ Gaussian protocols & \checkmark \\
Approximate bridge degree (Sec.~\ref{sec:approx}) & $\Lambda_{\mathrm{d},\epsilon}(\psi)$ & pure-state Gaussian protocols & \ding{55} \\
Approximate bridge fidelity (Sec.~\ref{sec:approx}) & $\nGF_\epsilon(\psi)$ & --- (lower bound for $\Lambda_{\mathrm{d},\epsilon}$) & \checkmark \\
Mixed-state bridge degree (Sec.~\ref{sec:mixed}) & $\tilde\Lambda_{\mathrm{d}}(\rho)$ & post-sel.\ Gaussian operations & \checkmark \\
Convex-roof extended bridge degree (Sec.~\ref{sec:maxlambda}) & $\nGMax^{\mathrm{cr}}(\rho)$ & post-sel.\ Gaussian protocols & \ding{55} \\
\bottomrule
\end{tabular}
\caption{Summary of the non-Gaussianity monotones and related quantities introduced in this work. The ``Monotone under'' column specifies the class of free operations under which the quantity is provably non-increasing. See Defs.~\ref{def:gaussian_operations} and~\ref{def:gaussian_protocol} for the definition of Gaussian operations and Gaussian protocols. The ``Computable'' column marks (\checkmark) quantities accessible from a state description without optimization, and (\ding{55}) those whose definition requires an optimization not known to be efficient.}
\label{tab:monotones_summary}
\end{table*}

\begin{table*}[t]
\centering
\renewcommand{\arraystretch}{1.2}
\begin{tabular}{@{}l@{\hspace{2em}}ll@{}}
\toprule
\textbf{Task} & \textbf{Sample complexity} & \textbf{Reference} \\
\midrule
Estimate $\nGF_\epsilon(\psi)$ & $\mathcal{O}\!\bigl(\Delta^{-2}\log(1/\Delta)\bigr)$ & Lemma~\ref{lem:sample_complexity_approx} \\
One-sided Gaussianity test ($F_\G=1$ vs.\ $\leq1-\!\epsilon$) & $\mathcal{O}(n^{2}/\epsilon)$  & Theorem~\ref{thm:gaussianity_test} \\
\quad for $t$-doped Gaussian states & $\mathcal{O}(t^{2}/\epsilon)$ & Theorem~\ref{thm:test_tdoped} \\
Tolerant Gaussianity test ($F_\G\!\geq\!1-\alpha$ vs.\ $\leq\!1-\beta$) & $\mathcal{O}(n^{2}/(\beta-2\alpha n^2))^{\ddagger}$ & Theorem~\ref{thm:gaussianity_test_tolerant} \\
Estimate $D(\mathcal{E})$ to additive accuracy $\epsilon$ & $\widetilde{\mathcal{O}}\bigl(\sqrt n/\epsilon^{2}\bigr)^{\S}$ & Lemma~\ref{lem:estimate_design_distance} \\
Approximate $2$-design test ($D(\mathcal{E})\!\leq\!\alpha$ vs.\ $\geq\!\alpha+\epsilon$) & $\widetilde{\mathcal{O}}\bigl(\sqrt n/\epsilon^{2}\bigr)^{\S}$ & Theorem~\ref{thm:test_design_distance} \\
Exact $2$-design test ($D(\mathcal{E})\!=\!0$ vs.\ $\geq\!\epsilon$) & $\mathcal{O}\bigl(n^{1/4}/\epsilon^{2}\bigr)^{\S}$ & Theorem~\ref{thm:nontolerant_test_design} \\
\bottomrule
\end{tabular}
\caption{Sample-efficient Bell-sampling protocols. All bounds hold with success probability at least $1-\delta$; the explicit $\log(1/\delta)$ dependence is suppressed. $F_\G$ denotes the Gaussian fidelity (see Eq.~\eqref{eq:gfid_def}). $D(\mathcal{E})$ denotes the approximation error of the state $2$-design (see Eq.~\eqref{eq:design_def_approx}). $\Delta$ denotes the gap between the cumulative bridge spectral distribution and the threshold $(1-\epsilon)^2$ (see Eq.~\eqref{eq:Delta_def}). The notation $\widetilde{\mathcal{O}}$ hides additional logarithmic factors.\\
$^{\ddagger}$ Valid for $\alpha = \mathcal{O}(\beta/n^{2}-2\alpha)$. \\
$^{\S}$ Requires the ensemble $\mathcal{E}$ to be matchgate-invariant.}
\label{tab:bell_sampling_protocols}
\end{table*}

\subsection{Overview of main results}
\label{sec:overview}

This section provides an overview of our key findings; the main quantities and the efficient protocols that we introduce are summarized in Tables~\ref{tab:monotones_summary} and~\ref{tab:bell_sampling_protocols}.
Our results fall under six themes: (i) the bridge degree as a fermionic non-Gaussianity monotone, (ii) no-go theorems and irreversibility of the resource theory, (iii) lower bounds on non-Gaussian gates for state designs, (iv) the approximate bridge degree and approximate conversion, (v) the mixed-state extension via the Choi--Jamio{\l}kowski isomorphism, and (vi) algorithmic primitives from Bell sampling. We summarize each below; formal statements and proofs are provided in the subsequent sections.

\textbf{(i) The bridge degree as a fermionic non-Gaussianity monotone (Sec.~\ref{sec:maxlambda}).}
We introduce a fermionic non-Gaussianity monotone for even pure states called the \emph{bridge degree}, built around the bridge operator,
\begin{equation} 
    \Lambda \coloneqq \sum_{j=1}^{2n} \gamma_j \otimes \gamma_j
\end{equation}
introduced by Bravyi~\cite{bravyi2005lagrangian} as a tool to detect Gaussian states on two copies of an $n$-mode fermionic system. 

We develop the theory first in the pure-state setting and extend it to mixed states via the Choi--Jamio{\l}kowski isomorphism in Sec.~\ref{sec:mixed} (theme (v) below). Here, we consider the \emph{convex} resource theory of fermionic non-Gaussianity~\cite{tarabunga2026}, where the free operations are \emph{Gaussian protocols} (Definition~\ref{def:gaussian_protocol}), consisting of Gaussian unitaries, composition with Gaussian states, partial tracing, occupation-number measurements, and conditioned operations. We also consider \emph{post-selected Gaussian protocols}, which further allow post-selection on measurement outcomes. Since Gaussian protocols become universal when augmented with magic-state injection~\cite{hebenstreit2019all}, they provide the natural class of free operations for studying the resource cost of universal fermionic quantum computation.

As a sum of $2n$ commuting Hermitian operators $\gamma_j\otimes\gamma_j$ with eigenvalues $\pm 1$, the spectrum of $\Lambda$ is $\{-2n,-2n+2,\dots,2n\}$; we denote the corresponding eigenspaces and projectors by $V_\lambda$ and $\Pi_\lambda$. For an even pure state $\psi$, the \emph{bridge spectral distribution}
\begin{equation}
    p_\psi(\lambda) \;\coloneqq\; \bra{\psi}^{\otimes 2}\Pi_\lambda\ket{\psi}^{\otimes 2}
\end{equation}
records how the two-copy state $\ket{\psi}^{\otimes 2}$ decomposes across the eigensectors of $\Lambda$ and serves as the unifying object of the entire framework. We show in Sec.~\ref{sec:eigenstructure} that $p_\psi$ is supported on the set $\mathcal{A} = 8\mathbb{Z}\cap[-2n, 2n]$, symmetric about zero, and invariant under matchgate unitaries, and that it is directly measurable via Bell sampling on two copies of $\psi$.

From this distribution, we define the bridge degree for an even pure state $\psi$ as
\begin{equation}
    \nGMax(\psi) \;\coloneqq\; \max\{\,\alpha\geq 0 \,:\, p_\psi(2\alpha) \neq 0\,\},
\end{equation}
that is, the largest eigenvalue sector of $\Lambda$ that $\ket{\psi}^{\otimes 2}$ occupies. We show that $\nGMax$: (i) takes discrete values in $\{0, 4, 8, \dots, 4\lfloor n/4\rfloor\}$; (ii) vanishes if and only if $\psi$ is a pure FGS; (iii) is additive under tensor products, $\nGMax(\psi\otimes\phi)=\nGMax(\psi)+\nGMax(\phi)$; and, most importantly, (iv) is non-increasing under post-selected Gaussian protocols. 
\begin{thm_intro}[Monotonicity of the bridge degree, see Theorem~\ref{thm:max_Lambda_monotone}]
For any even pure state $\ket{\psi}$ and any post-selected Gaussian protocol $\mathcal{E}$ such that $\mathcal{E}(\ket{\psi}\!\bra{\psi}) = \ket{\phi}\!\bra{\phi}$ is pure,
\begin{equation}
    \nGMax(\phi) \;\leq\; \nGMax(\psi).
\end{equation}
\end{thm_intro}

This is the first central result of the paper: the bridge degree is a bona fide non-Gaussianity monotone in the sense of resource theory~\cite{chitambar2019quantum,Gour2025}, and leads to a number of key applications discussed in the following. As an immediate consequence, the bridge degree lower-bounds the number of non-Gaussian gates required for state preparation by a \emph{$t$-doped Gaussian protocol} (Fig.~\ref{fig:t_doped}): a circuit built from arbitrary post-selected fermionic Gaussian protocols interleaved with $t$ non-Gaussian gates, each acting on at most four Majorana modes. 
\begin{cor_intro}[Non-Gaussian gate lower bound for state preparation; see Corollary~\ref{cor:bound_dope}]
If an even pure state $\psi$ can be prepared from the vacuum by a $t$-doped Gaussian protocol, then it must hold that
\begin{equation}
    t \;\geq\; \frac{\nGMax(\psi)}{4}.
\end{equation}
\end{cor_intro}

In particular, any state with $\nGMax(\psi) = \Omega(n)$ requires $\Omega(n)$ non-Gaussian gates to prepare on any fermionic platform with probabilistic protocols.

A particularly appealing feature of this lower bound is that the bridge degree can be witnessed directly from the Bell-sampling primitive (Algorithm~\ref{alg:bell_sampling}). Given $N$ i.i.d.\ samples $\lambda_1,\dots,\lambda_N\sim p_\psi$ obtained from Bell sampling on pairs of copies of $\psi$, the empirical estimator
\begin{equation}
    \widehat\Lambda_{\mathrm{d}} \;\coloneqq\; \tfrac{1}{2}\,\max_{k}|\lambda_k|
\end{equation}
is, by construction, always a lower bound on $\nGMax(\psi)$, and via Corollary~\ref{cor:bound_dope} on the non-Gaussian gate count $t$. In particular, even a \emph{single} Bell sample already certifies $t\geq |\lambda_1|/8$ non-Gaussian gates. %

We remark that the bridge degree admits a mixed-state extension via the standard convex-roof construction,
\begin{equation}
    \nGMax^{\,\mathrm{cr}}(\rho) \;\coloneqq\; \min_{\{p_i,\psi_i\}\,:\,\rho=\sum_i p_i\psi_i}\, \max_i  \nGMax(\psi_i),
\end{equation}
which is automatically monotone under post-selected Gaussian protocols~\cite{vidal2000entanglement,terhal2000schmidt} but generally intractable to compute. As explained below in theme (v), we develop a computationally tractable mixed-state extension --- monotone under a smaller class of operations --- via the Choi--Jamio{\l}kowski isomorphism.

\begin{figure}[t]
\centering
\definecolor{gaussblue}{HTML}{6C8EBF}      %
\definecolor{nongaussorange}{HTML}{E8A87C} %
\begin{quantikz}[column sep=0.35cm, row sep={0.5cm,between origins}]
\lstick[7,brackets=none]{\Large$\ket{\psi}\,=$\,\,\,\,\,} & \lstick{$\ket{0}$\,\,\,\,\,} & \gate[7,style={fill=gaussblue!35,minimum width=0.65cm}]{\mathcal{E}_0} & \gate[2,style={fill=nongaussorange!70,minimum width=0.65cm}]{W_1} & \gate[7,style={fill=gaussblue!35,minimum width=0.65cm}]{\mathcal{E}_1} & \gate[2,style={fill=nongaussorange!70,minimum width=0.65cm}]{W_2} & \gate[7,style={fill=gaussblue!35,minimum width=0.65cm}]{\mathcal{E}_2} & \ \ldots\ & \gate[2,style={fill=nongaussorange!70,minimum width=0.65cm}]{W_t} & \gate[7,style={fill=gaussblue!35,minimum width=0.65cm}]{\mathcal{E}_t} & \rstick[7]{$n$} \\
 & \lstick{$\ket{0}$\,\,\,\,\,} & &  &  &  &  & \ \ldots\ &  &  &  \\
 & \lstick{$\ket{0}$\,\,\,\,\,} & &  &  &  &  & \ \ldots\ &  &  &  \\
 & \lstick{$\vdots$\,\,}  \setwiretype{n} &  \setwiretype{n} &  \setwiretype{n} &  \setwiretype{n} &  \setwiretype{n} &  \setwiretype{n} &  \setwiretype{n} &  \setwiretype{n} &  \setwiretype{n} &  \\
 & \lstick{$\ket{0}$\,\,\,\,\,} & &  &  &  &  & \ \ldots\ &  &  &  \\
 & \lstick{$\ket{0}$\,\,\,\,\,} & &  &  &  &  & \ \ldots\ &  &  &  \\
 & \lstick{$\ket{0}$\,\,\,\,\,} & &  &  &  &  & \ \ldots\ &  &  &  \\
\end{quantikz}
\caption{A \emph{$t$-doped Gaussian protocol}. Blue boxes $\mathcal{E}_0,\mathcal{E}_1,\dots,\mathcal{E}_t$ denote arbitrary Gaussian protocols, possibly with post-selection; orange boxes $W_1,W_2,\dots,W_t$ are non-Gaussian gates. Here, $t$ counts the total number of non-Gaussian gates. Although the four Majorana modes on which $W_i$ acts are in principle arbitrary, we take them without loss of generality to be $\gamma_1,\gamma_2,\gamma_3,\gamma_4$ (the first two qubits), since any quadruple of Majorana modes can be mapped to these by a matchgate unitary that can be absorbed into the neighboring $\mathcal{E}_i$.}
\label{fig:t_doped}
\end{figure}
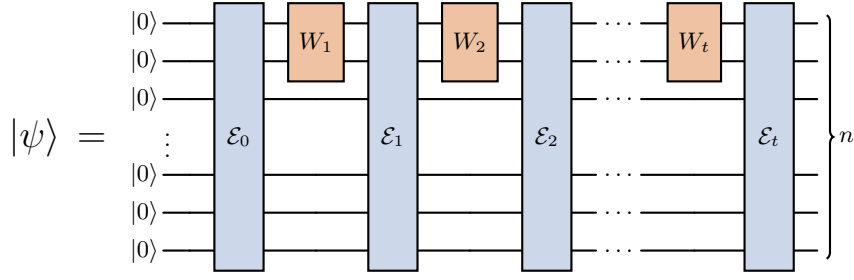

\medskip

\textbf{(ii) No-go theorems and irreversibility of the resource theory (Sec.~\ref{sec:consequences}).}
The monotonicity of the bridge degree is a powerful tool for proving no-go theorems for Gaussian conversion. In particular, additivity of $\nGMax$ under tensor products immediately implies that no Gaussian protocol can map an $n$-qubit non-Gaussian pure state to a tensor product of $m$ four-qubit non-Gaussian pure states with nonzero success probability whenever $4m > n$ (note that all even pure states on $\leq 3$ qubits are Gaussian~\cite{bravyi2005classicalcapacityfermionicproduct}, so four qubits is the smallest non-trivial block size). Remarkably, this impossibility holds even when the input state is arbitrarily far from the set of Gaussian states, while each output factor is arbitrarily close to it. This phenomenon closely parallels a recent no-go theorem for bosonic Gaussian state conversion~\cite{mele2025symplecticrank}.

\begin{cor_intro}[Impossibility of splitting non-Gaussianity by post-selected Gaussian protocols; see Corollary~\ref{cor:spr}]
Let $n<4m$, let $\psi$ be an $n$-qubit even non-Gaussian pure state, and let $\phi_1,\dots,\phi_m$ be four-qubit even non-Gaussian pure states. Then no post-selected Gaussian protocol can map $\psi$ to $\bigotimes_{i=1}^m \phi_i$ with nonzero success probability.
\end{cor_intro}

Beyond such no-go theorems for state conversions, the next natural question is whether the resource theory is \emph{reversible}: can the resources spent in converting $\rho$ to $\sigma$ be fully recovered by converting $\sigma$ back to $\rho$? For instance, entanglement theory is reversible for pure states~\cite{bennett1996concentrating} but irreversible for mixed states~\cite{vidal2001irreversibility}. We investigate this question for the asymptotic-conversion rate $n\mapsto m$ of $\rho^{\otimes n}\to\sigma^{\otimes m}$ under probabilistic Gaussian protocols (i.e., where the conversion is required to succeed exactly with some nonzero probability), and prove the following:

\begin{thm_intro}[Irreversibility of the resource theory of fermionic non-Gaussianity; see Theorem~\ref{thm:irreversibility}]
The resource theory of fermionic non-Gaussianity is irreversible under post-selected Gaussian protocols in the exact-conversion setting.
\end{thm_intro}

These results mirror the irreversibility theorem for bosonic non-Gaussianity~\cite{mele2025symplecticrank}, with the bridge degree playing the role of the symplectic rank.

\medskip

\textbf{(iii) Lower bounds on non-Gaussian gates for state designs (Sec.~\ref{sec:designs_lower}).}
A further operational consequence of bridge degree monotonicity is a lower bound on the non-Gaussian resources required to generate \emph{state designs} --- ensembles of pure states that reproduce or approximate the first $k$ statistical moments of the Haar measure~\cite{ambainis2007quantum}. Formally, an ensemble $\mathcal{E}$ of even pure states is an \emph{exact} state $k$-design if its $k$-th moment matches the Haar-averaged $k$-th moment, $\mathbb{E}_{\psi\sim\mathcal{E}}[\psi^{\otimes k}] = \mathbb{E}_{\phi\sim\mathrm{Haar}}[\phi^{\otimes k}]$, and an \emph{$\epsilon$-approximate} state $k$-design if the two ensembles are $\epsilon$-close in trace distance:
\begin{equation}
    D^{(k)}(\mathcal{E}) \;\coloneqq\; D\!\left(\,\mathbb{E}_{\psi\sim\mathcal{E}}\bigl[\psi^{\otimes k}\bigr],\;\mathbb{E}_{\phi\sim\mathrm{Haar}}\bigl[\phi^{\otimes k}\bigr]\,\right) \;\leq\; \epsilon.
\end{equation}
 Operationally, an $\epsilon$-approximate $k$-design is indistinguishable from a truly Haar-random ensemble by any quantum algorithm that makes $k$ queries to the state. Combining the bridge degree with a quantitative analysis of the Haar bridge spectral distribution yields:

\begin{thm_intro}[Non-Gaussian gate count for state $2$-designs; see Corollary~\ref{cor:exact_design} and Theorem~\ref{thm:approx_design}]
The number $t$ of non-Gaussian gates required for a state $k$-design for $k\geq2$ grows with system size:
\begin{itemize}
    \item \textbf{Exact $k$-designs} must contain at least one state requiring $t \geq \lfloor n/4 \rfloor = \Omega(n)$ non-Gaussian gates.
    \item \textbf{$\epsilon$-approximate $k$-designs} require $t = \Omega(\sqrt{n\log(1/\epsilon)})$ non-Gaussian gates with non-negligible probability over the ensemble (valid for $\epsilon \geq C_*/\sqrt n$, with $C_*$ an absolute constant).
\end{itemize}
\end{thm_intro}

This stands in sharp contrast to the qubit setting, where random Clifford circuits alone form an exact state $3$-design and approximate state $k$-designs are generated with a constant (independent of $n$) number of non-Clifford gates~\cite{leone2026nonclifford,haferkamp2022efficient,yuzhen2026designs}. On the matchgate side, a growing number of non-Gaussian gates is required already at the state $2$-design --- the first nontrivial design. The bound has an equivalent implication for a single state: any even pure state $\psi$ with $\nGMax(\psi)=o(\sqrt{n\log(1/\epsilon)})$ has its matchgate orbit at trace distance $>\epsilon$ from the Haar-random state.

\medskip

\textbf{(iv) Approximate bridge degree and approximate conversion (Sec.~\ref{sec:approx}).}
The bridge degree, being a discrete quantity, is fragile under perturbation: an arbitrarily small non-Gaussian perturbation of a Gaussian state can drive $\nGMax$ to its maximum value. We address this by introducing an \emph{$\epsilon$-approximate bridge degree}
\begin{equation}
    \Lambda_{\mathrm{d},\epsilon}(\psi) \;\coloneqq\; \min_{\phi:\, 1-F(\psi,\phi)\leq\epsilon}\nGMax(\phi),
\end{equation}
the minimum of $\nGMax$ over even pure states in a fidelity-$\epsilon$ ball around $\psi$. We show that $\Lambda_{\mathrm{d},\epsilon}$ inherits monotonicity from $\nGMax$:

\begin{thm_intro}[Monotonicity of the $\epsilon$-approximate bridge degree; see Theorem~\ref{thm:approx_monotone}]
For any even pure state $\psi$, any $\epsilon\in[0,1]$, and any Gaussian protocol $\mathcal{E}$ such that $\mathcal{E}(\psi)$ is pure,
\begin{equation}
    \Lambda_{\mathrm{d},\epsilon}\bigl(\mathcal{E}(\psi)\bigr) \;\leq\; \Lambda_{\mathrm{d},\epsilon}(\psi).
\end{equation}
\end{thm_intro}

The approximate bridge degree is not directly measurable --- its definition involves an optimization over the full set of even pure states, which is intractable. We circumvent this by introducing the \emph{bridge fidelity}
\begin{equation}
    F^\Lambda_k(\psi) \;=\; \sqrt{\textstyle\sum_{|\lambda|\leq k} p_\psi(\lambda)},
\end{equation}
the square root of the cumulative bridge spectral distribution at threshold $k$. The \emph{$\epsilon$-approximate bridge fidelity}
\begin{equation}
    \nGF_\epsilon(\psi) \;\coloneqq\; \min\{\alpha\geq 0 \,:\, F^\Lambda_{2\alpha}(\psi)\geq 1-\epsilon\}
\end{equation}
provides a computable lower bound on $\Lambda_{\mathrm{d},\epsilon}$ (Lemma~\ref{lem:relation_max_Lambda_fidelity}). As a corollary, this result yields a mixed-state non-Gaussianity witness based on fidelity: if an experimentally prepared state $\rho$ has fidelity at least $1-\epsilon$ with a known even pure target $\ket{\phi}$, then $\nGMax^{\,\mathrm{cr}}(\rho)$ is lower-bounded by $\nGF_\epsilon(\phi)$ (Corollary~\ref{cor:fidelity_witness}). Thus, a fidelity estimate with respect to a target state admitting a known classical description not only certifies the non-Gaussianity of the actual mixed laboratory state whenever $\nGF_\epsilon(\phi)>0$, but also quantitatively witnesses it by providing a lower bound on the mixed-state non-Gaussianity monotone $\nGMax^{\,\mathrm{cr}}$.

Furthermore, $\nGF_\epsilon(\psi)$ is efficiently estimable from Bell sampling using $\mathcal{O}(\Delta^{-2}(\log(1/\Delta)+\log(1/\delta)))$ samples, where $\Delta$ is the gap between the cumulative bridge spectral distribution and the threshold $(1-\epsilon)^2$ (Lemma~\ref{lem:sample_complexity_approx}, Algorithm~\ref{alg:algo_Lambda_fid}). Together with the monotonicity of $\Lambda_{\mathrm{d},\epsilon}$, this yields efficient bounds on approximate Gaussian conversion. In particular, applied to the canonical four-qubit magic state $\ket{M}=\frac{1}{2}\bigl(\ket{0000}+\ket{0101}+\ket{1010}+\ket{1111}\bigr)$ as a unit of non-Gaussian resource, they lower-bound the single-shot magic-state cost of approximately preparing an arbitrary non-Gaussian state:

\begin{thm_intro}[Magic-state cost lower bound; see Theorem~\ref{thm:lower_bound_cost}]
For any even pure state $\psi$ and any $\epsilon\geq 0$, the number of canonical magic states $\ket{M}$ required to prepare $\psi$ within trace distance $\epsilon$ by any post-selected Gaussian protocol satisfies
\begin{equation}
    \mathrm{Cost}^\epsilon(\psi) \;\geq\; \frac{\Lambda_{\mathrm{d},\epsilon}(\psi)}{4} \;\geq\; \frac{\nGF_\epsilon(\psi)}{4}.
\end{equation}
The lower bound $\nGF_\epsilon(\psi)/4$ is efficiently estimable from Bell-sampling data.
\end{thm_intro}

Since $\ket{M}$ is the resource state for the gadget construction of the SWAP gate~\cite{hebenstreit2019all}, Theorem~\ref{thm:lower_bound_cost} equivalently lower-bounds the minimum number of SWAP gates needed to prepare $\psi$ approximately. This result provides an experimentally certifiable lower bound on the non-Gaussian cost of state preparation in the approximate setting.

\medskip

\textbf{(v) Mixed-state extension via the Choi--Jamio{\l}kowski isomorphism (Sec.~\ref{sec:mixed}).}
 We extend the pure-state machinery of themes~(i), (ii), and (iv) to arbitrary even mixed states $\rho$ via the Choi--Jamio{\l}kowski isomorphism, in which an operator or density matrix is encoded in a pure quantum state on a doubled Hilbert space through its action on half of a maximally entangled Bell state $\kket{\rho} = (\rho\otimes I)\sum_i\ket{i}\otimes \ket{i}$. The bridge operator $\Lambda$ is lifted to a new operator
\begin{equation}
    \tilde{\Lambda} \;\coloneqq\; \Lambda \otimes I \;-\; I \otimes \Lambda
\end{equation}
on the doubled Choi space $\mathcal{H}_n^{\otimes 4}$, whose spectrum is again integer-valued. The resulting \emph{mixed-state bridge spectral distribution}
\begin{equation}
    \tilde p_\rho(\lambda) \;\coloneqq\; \bbra{\rho^{\otimes 2}}\tilde{\Pi}_\lambda\kket{\rho^{\otimes 2}}
\end{equation}
plays the same role as $p_\psi(\lambda)$ in the pure-state theory, and the \emph{mixed-state bridge degree}
\begin{equation}
    \tilde{\Lambda}_{\mathrm{d}}(\rho) \;\coloneqq\; \max\{\,\alpha\geq 0 \,:\, \tilde p_\rho(4\alpha)\neq 0\,\}
\end{equation}
is the natural extension of $\nGMax$ to mixed states. It inherits the full set of structural properties: it is integer-valued, faithful on the set of mixed FGSs (Lemma~\ref{lem:faithfulness_mixed}), additive under tensor products (Lemma~\ref{lem:additive_mixed}), and reduces to the pure-state bridge degree on pure inputs (Lemma~\ref{lem:pure_state_reduction}). Crucially, it remains a monotone under the (smaller) class of post-selected Gaussian operations (Definition~\ref{def:gaussian_operations}) --- those built from Gaussian unitaries, composition with Gaussian states, occupation-number measurements with post-selection, and partial tracing:

\begin{thm_intro}[Monotonicity of the mixed-state bridge degree; see Theorem~\ref{thm:mixed_max_Lambda_monotone}]
For any even state $\rho$ and any post-selected Gaussian operation $\mathcal{E}$ with $\mathcal{E}(\rho) = \sigma$,
\begin{equation}
    \tilde{\Lambda}_{\mathrm{d}}(\sigma) \;\leq\; \tilde{\Lambda}_{\mathrm{d}}(\rho).
\end{equation}
\end{thm_intro}

The Choi-based extension makes the application to Gaussian-conversion bounds, including in the approximate case, directly applicable to mixed states, with $\nGMax$ replaced by $\tilde{\Lambda}_{\mathrm{d}}$. Unlike the pure-state case, however, the mixed-state bridge spectral distribution is not directly accessible from copies of $\rho$ via Bell sampling, since the Choi state is not directly preparable from copies of $\rho$ alone. For the important special case in which mixedness arises from global depolarizing noise acting on an underlying pure state, we nevertheless show that error mitigation enables certification of the bridge degree of the ideal state (Proposition~\ref{prop:noise_robust_valid}, Algorithm~\ref{alg:noise_robust_maxlambda}).

\medskip

\textbf{(vi) Algorithmic primitives from Bell sampling (Secs.~\ref{sec:measurements} and~\ref{sec:designs_cert}).}
Many quantities introduced in the previous themes are functionals of $p_\psi$ and are therefore efficiently estimable from a single Bell-sampling primitive (Algorithm~\ref{alg:bell_sampling}). Building on this, we develop two concrete algorithmic primitives addressing central tasks in fermionic non-Gaussianity.

The first is a Gaussianity test, addressing the \emph{property-testing} question~\cite{buhrman2008quantum}: given black-box access to copies of an unknown state $\psi$, can one decide whether $\psi$ is Gaussian using as few samples as possible? The analogous question has been extensively studied across resource theories~\cite{gross2021schur,girardi2025gaussian,bittel2025optimal}. Bell sampling realizes this test with two-copy structure:

\begin{thm_intro}[Gaussianity testing; see Theorem~\ref{thm:gaussianity_test}]
There exists a Bell-sampling-based algorithm that distinguishes, with high probability:
\begin{itemize}
    \item any pure FGS, which is \textbf{always accepted} (i.e., the test has \emph{perfect completeness}); from
    \item any pure state with Gaussian fidelity at most $1-\epsilon$,
\end{itemize}
using $\mathcal{O}(n^2/\epsilon)$ copies. The protocol achieving this sample complexity consists of repeatedly performing Bell measurements on two copies of the state and accepting if and only if every observed Bell outcome corresponds to the $\Lambda$-eigenvalue $\lambda=0$.
\end{thm_intro}

For comparison with prior works, several fermionic Gaussianity tests have been proposed. The test of Ref.~\cite{bittel2025optimal} uses only single-copy measurements and has $\tilde{\mathcal{O}}(n^5/\epsilon^2)$ sample complexity, but does not satisfy perfect completeness. The test of Ref.~\cite{lyu2024fermionicgaussiantesting} satisfies perfect completeness, but requires three copies of the state at a time and has no explicit sample-complexity bound. The test of Ref.~\cite{walter2025randompurification} achieves $\mathcal{O}(n^2/\epsilon)$ sample complexity, but requires a collective measurement on an unbounded number of copies and does not satisfy perfect completeness. Our test, by contrast, satisfies perfect completeness, uses only two copies at a time, and has polynomial sample complexity $\mathcal{O}(n^2/\epsilon)$. It can thus be seen as the fermionic counterpart of the stabilizer property test of Ref.~\cite{gross2021schur} and the bosonic Gaussianity test of Ref.~\cite{girardi2025gaussian}, both of which share these properties. Moreover, it is \emph{optimal} among all two-copy tests with perfect completeness: it minimizes the false-acceptance probability on every non-Gaussian state.

We further show that the $n^2$ scaling above can be replaced by the number of non-Gaussian gates used to prepare the state:

\begin{thm_intro}[Gaussianity testing for $t$-doped Gaussian states; see Theorem~\ref{thm:test_tdoped}]
The same test distinguishes the two cases above for any state prepared by a $t$-doped Gaussian protocol, using $\mathcal{O}(t^2/\epsilon)$ copies --- a sample complexity independent of the system size $n$.
\end{thm_intro}

Since the non-Gaussian gates may be arbitrarily weak, this already covers states arbitrarily close to the Gaussian set, which should intuitively be the most difficult to distinguish from Gaussian. We are thus led to conjecture that a system-size-independent sample complexity holds for general states.

We also develop a \emph{tolerant} variant that distinguishes states with $F_\G(\psi)\geq 1-\alpha$ from states with $F_\G(\psi)\leq 1-\beta$ in the regime $\alpha = \mathcal{O}(\beta/n^2 - 2\alpha)$, using the same Bell-sampling primitive.

\begin{thm_intro}[Tolerant Gaussianity test; see Theorem~\ref{thm:gaussianity_test_tolerant}]
For any $0\leq \alpha < \beta\leq 1$ such that $\varepsilon \coloneqq \beta/n^2 - 2\alpha > 0$ and $\alpha = \mathcal{O}(\varepsilon)$, there is a Bell-sampling-based algorithm that distinguishes case (a) $F_\G(\psi)\geq 1-\alpha$ from case (b) $F_\G(\psi)\leq 1-\beta$ with success probability at least $1-\delta$ using $N = \mathcal{O}(n^2/(\beta-2\alpha n^2)\,\log(1/\delta))$ copies of $\psi$.
\end{thm_intro}

The second primitive is a Bell-sampling-based certifier for state $2$-designs. Complementing the gate-count lower bounds of theme (iii), we ask: given a candidate ensemble, can one verify from experimental data that it actually forms a state $2$-design? For matchgate-invariant ensembles --- those whose distribution is invariant under conjugation by every matchgate, including in particular every matchgate orbit of an underlying ensemble --- we show that the approximation error $D(\mathcal{E})\equiv D^{(2)}(\mathcal{E})$ reduces to a total variation distance between the ensemble-averaged and Haar bridge spectral distributions,
\begin{equation}
    p_{\mathcal{E}}(\lambda) \;\coloneqq\; \mathbb{E}_{\psi\sim\mathcal{E}}\bigl[p_\psi(\lambda)\bigr],
    \qquad
    D(\mathcal{E}) \;=\; \mathrm{TV}\!\bigl(p_{\mathcal{E}}, p_{\mathrm H}\bigr),
\end{equation}
and is directly accessible from \emph{ensemble Bell sampling} (Algorithm~\ref{alg:ensemble_bell_sampling}), which draws a fresh $\psi\sim\mathcal{E}$ at each round and runs Bell sampling on $\psi^{\otimes 2}$.

\begin{thm_intro}[Bell-sampling certification of state $2$-designs; see Lemma~\ref{lem:estimate_design_distance} and Theorems~\ref{thm:test_design_distance} and \ref{thm:nontolerant_test_design}]
For any matchgate-invariant ensemble $\mathcal{E}$:
\begin{itemize}
    \item \textbf{Approximate design testing.} Whether $D(\mathcal{E})\leq \alpha$ or $D(\mathcal{E})\geq \alpha+\epsilon$ is testable from $\widetilde{\mathcal{O}}(\sqrt n/\epsilon^2)$ Bell samples; in particular, $D(\mathcal{E})$ is estimable to additive accuracy $\epsilon$ at the same cost.
    \item \textbf{Exact design testing.} Whether $D(\mathcal{E})=0$ or $D(\mathcal{E})\geq \epsilon$ is testable from $\mathcal{O}(n^{1/4}/\epsilon^2)$ Bell samples.
\end{itemize}
\end{thm_intro}

This is striking: the corresponding decision problem for arbitrary ensembles is not known to admit an efficient quantum algorithm in general~\cite{nakata2025computational}, whereas matchgate invariance makes the $2$-design property efficiently testable from Bell sampling.

\subsection{Organization}
\label{sec:organization}

The remainder of the paper is organized as follows. Sec.~\ref{sec:preliminaries} sets notation and reviews the preliminary concepts that will be used throughout the paper. In Sec.~\ref{sec:eigenstructure}, we analyze the eigenvalue and eigenspace structure of the bridge operator $\Lambda$ --- a critical component of our approach. In Sec.~\ref{sec:maxlambda}, we introduce the bridge degree, establish its basic structural properties, and prove its monotonicity under post-selected Gaussian protocols. In Sec.~\ref{sec:consequences}, we develop the operational consequences of this monotonicity: no-go theorems for Gaussian conversion (Sec.~\ref{sec:nogo_conversion}), irreversibility of the resource theory of fermionic non-Gaussianity (Sec.~\ref{sec:irreversibility}), and lower bounds on the non-Gaussian gate count required to generate state designs in both the exact and approximate settings (Sec.~\ref{sec:designs_lower}). In Sec.~\ref{sec:approx}, we develop the $\epsilon$-approximate theory, including a computable and efficiently measurable lower bound. In Sec.~\ref{sec:mixed}, we extend the framework to mixed states via the Choi representation and introduce a mixed-state version of the bridge degree, with structural properties analogous to those of the pure-state version. In Sec.~\ref{sec:measurements}, we turn to algorithmic applications, presenting a sample-efficient Bell-sampling protocol for Gaussianity testing in both one-sided and tolerant forms. In Sec.~\ref{sec:designs_cert}, we present a sample-efficient Bell-sampling protocol to certify matchgate-invariant state $2$-designs, complementing the structural lower bounds of Sec.~\ref{sec:designs_lower}. We conclude in Sec.~\ref{sec:conclusion} with a discussion and open problems.

\section{Preliminaries}
\label{sec:preliminaries}

We consider a system of $n$ qubits with associated Hilbert space $\mathcal{H}_n \simeq \mathbb{C}^d$, where $d = 2^n$ is the dimension of the $n$-qubit space. The set of Hermitian operators on a vector space $\mathcal{V}$ is denoted by $\mathcal{B}(\mathcal{V})$. The set of quantum states is defined as
$\mathcal{D}(\mathcal{H}_n) \coloneqq \{ \rho \in \mathcal{B}(\mathcal{H}_n) \,|\, \rho \ge 0, \ \Tr(\rho) = 1 \}$.
We denote the set of integers from $1$ to $n$ by $[n] \coloneqq \{1, \dots, n\}$.

The trace distance between two quantum states $\rho, \sigma \in \mathcal{D}(\mathcal{H}_n)$ is defined as half the trace norm of their difference:
\begin{equation}
    D(\rho, \sigma) = \frac{1}{2} \left\| \rho - \sigma \right\|_1,
\end{equation}
where the trace norm of an operator $A$ is given by $\left\| A \right\|_1 \coloneqq \Tr \big( \sqrt{A^\dagger A} \big)$. The Hilbert--Schmidt inner product is defined as
\begin{equation}
    \langle A, B \rangle_{\mathrm{HS}} \coloneqq \Tr(A^\dagger B).
\end{equation}

\subsection{Majorana operators}

A system of $n$ qubits is equivalent to a system of $n$ fermionic modes via the Jordan--Wigner transformation. Throughout this paper, we use the qubit and fermionic representations interchangeably. In this mapping, the Majorana operators are defined in terms of Pauli operators as
\begin{align} \label{eq:majorana_def}
    \gamma_{2j-1} &\coloneqq \left(\prod_{k=1}^{j-1} Z_k\right) X_j, \\
    \gamma_{2j}   &\coloneqq \left(\prod_{k=1}^{j-1} Z_k\right) Y_j,
\end{align}
for $j \in [n]$.

The Majorana operators are Hermitian and traceless, and they satisfy the canonical anticommutation relations
\begin{equation}
    \{\gamma_j, \gamma_k\} = 2 \delta_{j,k} I,
\end{equation}
for all $j,k \in [2n]$.

Given an ordered subset $A \subset [2n]$, we define the associated Majorana monomial by
\begin{equation}
    \gamma_A \coloneqq i^{|A|(|A|-1)/2} \prod_{j \in A} \gamma_j,
\end{equation}
where $|A|$ denotes the cardinality of $A$, and the product is taken in increasing order. We also define $\gamma_{\varnothing} \coloneqq I$. The $4^n$ operators $\{\gamma_A\}$ are Hermitian and form an orthogonal basis of $\mathcal{B}(\mathcal{H}_n)$ with respect to the Hilbert--Schmidt inner product, referred to as the Majorana basis. Accordingly, any operator $O \in \mathcal{B}(\mathcal{H}_n)$ admits the expansion
\begin{equation}
    O = \sum_A O_A \gamma_A,
    \qquad
    O_A = 2^{-n} \Tr(O \gamma_A).
\end{equation}

An operator $O$ is said to be \emph{even} if $O_A = 0$ for all subsets $A$ with odd cardinality. This property is characterized by the parity operator
\begin{equation}
    P \coloneqq \prod_{j=1}^n Z_j,
\end{equation}
with respect to which an operator is even if and only if it commutes with $P$.

The parity operator induces a decomposition of the Hilbert space
$\mathcal{H}_n = \mathcal{H}_n^+ \oplus \mathcal{H}_n^-$,
where $\mathcal{H}_n^\pm \coloneqq \{ \ket{\psi} \in \mathcal{H}_n \mid P \ket{\psi} = \pm\ket{\psi} \}$. Namely, a pure state $\psi = \ket{\psi}\bra{\psi}$ is even if and only if $\ket{\psi} \in \mathcal{H}_n^+$ or $\ket{\psi} \in \mathcal{H}_n^-$.

\subsection{Fermionic Gaussian states and matchgates}

Let $O(2n)$ denote the group of real orthogonal $2n \times 2n$ matrices. For any $Q \in O(2n)$, a \emph{matchgate unitary}, or equivalently a \emph{fermionic Gaussian unitary}, $U_Q$ is defined through its action on the Majorana operators as
\begin{equation} \label{eq:action_matchgate}
    U_Q \gamma_i U_Q^\dagger = \sum_{j=1}^{2n} Q_{ij} \gamma_j,
\end{equation}
for all $i \in [2n]$. The unitary $U_Q$ is uniquely determined by $Q$ up to a global phase. Matchgates are generated by two-qubit $X_j X_{j+1}$ rotations, single-qubit $Z_j$ rotations, and the operator $X_n$ acting on the last qubit. The set of all matchgates on $n$ qubits forms a group, which we denote by $M_n$.

Fermionic Gaussian states are even states of the form
\begin{equation}
    \rho
    =
    U_Q
    \left(
        \bigotimes_{j=1}^n \frac{I + r_j Z_j}{2}
    \right)
    U_Q^\dagger,
\end{equation}
where $U_Q \in M_n$ and $\{r_j\}_{j=1}^{n}$ are real numbers which satisfy $|r_j| \leq 1$. A pure FGS corresponds to the case $r_j = \pm 1$ for all $j \in [n]$, and can always be written as $\ket{\psi} = U_Q \ket{0}^{\otimes n}$. Any mixed FGS arises as the reduced state of a pure FGS. We denote the set of FGSs as $\mathcal{G}_{\mathrm F}$.

Finally, we introduce the bridge operator $\Lambda$, defined in Ref.~\cite{bravyi2005lagrangian}, as
\begin{equation} \label{eq:Lambda_def}
    \Lambda \coloneqq \sum_{j=1}^{2n} \gamma_j \otimes \gamma_j.
\end{equation}
This operator plays a central role in characterizing FGSs. In particular, a pure state $\ket{\psi}$ is an FGS if and only if~\cite{dias2024classical,melo2013}
\begin{equation} \label{eq:Lambda_condition}
    \Lambda \ket{\psi}^{\otimes 2} = 0.
\end{equation}
More generally, a mixed state $\rho$ is an FGS if and only if~\cite{bravyi2005lagrangian}
\begin{equation} \label{eq:Lambda_condition_mixed}
    [\Lambda, \rho^{\otimes 2}] = 0.
\end{equation}
Another crucial property of $\Lambda$ is its invariance under matchgate unitaries: $[\Lambda, U^{\otimes 2}] = 0$ for any $U \in M_n$~\cite{bravyi2005lagrangian}.

\subsection{Resource theory of fermionic non-Gaussianity}\label{sec:resource_theory}

Quantum resource theories provide a systematic framework for quantifying the usefulness of quantum states for information-processing tasks~\cite{chitambar2019quantum,Gour2025}. A resource theory is specified by two ingredients: a set of \emph{free states} and a set of \emph{free operations}. Free states are those preparable without consuming the resource, while free operations are physical transformations that cannot generate the resource from free states. Prominent examples include the resource theories of entanglement~\cite{vidal2000entanglement}, coherence~\cite{baumgratz2014quantifying}, non-stabilizerness~\cite{veitch2014resource}, and bosonic non-Gaussianity~\cite{takagi2018convex}.

In this work we focus on the resource theory of \emph{fermionic non-Gaussianity}. The choice of free operations admits two natural variants, which we now make precise.

\begin{defn}[Gaussian operations]\label{def:gaussian_operations}
A Gaussian operation is any composition of:
\begin{itemize}[leftmargin=2em]
\item[(a)] applying a fermionic Gaussian (matchgate) unitary;
\item[(b)] composition with an FGS;
\item[(c)] taking a partial trace.
\end{itemize}
A post-selected Gaussian operation additionally allows
\begin{itemize}[leftmargin=2em]
\item[(d)] performing an occupation-number measurement with post-selection on a specific outcome.
\end{itemize}
\end{defn}

Gaussian operations are precisely the completely positive trace-preserving (CPTP) Gaussian channels~\cite{SpSc18,bravyi2005lagrangian}; they map FGSs to FGSs and have the FGSs themselves as free states. Post-selection of (d) also maps FGSs to FGSs, so the set of free states is unchanged. Crucially, however, the set of FGSs is \emph{not} convex --- a probabilistic mixture of distinct FGSs is generally not Gaussian --- so the corresponding resource theory falls outside the standard convex framework~\cite{chitambar2019quantum,Gour2025}.

Augmenting Gaussian operations with operations conditioned on classical randomness or measurement outcomes yields the \emph{protocol} versions~\cite{tarabunga2026}.

\begin{defn}[Gaussian protocols]\label{def:gaussian_protocol}
A Gaussian protocol is any composition of (a)--(c) of Definition~\ref{def:gaussian_operations} together with:
\begin{itemize}[leftmargin=2em]
\item[(d)] performing an occupation-number measurement;
\item[(e)] applying any of the above operations conditioned on classical randomness or measurement outcomes.
\end{itemize}
A post-selected Gaussian protocol additionally allows post-selection on the outcome in (d).
\end{defn}

Unlike Gaussian operations, Gaussian protocols allow classical mixing and can therefore produce probabilistic mixtures of FGSs. Such mixtures are generally non-Gaussian, since the set of FGSs is not convex~\cite{SpSc18}. This places fermionic non-Gaussianity within the convex resource-theory framework~\cite{tarabunga2026,chitambar2019quantum}, whose free states are the convex hull of the FGSs,
\begin{equation}\label{eq:convex_fgs}
    \mathrm{conv}(\mathcal{G}_{\mathrm F}) \;=\; \left\{\,\int\!d\mu(\nu)\, \sigma_\nu \,\middle|\, \sigma_\nu\in\mathcal{G}_{\mathrm F},\, \mu \text{ a probability measure}\right\},
\end{equation}
the so-called \emph{convex-Gaussian states}~\cite{melo2013}. Equivalently, $\mathrm{conv}(\mathcal{G}_{\mathrm F})$ is the convex hull of the pure FGSs, since any state $\sigma_\nu\in\mathcal{G}_{\mathrm F}$ is itself a convex combination of pure FGSs. These states can be simulated efficiently~\cite{melo2013}, as long as
one can efficiently sample from the probability measure $d\mu$. The resulting convex resource theory of fermionic non-Gaussianity parallels the convex resource theories of entanglement~\cite{vidal2000entanglement}, non-stabilizerness~\cite{veitch2014resource}, and bosonic non-Gaussianity~\cite{takagi2018convex,albarelli2018}. Like Clifford circuits, Gaussian protocols become universal upon injection of a canonical magic state~\cite{hebenstreit2019all}, making them the natural class for studying the resource cost of universal fermionic quantum computation.

\subsection{Fermionic non-Gaussianity monotones}
\label{sec:nongaussianity_monotones}

Given a function $M$ defined on the set of states, we say that $M$ is a fermionic non-Gaussianity monotone if it is \emph{non-increasing} under any Gaussian protocol $\mathcal{E}$,
\begin{equation}
	\label{eq:monotonicity_condition}
M\left[\mathcal{E}(\rho)\right]\leq M(\rho)\,.
\end{equation} 
One important use of a resource monotone is to constrain \emph{state conversion} under the free operations --- here, Gaussian protocols. Since monotonicity implies that a conversion $\rho\to\sigma$ is impossible whenever $M(\sigma)>M(\rho)$, the monotone immediately yields no-go theorems for Gaussian state conversion. A central question in a resource theory is how many copies of a target state $\sigma$ can be produced per copy of a resource state $\rho$ under the free operations. In this context, a key property is \emph{additivity} under tensor products, $M(\rho\otimes\sigma)=M(\rho)+M(\sigma)$, which turns a monotone into a quantitative tool for analyzing resource conversion: any additive non-Gaussianity monotone yields a bound on the conversion rate~\cite{chitambar2019quantum,Gour2025},
\begin{equation}\label{eq:conversion_rate_bound}
    r(\rho \to \sigma) \leq \frac{M(\rho)}{M(\sigma)},
\end{equation}
where $r(\rho\to\sigma)$ denotes the largest ratio $m/n$ at which $\rho^{\otimes n}$ can be converted into $\sigma^{\otimes m}$ using Gaussian protocol. Additivity is thus essential in reducing the many-copy problem to single-copy values, avoiding the difficult evaluation of $M$ on many copies of a state. If, in addition, $M$ is non-increasing under the larger class of post-selected Gaussian protocols, these no-go theorems and conversion bounds extend to \emph{probabilistic} Gaussian conversions, i.e., those succeeding with any nonzero probability. %

We now survey some known fermionic non-Gaussianity monotones and review some of their properties. Throughout, we focus on monotones defined on pure states $\ket{\psi}$, as this is the main focus of our work.

The \emph{Gaussian fidelity}~\cite{dias2024classical, cudby2023gaussian} is
\begin{equation}
\label{eq:gfid_def}
    F_{\G}(\ket{\psi})
    = \max_{\ket{\phi} \in \mathcal{G}_{\mathrm F}}
    |\langle\phi|\psi\rangle|^2 ,
\end{equation}
the maximal overlap between $\ket{\psi}$ and any pure FGS. The quantity $-\log F_{\G}$ is a monotone under pure-state Gaussian protocols and sub-additive under tensor products, becoming additive for tensor products of $4$-qubit GHZ states~\cite{reardonSmith2025extentmultiplicative}.

The \emph{Gaussian rank} is~\cite{dias2024classical, cudby2023gaussian}
\begin{equation}
    \label{eq:grank_def}
    \chi_{\G}(\ket{\psi})
    = \min\!\big\{ \chi\in\mathbb{N} \,\big|\,
    \exists\,\{\ket{\phi_j}\}_{j=1}^\chi\!\subset \mathcal{G}_{\mathrm F},
    \; x_j\in \mathbb{C}
    \text{ s.t. }
    \ket{\psi}=\textstyle\sum_{j=1}^\chi x_j \ket{\phi_j} \big\},
\end{equation}
the minimal number of (not necessarily orthogonal) pure FGSs needed to express $\ket{\psi}$ as a superposition; it upper-bounds the cost of exact classical simulation. Closely related is the \emph{Gaussian extent}~\cite{dias2024classical, cudby2023gaussian},
\begin{equation}
    \label{eq:gextent_def}
    \xi_\G(\ket{\psi})
    = \inf_{M\in \mathbb{N}} \,
    \inf_{\{\ket{\phi_j}\}\subset\mathcal{G}_{\mathrm F}}
    \Big\{
    \lVert x \rVert_1^2
    \,\Big|\,
    \ket{\psi} = \textstyle\sum_{j=1}^M x_j \ket{\phi_j}
    \Big\},
\end{equation}
where $\| x \|_1 = \sum_{j} |x_j|$; it is the squared $\ell^1$-norm of the coefficients, minimized over all finite FGS superpositions of $\ket{\psi}$; it upper-bounds the classical complexity of approximately simulating near-Gaussian fermionic circuits~\cite{dias2024classical, reardonSmith2024improved}. Both the Gaussian rank and the Gaussian extent are monotones under Gaussian protocols, and the Gaussian rank is additionally a monotone under post-selected Gaussian protocols. They are sub-additive under tensor products, and the Gaussian extent becomes additive for tensor products of $4$-qubit states~\cite{reardonSmith2025extentmultiplicative}.

The \emph{relative entropy of fermionic non-Gaussianity}~\cite{tarabunga2026,gottlieb2006properties,gottlieb2014correlations,lyu2024fermionicgaussiantesting} reads, for pure states,
\begin{equation} \label{eq:ngr_def}
    \nGR(\ket{\psi}) = S\!\left(\mathcal{G}(\psi)\right),
\end{equation}
where $\mathcal{G}(\psi)$ is the unique FGS with the same covariance matrix as $\ket{\psi}$ and $S$ is the von Neumann entropy. Equivalently, it is the minimum quantum relative entropy from $\ket{\psi}$ to the set of FGSs~\cite{lyu2024fermionicgaussiantesting}. It is a monotone under pure-state Gaussian protocols~\cite{tarabunga2026} and additive under tensor products.

Finally, the \emph{second moment of $\Lambda$}~\cite{tarabunga2026,coffman2025measuringnongaussianmagicfermions,sierant2025fermionicmagicresourcesquantum} is
\begin{equation} \label{eq:m_Lambda_prelim}
    \nGLambda(\psi) \;\coloneqq\; \tfrac{1}{2}\, \bra{\psi}^{\otimes 2} \Lambda^2 \ket{\psi}^{\otimes 2},
\end{equation}
which is a special case of occupation number entropy~\cite{tarabunga2026} and fermionic anti-flatness~\cite{sierant2025fermionicmagicresourcesquantum}. It is likewise a monotone under pure-state Gaussian protocols~\cite{tarabunga2026} and additive under tensor products. %

 We summarize these monotones and their properties in Tab.~\ref{tab:known_monotones}. The bridge degree introduced in this work stands out as the only quantity that is simultaneously monotone under post-selection, additive, and computable from the state description.

\begin{table*}[t]
\centering
\renewcommand{\arraystretch}{1.2}
\begin{tabular}{@{}l@{\hspace{2em}}l@{\hspace{1.5em}}ccc@{}}
\toprule
\textbf{Monotone} & \textbf{Symbol} & \textbf{Post-sel.\ monotone} & \textbf{Additive} & \textbf{Computable} \\
\midrule
Gaussian fidelity & $-\log F_\G$ & \ding{55} & \ding{55} & \ding{55} \\
Gaussian rank & $\log\chi_\G$ & \checkmark & \ding{55} & \ding{55} \\
Gaussian extent & $\log\xi_\G$ & \ding{55} & \ding{55}$^{\dagger}$ & \ding{55} \\
Relative entropy of non-Gaussianity & $\nGR$ & \ding{55} & \checkmark & \checkmark \\
Second moment of $\Lambda$ & $\nGLambda$ & \ding{55} & \checkmark & \checkmark \\
\midrule
Bridge degree (this work) & $\nGMax$ & \checkmark & \checkmark & \checkmark \\
\bottomrule
\end{tabular}
\caption{Comparison of fermionic non-Gaussianity monotones for pure states. The ``Post-sel.\ monotone'' column marks (\checkmark) monotones provably non-increasing under post-selected Gaussian protocols. The ``Additive'' column marks (\checkmark) monotones additive under tensor products. The ``Computable'' column marks (\checkmark) quantities accessible from the state description without optimization, and (\ding{55}) those whose definition requires optimization not known to be efficient. The bridge degree is the only monotone that is simultaneously monotone under post-selection, additive, and computable. $^{\dagger}$Additive for tensor products of $4$-qubit states~\cite{reardonSmith2025extentmultiplicative}.}
\label{tab:known_monotones}
\end{table*}

\subsection{Haar measure and state designs}

Let $U_P(2^n)$ denote the group of parity-preserving unitary operators acting on $n$ qubits, and let
$\mathcal{U}_n \subseteq U_P(2^n)$ be a (continuous or discrete) subgroup. The Haar measure on
$\mathcal{U}_n$ is the unique probability measure that is both left- and right-invariant under the group action.
For a positive integer $k$, we define
the $k$-fold twirling channel over $\mathcal{U}_n$ by
\begin{equation}
    \mathcal{T}_{\mathcal{U}_n}^{(k)}(O)
    \coloneqq
    \int_{\mathcal{U}_n} d\mu_H(U)\,
    U^{\otimes k}  O (U^\dagger)^{\otimes k},
\end{equation}
for any operator $O \in \mathcal{B}(\mathcal{H}_n^{\otimes k})$. 

The $k$-th order commutant of $\mathcal{U}_n$ consists of all operators
$O \in \mathcal{B}(\mathcal{H}_n^{\otimes k})$ that commute with $U^{\otimes k}$ for every
$U \in \mathcal{U}_n$:
\begin{equation}
    \mathrm{Com}_k(\mathcal{U}_n)
    \coloneqq
    \left\{
        O \in \mathcal{B}(\mathcal{H}_n^{\otimes k})
        \,\middle|\,
        (U^\dagger)^{\otimes k} O U^{\otimes k} = O,
        \ \forall U \in \mathcal{U}_n
    \right\}.
\end{equation}
This object is central to the analysis of twirling channels over a unitary group, since the twirling channel acts as the orthogonal projector onto the commutant with respect to the Hilbert--Schmidt inner product~\cite{mele2024haar,kliesch2021}. Concretely, if $\{ P_1, \dots, P_{\dim(\mathrm{Com})} \}$ is an orthonormal basis of $\mathrm{Com}_k$, then for any
    $O \in \mathcal{B}(\mathcal{H}_n^{\otimes k})$,
    \begin{equation} \label{eq:twirling_project_commutant}
        \mathcal{T}_{\mathcal{U}_n}^{(k)}(O)
        =
        \sum_{j=1}^{\dim(\mathrm{Com})}
        \Tr(P_j O)\, P_j .
    \end{equation}
This provides a practical recipe for computing the action of the twirling channel associated with a given unitary group.

We also introduce the notion of \emph{state $k$-designs}, also known as \emph{spherical $k$-designs}. Throughout this work, all notions of state design are framed on the fermionic Hilbert space, i.e. restricted to even states, unless stated otherwise. Given an ensemble of even pure states equipped with a probability measure, define the $k$-th moment
\begin{equation}
    \Psi_{\mathcal{E}}^{(k)} \;\coloneqq\; \mathbb{E}_{\psi\sim\mathcal{E}}\bigl[\psi^{\otimes k}\bigr].
\end{equation}
The ensemble $\mathcal{E}=\{\psi\}$ is a \emph{state $k$-design} if
\begin{equation}\label{eq:design_def}
    \Psi_{\mathcal{E}}^{(k)} \;=\; \Psi_{\mathrm{Haar}}^{(k)},
\end{equation}
where $\Psi_{\mathrm{Haar}}^{(k)} = \mathcal{T}^{(k)}_{U_P(n)}(\psi)$ is the $k$-th Haar moment (independent of the reference even pure state $\psi$). It is an \emph{$\epsilon$-approximate state $k$-design} in trace distance if
\begin{equation}\label{eq:design_def_approx}
    D^{(k)}(\mathcal{E}) \;\coloneqq\; D\!\left(\Psi_{\mathcal{E}}^{(k)},\; \Psi_{\mathrm{Haar}}^{(k)}\right) \;\leq\; \epsilon.
\end{equation}
 An $\epsilon$-approximate $k$-design in trace distance carries a clean operational meaning via the Holevo--Helstrom theorem: it cannot be distinguished from a Haar-random state with success probability exceeding $\tfrac{1+\epsilon}{2}$ by any quantum experiment that makes at most $k$ queries to the state. In what follows, for $k=2$ we will write $D(\mathcal{E}) \coloneqq D^{(2)}(\mathcal{E})$.

\subsection{Pauli operators and Bell sampling}

We define the single-qubit Pauli operators as
$\sigma_{r^z r^x} \coloneqq i^{r^z r^x} X^{r^x} Z^{r^z}$,
where $r^z, r^x \in \{0,1\}$. The Bell states are defined by
$\ket{\sigma_{r^z r^x}} \coloneqq i^{r^z r^x} (\sigma_{r^z r^x} \otimes I)\ket{\Phi^+}$,
where $\ket{\Phi^+} \coloneqq \frac{1}{\sqrt{2}}(\ket{00} + \ket{11})$. Explicitly, the four Bell states are
\begin{align} \label{eq:bell_states}
    \ket{\sigma_{00}} &= \tfrac{1}{\sqrt{2}}(\ket{00} + \ket{11}), &
    \ket{\sigma_{01}} &= \tfrac{1}{\sqrt{2}}(\ket{10} + \ket{01}), \\
    \ket{\sigma_{10}} &= \tfrac{1}{\sqrt{2}}(\ket{00} - \ket{11}), &
    \ket{\sigma_{11}} &= \tfrac{1}{\sqrt{2}}(\ket{01} - \ket{10}).
\end{align}
For an $n$-qubit system, we define the tensor product of Pauli operators by
$\sigma_{\mathbf{r}} \coloneqq \bigotimes_{j=1}^n \sigma_{r^z_j r^x_j}$
and the tensor-product of Bell states by
$\ket{\sigma_{\boldsymbol{r}^z \boldsymbol{r}^x}} \coloneqq \bigotimes_{j=1}^n \ket{\sigma_{r^z_j r^x_j}}$.

For two quantum states $\rho_A$ and $\rho_B$, a measurement in the Bell basis can be implemented by preparing the tensor-product state $\rho_A \otimes \rho_B$, applying a unitary $U_{\mathrm{Bell}}$ that maps the Bell basis to the computational basis, and measuring in the computational basis. As an important special case, performing a Bell measurement on two identical copies of a pure state $\ket{\psi}$ produces an outcome
$\mathbf{r}^z, \mathbf{r}^x \in \{0,1\}^n$ with probability
\begin{equation} \label{eq:pauli_dist}
    P(\mathbf{r})
    = \bra{\psi}\bra{\psi} O_{\mathbf{r}} \ket{\psi}\ket{\psi} =
    \frac{1}{2^n}
    \left|
        \bra{\psi} \sigma_{\mathbf{r}} \ket{\psi^{*}}
    \right|^2,
\end{equation}
where $O_{\mathbf{r}}$ is the projector onto a tensor product of Bell states and $\ket{\psi^{*}}$ is the complex conjugate of $\ket{\psi}$.

The unitary $U_{\mathrm{Bell}}$ admits two physically distinct implementations. On qubit-encoded platforms, the standard implementation is the depth-$1$ Clifford circuit
\begin{equation}\label{eq:U_Bell_Clifford}
    U_{\mathrm{Bell}} \;=\; \bigotimes_{i=1}^n (\mathrm{H} \otimes I_2)\,\mathrm{CNOT}_{i,\,n+i},
\end{equation}
where $\mathrm{H}$ is the Hadamard gate and $\mathrm{CNOT}_{i,\,n+i}$ is the controlled-NOT gate with qubit $i$ (from the first copy) as control and qubit $n+i$ (from the second copy) as target. On fermionic platforms, $U_{\mathrm{Bell}}$ can be realized as a depth-$1$ matchgate circuit consisting of a layer of fermionic beam splitters between corresponding modes of the two copies,
\begin{equation}\label{eq:U_Bell_matchgate}
    U_{\mathrm{Bell}} \;=\; \exp\!\Bigl(\,\tfrac{\pi}{4}\!\sum_{j=1}^{n} \gamma_{2j}\,\gamma_{2n+2j-1}\,\Bigr).
\end{equation}
Note that this operation is also a Clifford operation, corresponding to swapping the modes $\gamma_{2j}$ and $\gamma_{2n+2j-1}$. This second implementation is more natural in fermionic platforms, where matchgate operations (fermionic linear optics) are the native free operations.  In either case, the operational cost is the same --- two copies of the state and a constant-depth Gaussian-or-Clifford circuit --- and the outcome distribution is identical~\footnote{For the matchgate implementation, the Bell states are obtained up to phases.}, given by Eq.~\eqref{eq:pauli_dist}. We refer to either implementation as \emph{Bell sampling} throughout this work, with the understanding that the choice of implementation is hardware-dependent but in both cases efficient.

\section{Eigenstructure of the bridge operator}
\label{sec:eigenstructure}

The bridge operator $\Lambda$ is the central object of this work. In this section, we analyze its spectral structure in detail and identify the constraints that the physics of fermionic states imposes on it. First, we identify the spectrum of $\Lambda$ and construct an explicit eigenbasis in terms of tensor products of Bell states; this is the key fact that makes Bell sampling the natural measurement protocol for characterizing fermionic non-Gaussianity. Then, we use the exchange symmetry of $\ket{\psi}^{\otimes 2}$ and the parity symmetry of even states to narrow down the eigenvalue sectors that can actually be occupied. Finally, using that the spectral projectors of $\Lambda$ form an orthonormal basis for the second-order matchgate commutant, we introduce a matchgate-invariant probability distribution that records the weight of the two-copy state $\ket{\psi}^{\otimes 2}$ on each eigenspace of $\Lambda$. This distribution is the foundational object underlying the non-Gaussianity quantifiers developed in subsequent sections.

\subsection{Spectrum and Bell-state eigenbasis}
\label{sec:spectrum_bell}

Since $\Lambda$ is a sum of $2n$ mutually commuting Hermitian operators $\gamma_j \otimes \gamma_j$, each with eigenvalues $\mu_j \in \{\pm 1\}$, its spectrum is $\lambda = \sum_{j=1}^{2n} \mu_j \in \{-2n,-2n+2,\ldots,2n\}$. We denote by $V_\lambda$ the eigenspace with eigenvalue $\lambda$ and by $\Pi_\lambda$ the corresponding orthogonal projector. The dimension of $V_\lambda$ is $\binom{2n}{n + \lambda/2}$.

For later convenience, we introduce the binary variables
\begin{equation} \label{eq:s_def}
    s_j \coloneqq 
    \begin{cases}
        (1-\mu_j)/2, & \text{if $j$ is odd}, \\
        (1+\mu_j)/2, & \text{if $j$ is even},
    \end{cases}
\end{equation}
so that $s_j \in \{0,1\}$ for all $j$.

The following lemma shows that the eigenspaces of $\Lambda$ admit an explicit basis in terms of tensor products of Bell states. 

\begin{lem}[Bell-state eigenbasis] \label{lem:eigenbasis_bell}
Each eigenspace $V_\lambda$ of $\Lambda$ is spanned by tensor products of Bell states of the form
$\ket{\Psi_{\boldsymbol{\mu}}} \coloneqq \ket{\sigma_{\boldsymbol{r}^z \boldsymbol{r}^x}}$,
where $\boldsymbol{\mu} \in \{\pm 1\}^{2n}$ satisfies $\sum_{j=1}^{2n} \mu_j = \lambda$, and $\boldsymbol{r}^z, \boldsymbol{r}^x \in \{0,1\}^{n}$ are defined by
\begin{equation} \label{eq:r_func}
    r^z_j = \left( \bigoplus_{k=1}^{2j-2} s_k \right) \oplus s_{2j-1}, 
    \qquad
    r^x_j = s_{2j-1} \oplus s_{2j},
\end{equation}
with $\{s_j\}_{j=1}^{2n}$ defined in Eq.~\eqref{eq:s_def}. In particular, $\ket{\Psi_{\boldsymbol{\mu}}}$ is a simultaneous eigenstate of all operators $\gamma_j \otimes \gamma_j$ with eigenvalues $\mu_j$.
\end{lem}

\begin{proof}

Using the definition of the Majorana operators
[cf.~Eq.~\eqref{eq:majorana_def}], the operators
$\gamma_{2j-1}\otimes\gamma_{2j-1}$ and
$\gamma_{2j}\otimes\gamma_{2j}$ can be written as tensor products
of operators $Z\otimes Z$, $X\otimes X$, and $Y\otimes Y$ acting on
the first $j$ pairs of qubits. Since $\ket{\Psi_{\boldsymbol{\mu}}}$ is a tensor
product of Bell states, the action of these operators can be evaluated factor by factor using the identities
\begin{equation} \label{eq:pauli_eig_bell}
\begin{aligned}
    Z \otimes Z \ket{\sigma_{r^z r^x}} &= (-1)^{r^x} \ket{\sigma_{r^z r^x}}, \\
    X \otimes X \ket{\sigma_{r^z r^x}} &= (-1)^{r^z} \ket{\sigma_{r^z r^x}}, \\
    Y \otimes Y \ket{\sigma_{r^z r^x}} &= (-1)^{r^z + r^x + 1} \ket{\sigma_{r^z r^x}}.
\end{aligned}
\end{equation}
Hence, we verify that $\ket{\Psi_{\boldsymbol{\mu}}}$ is an eigenstate of $\gamma_{2j-1} \otimes \gamma_{2j-1}$ with eigenvalue
\begin{equation}
    \prod_{k=1}^{j-1} (-1)^{r^x_k} \, (-1)^{r^z_j}
    = (-1)^{s_{2j-1}} = \mu_{2j-1}.
\end{equation}
Similarly, it is an eigenstate of $\gamma_{2j} \otimes \gamma_{2j}$ with eigenvalue
\begin{equation}
    \prod_{k=1}^{j-1} (-1)^{r^x_k} \, (-1)^{r^z_j + r^x_j + 1}
    = (-1)^{s_{2j}+1} = \mu_{2j}.
\end{equation}
This shows that $\ket{\Psi_{\boldsymbol{\mu}}}$ is a simultaneous eigenstate of
$\gamma_j \otimes \gamma_j$ with eigenvalue $\mu_j$ for all $j$, and thus an eigenstate of $\Lambda$ with eigenvalue
$\lambda = \sum_j \mu_j$.
Finally, since $\{\ket{\Psi_{\boldsymbol{\mu}}}\}$ forms an orthonormal basis of $\mathcal{H}_n^{\otimes 2}$, the subset with fixed $\lambda$ spans $V_\lambda$.
\end{proof}

Lemma~\ref{lem:eigenbasis_bell} is the foundation of our measurement protocol: it establishes an explicit dictionary between the eigenvalue $\lambda$ of $\Lambda$ and the bitstring $(\boldsymbol{r}^z,\boldsymbol{r}^x)$ produced by a Bell measurement on two copies of $\ket{\psi}$. Inverting Eq.~\eqref{eq:r_func} gives
\begin{equation} \label{eq:inverse_map}
    s_{2j-1} = \Bigl(\bigoplus_{k=1}^{j-1} r^x_k\Bigr) \oplus r^z_j,
    \qquad
    s_{2j}   = \Bigl(\bigoplus_{k=1}^{j} r^x_k\Bigr) \oplus r^z_j,
\end{equation}
from which the eigenvalue is recovered as
\begin{equation}\label{eq:lambda_from_s}
    \lambda \;=\; 2\sum_{j=1}^{2n} (-1)^{j}s_j.
\end{equation}
Consequently, the probability of observing $\lambda$ in a single Bell sample on two copies of $\psi$ is
\begin{equation}
    p_\psi(\lambda) \;\coloneqq\; \bra{\psi}^{\otimes 2}\,\Pi_\lambda\,\ket{\psi}^{\otimes 2}.
\end{equation}
This distribution $p_\psi$ is the foundational object of the framework, formally introduced and analyzed in Sec.~\ref{sec:bell_weights} as the \emph{bridge spectral distribution}. The full Bell-sampling algorithm, which draws $N$ i.i.d.\ samples from $p_\psi$, is summarized in Algorithm~\ref{alg:bell_sampling}.

\begin{algorithm}[H]
\caption{Bell sampling for bridge spectral distribution}
\label{alg:bell_sampling}
\begin{flushleft}
\textbf{Input:} $N\geq 1$, and access to $2N$ identical copies of an even pure state $\psi$ on $n$ qubits. \\
\textbf{Output:} eigenvalues $(\lambda_1,\dots,\lambda_N)$ drawn i.i.d.\ from $p_\psi$. The empirical bridge spectral distribution is $\widehat p_\psi(\lambda) \coloneqq N^{-1}\sum_{k=1}^N \mathbf{1}\{\lambda_k=\lambda\}$.
\end{flushleft}
\begin{algorithmic}[1]
\For{$k=1,\dots,N$}
    \State Prepare two copies of $\psi$ on registers $A,B$ of $n$ qubits each.
    \State Apply $U_{\mathrm{Bell}}$ (Eq.~\eqref{eq:U_Bell_Clifford} or Eq.~\eqref{eq:U_Bell_matchgate}) across corresponding qubits of $A$ and $B$.
    \State Measure all $2n$ qubits in the computational basis, obtaining $(\boldsymbol{r}^z,\boldsymbol{r}^x)\in\{0,1\}^n\times\{0,1\}^n$.
    \State Compute $\boldsymbol{s}\in\{0,1\}^{2n}$ from $(\boldsymbol{r}^z,\boldsymbol{r}^x)$ via the inverse map of Eq.~\eqref{eq:inverse_map}.
    \State Set $\lambda_k \gets 2\sum_{j=1}^{2n} (-1)^{j}s_j$.
\EndFor
\State \Return $(\lambda_1,\dots,\lambda_N)$.
\end{algorithmic}
\end{algorithm}

\subsection{Symmetry constraints from the two-copy structure}
\label{sec:symmetry_constraints}

We now identify which eigenvalues of $\Lambda$ can actually appear when $\ket{\psi}$ is an even pure state. There are two independent constraints at work, both arising from symmetries of $\ket{\psi}^{\otimes 2}$.

The first constraint comes from \emph{exchange symmetry}. For any pure state $\ket{\psi}$, the two-copy state $\ket{\psi}^{\otimes 2}$ lies in the symmetric subspace $\mathrm{Sym}_2(\mathcal{H}_n) = \{\ket{v}\in\mathcal{H}_n^{\otimes 2} : \mathbb{S}\ket{v}=\ket{v}\}$ under the swap operator $\mathbb{S} = \bigotimes_{j=1}^n \mathbb{S}_j$, where $\mathbb{S}_j$ exchanges the two copies of qubit $j$. Since $\Lambda$ commutes with $\mathbb{S}$ (each term $\gamma_j\otimes\gamma_j$ is swap-invariant), each eigenspace $V_\lambda$ decomposes into symmetric and antisymmetric components. The next lemma shows that, in fact, each eigenspace is contained \emph{entirely} in one or the other, with membership determined by $\lambda \bmod 8$.

\begin{lem}[Exchange-symmetry constraint] \label{lem:symmetry_eigenspace_Lambda}
For each integer $k$, the eigenspaces of $\Lambda$ satisfy:
$V_{\lambda} \subseteq \mathrm{Sym}_2(\mathcal{H}_n)$ for $\lambda \equiv 0,2 \pmod{8}$, and $V_{\lambda} \subseteq \mathrm{Asym}_2(\mathcal{H}_n)$ for $\lambda \equiv 4,6 \pmod{8}$.
\end{lem}
\begin{proof}
Consider the eigenbasis $\{\ket{\Psi_{\boldsymbol{\mu}}}\}$ of Lemma~\ref{lem:eigenbasis_bell}.
A Bell state $\ket{\sigma_{r^z r^x}}$ is antisymmetric under copy exchange if and only if $(r^z,r^x)=(1,1)$. Hence $\ket{\Psi_{\boldsymbol{\mu}}}$ is symmetric if and only if the number of $\ket{\sigma_{11}}$ factors is even.

Only the two Bell states 
$\ket{\sigma_{01}}$ and $\ket{\sigma_{11}}$ give nontrivial contributions to $\lambda=\sum_{j}\mu_j$, whereas 
$\ket{\sigma_{00}}$ and $\ket{\sigma_{10}}$ contribute zero at each site. We encode each eigenstate by a binary sequence $x_1,\dots,x_m$ where $x_j=0$ for $\ket{\sigma_{01}}$ and $x_j=1$ for $\ket{\sigma_{11}}$, listed in the order in which these Bell states appear among the $n$ sites. Let $m_0,m_1$ be the respective counts, with $m=m_0+m_1$. Define the alternating sum
\begin{equation}\label{eq:defT}
    T \coloneqq \sum_{j=1}^m (-1)^{j+1} x_j,
\end{equation}
which satisfies $T \equiv m_1 \pmod{2}$.

We now group the sequence into runs of equal bits, that is, consecutive sequences of (at least) one $0$ and $1$, respectively. A run of bit $b\in\{0,1\}$ starting at position $s$ with length $\ell$ contributes to $\lambda$ if and only if $\ell$ is odd, in which case the contribution is $c=2(-1)^{s+b+1}$. Let $A$ be the sum of $(-1)^{s+1}$ over all odd-length $0$-runs, and $B$ the sum of $(-1)^{s+1}$ over all odd-length $1$-runs. Then $\lambda = 2\alpha$ with $\alpha = A-B$.

Grouping the sum in Eq.~\eqref{eq:defT} over runs shows that only odd-length $1$-runs contribute, giving $T=B$. Similarly, $A =\sum_{j=1}^m (-1)^{j+1}(1-x_j) = \frac{1-(-1)^m}{2} - T$. Therefore
\begin{equation}\label{eq:alphaT}
    \alpha = A - B = \frac{1-(-1)^m}{2} - 2T.
\end{equation}

We can now determine $\alpha \bmod 4$ from the parities of $m_0$ and $m_1$:
\begin{itemize}
\item $m_1$ even, $m_0$ even: then $m$ even and $T$ even, so $\alpha\equiv 0 \pmod{4}$;
\item $m_1$ even, $m_0$ odd: then $m$ odd and $T$ even, so $\alpha\equiv 1 \pmod{4}$;
\item $m_1$ odd, $m_0$ even: then $m$ odd and $T$ odd, so $\alpha\equiv 3 \pmod{4}$;
\item $m_1$ odd, $m_0$ odd: then $m$ even and $T$ odd, so $\alpha\equiv 2 \pmod{4}$.
\end{itemize}
Since $\lambda=2\alpha$ and $\ket{\Psi_{\boldsymbol{\mu}}}$ is symmetric if and only if $m_1$ is even, the claim immediately follows.
\end{proof}

The consequence for Bell sampling is immediate: since $\ket{\psi}^{\otimes 2}$ always lies in the symmetric subspace, Lemma~\ref{lem:symmetry_eigenspace_Lambda} rules out support on any eigenvalue $\lambda \equiv 4$ or $6 \pmod{8}$. This narrows the candidates to $\lambda \equiv 0$ or $2 \pmod{8}$.

The second constraint comes from \emph{parity symmetry}: if $\ket{\psi}$ is an even state, then $P\ket{\psi} = \pm\ket{\psi}$. The next lemma describes the action of the parity operator on the Bell eigenbasis.

\begin{lem}[Parity reversal]  \label{lem:inverting_bell}
    Let $P^{(1)}= \prod_j Z_j\otimes I_j$ be the parity operator acting on the first copy. Then,
    \begin{equation}
        P^{(1)} \ket{\Psi_{\boldsymbol{\mu}}} \bra{\Psi_{\boldsymbol{\mu}}} P^{(1)} = \ket{\Psi_{-\boldsymbol{\mu}}} \bra{\Psi_{-\boldsymbol{\mu}}}. 
    \end{equation}
\end{lem}
\begin{proof}
    $P^{(1)}$ acts on the Bell states by flipping $r^z_j \mapsto r^z_j\oplus 1$ (up to a sign) for all $j\in [n]$, which by Eq.~\eqref{eq:r_func} corresponds to $\boldsymbol{\mu} \mapsto -\boldsymbol{\mu}$.
\end{proof}

\begin{cor}  \label{cor:inverting_Pi}
    $P^{(1)} \Pi_\lambda P^{(1)} = \Pi_{-\lambda}$.
\end{cor}

Combining the exchange-symmetry constraint (Lemma~\ref{lem:symmetry_eigenspace_Lambda}) with the parity-reversal constraint (Lemma~\ref{lem:inverting_bell}) pins down the allowed eigenvalues to a subset of the integers, as we show in the next subsection.

\subsection{Bell-state weights and eigenvalue restriction}
\label{sec:bell_weights}

We now define the two probability distributions that will play a central role throughout the paper, and use the symmetry constraints of the previous subsection to establish their key properties.
For a given pure state $\ket{\psi}$, the \emph{Bell-state weights} are
\begin{equation}
    p_\psi({\boldsymbol{\mu}}) \coloneqq  |\langle\Psi_{\boldsymbol{\mu}}|\psi\rangle^{\otimes 2}|^2 = 2^{-n}|\bra{\psi}\sigma_{\mathbf{r}}\ket{\psi^*}|^2,
\end{equation}
where $\mathbf{r}=\{\mathbf{r}^{z},\mathbf{r}^{x}\}$ is defined from $\boldsymbol{\mu}$ as in Eq.~\eqref{eq:r_func}. Moreover, the bridge spectral distribution is
\begin{equation}
    p_\psi(\lambda) \coloneqq \bra{\psi}^{\otimes 2} \Pi_\lambda \ket{\psi}^{\otimes 2} = \sum_{\boldsymbol{\mu}:\sum_j\mu_j=\lambda} p_\psi(\boldsymbol{\mu}).
\end{equation}

By Lemma~\ref{lem:eigenbasis_bell}, $p_\psi(\lambda)$ is the probability of observing eigenvalue $\lambda$ when performing a Bell measurement on two copies of $\ket{\psi}$. The parity constraint immediately yields the following reflection symmetry.

\begin{lem}[Symmetry of the Bell-state weights] \label{lem:bell_weight}
Let $\psi$ be an even pure state. For any $\boldsymbol{\mu} \in \{\pm1\}^{2n}$, it holds that
\begin{equation}
    p_\psi({\boldsymbol{\mu}})=p_\psi(-{\boldsymbol{\mu}}).
\end{equation}
\end{lem}
\begin{proof}
    Since $\ket{\psi}$ is an even pure state, we have $P^{(1)}\ket{\psi}^{\otimes 2}=\pm\ket{\psi}^{\otimes 2}$ (depending on whether $\ket{\psi}\in \mathcal{H}_n^+$ or $\ket{\psi}\in \mathcal{H}_n^-$). Therefore,
    \begin{equation}
        p_\psi(-{\boldsymbol{\mu}})= \left\lvert\bra{\Psi_{-\boldsymbol{\mu}}}\psi\rangle^{\otimes 2}\right\rvert^2=\left\lvert\bra{\Psi_{\boldsymbol{\mu}}} P^{(1)}\ket{\psi}^{\otimes 2}\right\rvert^2=\left\lvert\bra{\Psi_{\boldsymbol{\mu}}}\psi\rangle^{\otimes 2}\right\rvert^2=  p_\psi(\boldsymbol{\mu}),
    \end{equation}
    where we have used Lemma~\ref{lem:inverting_bell}. This concludes the proof.
\end{proof}

\begin{cor} \label{cor:pm_Lambda}
Let $\psi$ be an even pure state. It holds that $p_\psi(\lambda)=p_\psi(-\lambda)$.
\end{cor}

We now combine the two constraints --- exchange symmetry (Lemma~\ref{lem:symmetry_eigenspace_Lambda}) and parity symmetry (Lemma~\ref{lem:bell_weight}) --- to pin down the set of allowed eigenvalues to $8\mathbb{Z}$.

\begin{lem}[Eigenvalue restriction] \label{lem:allowed_eig_Lambda}
 Let $\psi$ be an even pure state. If $p_\psi(\lambda)\neq 0$, then $\lambda =8k$ for some integer $k$.
\end{lem}
\begin{proof}
     Since $\ket{\psi}^{\otimes 2}$ lies in the symmetric subspace, Lemma~\ref{lem:symmetry_eigenspace_Lambda} implies that any eigenvalue $\lambda$ for which $p_\psi(\lambda)\neq 0$ must be of the form $\lambda \in \{8k,8k+2\}$ for some integer $k$.

     On the other hand, by Corollary~\ref{cor:pm_Lambda}, if $p_\psi(\lambda) \neq 0$, then $p_\psi(-\lambda) \neq 0$ as well. Hence, $-\lambda$ must also belong to the same set of admissible eigenvalues, i.e., $-\lambda \in \{8k,8k+2\}$. The only values of $\lambda$ consistent with both conditions are those of the form $\lambda = 8k$.
\end{proof}

Combining Lemma~\ref{lem:allowed_eig_Lambda} with the spectrum bound $|\lambda|\leq 2n$ identifies the support of the bridge spectral distribution as
\begin{equation}\label{eq:def_calA}
    \mathcal{A} \;\coloneqq\; 8\mathbb{Z}\cap[-2n,\,2n].
\end{equation}
 In particular, $|\mathcal{A}| = 2\lfloor n/4\rfloor+1 $.

\subsection{The second-order matchgate commutant}
\label{sec:commutant}

Having established the spectral structure of $\Lambda$ and the constraints on the bridge spectral distribution, we now connect these results to the second-order commutant of the matchgate group. The key observation is that the spectral projectors $\{\Pi_\lambda\}$ provide an orthogonal basis for the second-order matchgate commutant~\cite{sierant2026theorymatchgatecommutant,BrDi26}. This allows us to express the matchgate twirling channel in closed form where the resulting coefficients are precisely given by the weights $\{p_\psi(\lambda)\}$, which is accessible by Bell sampling. These weights capture the entire second moment of the matchgate orbit and therefore form a complete set of invariants at second order. 

\begin{lem}[Basis for the second-order matchgate commutant,~\cite{sierant2026theorymatchgatecommutant,BrDi26}] \label{lem:projector_ortho_basis}
    The operators $\Pi_\lambda / \sqrt{\binom{2n}{n+\lambda/2}}$ form an orthonormal basis of $\mathrm{Com}_2(M_n)$.
\end{lem}

With an orthonormal basis of $\mathrm{Com}_2(M_n)$ in hand, we can immediately write down the matchgate twirling channel $\mathcal{T}_{M_n}^{(2)}$ in closed form using Eq.~\eqref{eq:twirling_project_commutant}. 

\begin{lem}[Matchgate twirling channel]\label{lem:twirling_Pi_decomp}
    Let $O\in \mathcal{B}(\mathcal{H}^{\otimes 2})$. It holds that
    \begin{equation} \label{eq:a}
        \mathcal{T}_{M_n}^{(2)} (O)  = \sum_{\lambda\in 2\mathbb{Z}} \Tr[\Pi_\lambda O] \frac{\Pi_\lambda}{\binom{2n}{n+\lambda/2}}.
    \end{equation}
    In addition, if $O\in \mathcal{D}(\mathcal{H}^{\otimes 2})$, then $\Tr[\Pi_\lambda O]\geq0$ and $\sum_{\lambda\in 2\mathbb{Z}} \Tr[\Pi_\lambda O]=1$.
\end{lem}
\begin{proof}
    Eq.~\eqref{eq:a} follows from Eq.~\eqref{eq:twirling_project_commutant}, by expanding in the orthonormal basis of Lemma~\ref{lem:projector_ortho_basis}. Since each $\Pi_\lambda$ is a projector, $\Tr[\Pi_\lambda O]\ge0$ for any positive operator $O$ and $\sum_{\lambda\in 2\mathbb{Z}} \Tr[\Pi_\lambda O]=\Tr[O]=1$.
\end{proof}

This representation has a simple but powerful consequence for distances: since the commutant is spanned by the orthogonal projectors $\{\Pi_\lambda\}$, the trace distance between positive operators in the commutant reduces to the total variation distance between their coefficients.

\begin{cor}[Trace distance in the commutant] \label{cor:trace_distance_comm}
    Let $O_1, O_2 \in \mathrm{Com}_2(M_n)$ with $O_1,O_2\in \mathcal{D}(\mathcal{H}^{\otimes 2})$. Then,
    \begin{equation}
        D(O_1,O_2) = \mathrm{TV}\!\left(\{\Tr[\Pi_\lambda O_1]\}_\lambda,\,\{\Tr[\Pi_\lambda O_2]\}_\lambda\right).
    \end{equation}
\end{cor}

Specializing Lemma~\ref{lem:twirling_Pi_decomp} to an even pure state gives the following explicit form of the matchgate-twirled two-copy state.

\begin{lem}[Twirled two-copy state] \label{lem:two_copy_ave_matchgate}
    Let $\psi$ be an even pure state. It holds that
    \begin{equation}
        \mathcal{T}_{M_n}^{(2)} (\psi^{\otimes 2}) = \sum_{\lambda\in \mathcal{A}} p_\psi(\lambda) \frac{\Pi_\lambda}{\binom{2n}{n + \lambda/2}}.
    \end{equation}
\end{lem}
\begin{proof}
    Applying Lemma~\ref{lem:twirling_Pi_decomp} with $O = \psi^{\otimes 2}$ gives
    \begin{equation}
        \mathcal{T}_{M_n}^{(2)} (\psi^{\otimes 2}) = \sum_\lambda p_\psi(\lambda) \frac{\Pi_\lambda}{\binom{2n}{n+\lambda/2}}.
    \end{equation}
    Lemma~\ref{lem:allowed_eig_Lambda} restricts the sum to $\lambda \in \mathcal{A}$.
\end{proof}

Lemma~\ref{lem:two_copy_ave_matchgate} shows that the entire matchgate-averaged second moment of a state is determined solely by the weights $\{p_\psi(\lambda)\}$. A first consequence is that the bridge spectral distribution is invariant under matchgate unitaries:

\begin{lem}[Matchgate invariance of the bridge spectral distribution] \label{lem:invariance_weight}
    Let $\psi$ be an even pure state, and let $\ket{\phi} = U\ket{\psi}$ for a matchgate $U$. Then, for any eigenvalue $\lambda$ of $\Lambda$,
    \begin{equation}
        p_\phi(\lambda) = p_\psi(\lambda).
    \end{equation}
\end{lem}
\begin{proof}
This follows from the definition of $p_\psi$ and the fact that $\Pi_\lambda$ is in the second-order matchgate commutant (Lemma~\ref{lem:projector_ortho_basis}).
\end{proof}

A second consequence is an explicit characterization of Gaussian states: they are precisely the states for which all bridge spectral distribution is concentrated at $\lambda = 0$.

\begin{lem}[Twirl of Gaussian states]\label{lem:twirl_vacuum}
Let $\ket{\mathbf{0}}$ denote the fermionic vacuum. Then
\begin{equation}
\mathcal{T}_{M_n}^{(2)} \left(\ketbra{\mathbf{0}}{\mathbf{0}}^{\otimes 2}\right) =
\frac{\Pi_0}{\binom{2n}{n}} .
\end{equation}
Conversely, let $\psi$ be an even pure state. If
$\mathcal{T}_{M_n}^{(2)}\left(\psi^{\otimes 2}\right) =
\Pi_0/\binom{2n}{n}$,
then $\psi$ is an FGS.
\end{lem}
\begin{proof}
 Since the vacuum state is an FGS, the Gaussianity condition in Eq.~\eqref{eq:Lambda_condition} immediately implies that $\ket{\mathbf{0}}^{\otimes 2} \in V_0$, and hence $p_{\mathbf{0}}(0)=1$. Substituting into Lemma~\ref{lem:two_copy_ave_matchgate} gives the first claim. Conversely, $p_\psi(0)=1$ implies $\Lambda\ket{\psi}^{\otimes 2}=0$, so $\psi$ is an FGS by Eq.~\eqref{eq:Lambda_condition}.
\end{proof}

Since any two states related by a matchgate share the same weights $\{p_\psi(\lambda)\}$ (Lemma~\ref{lem:invariance_weight}), the bridge spectral distribution is a genuine invariant of the Gaussian equivalence class of $\psi$. In particular, a pure state is an FGS if and only if all weight is concentrated at $\lambda = 0$ (Lemma~\ref{lem:twirl_vacuum}); any weight at $\lambda \neq 0$ is a direct, experimentally accessible signature of non-Gaussianity. To summarize, the distribution $\{p_\psi(\lambda)\}_{\lambda \in \mathcal{A}}$ enjoys the following properties on even pure states $\psi$:
\begin{enumerate}[label=(\roman*)]
    \item it is \emph{directly measurable} via Bell sampling on two copies of $\psi$;
    \item it is \emph{supported on at most $2\lfloor n/4 \rfloor + 1$ integer values} for any even pure state;
    \item it is \emph{symmetric} around zero: $p_\psi(\lambda) = p_\psi(-\lambda)$;
    \item it is \emph{invariant under matchgate unitaries}; and
    \item it is \emph{a complete second-moment invariant}: two states have the same matchgate-twirled two-copy state if and only if they have the same bridge spectral distribution.
\end{enumerate}
These properties make $\{p_\psi(\lambda)\}$ a natural foundation for quantifying non-Gaussianity, and we now turn to formalizing this in the next section.

\section{The bridge degree}
\label{sec:maxlambda}

The spectral theory of the previous section has revealed that, for every even pure state, the bridge spectral distribution $\{p_\psi(\lambda)\}_{\lambda\in \mathcal{A}}$ provides a natural language for quantifying fermionic non-Gaussianity. This motivates the definition of the \emph{bridge degree}, defined as the largest eigenvalue of the bridge operator whose eigenspace is occupied by the two-copy state. We devote this section to showing that it is a genuine non-Gaussianity monotone. In particular, we shall establish that the bridge degree:
\begin{itemize}[leftmargin=2em]
\item[(i)] takes \emph{discrete} integer values in $\{0, 4,8, \ldots\}$;
\item[(ii)] is \emph{faithful}: it vanishes precisely on pure FGSs;
\item[(iii)] is \emph{additive} under tensor products of even pure states; and
\item[(iv)] is \emph{non-increasing under pure-state post-selected Gaussian protocols};
\end{itemize}
Unlike most known non-Gaussianity monotones~\cite{dias2024classical,cudby2023gaussian}, which are typically computed via optimization procedure, bridge degree is determined directly from the support of the bridge spectral distribution. Operational consequences of these properties are developed in Sec.~\ref{sec:consequences}.

\subsection{Definition and basic properties}

We begin with a formal definition.
\begin{defn}[Bridge degree]\label{def:maxlambda}
    Let $\ket{\psi}$ be an even pure state of $n$ qubits. The bridge degree $\nGMax(\psi)$ is defined as the largest $\alpha \geq 0$ such that $p_\psi(2\alpha)\neq 0$.
\end{defn}

Writing $\mathcal{V}_\lambda \coloneqq \bigoplus_{|j|\leq \lambda} V_j$, an equivalent formulation is that $\nGMax(\psi)$ is the smallest $\alpha \geq 0$ such that $\ket{\psi}^{\otimes 2} \in \mathcal{V}_{2\alpha}$. 

The first thing to note is that the bridge degree is an integer-valued quantity. This is a direct consequence of the eigenvalue structure of $\Lambda$.

\begin{lem}[Discreteness and boundedness]\label{lem:values_max_Lambda}
For any $n$-qubit even pure state $\psi$, the bridge degree $\nGMax(\psi)$ takes values of the form $4k$ for some integer $k$ and satisfies
\begin{equation}
    0 \leq \nGMax(\psi) \leq 4\lfloor n/4 \rfloor.
\end{equation}
\end{lem}
\begin{proof}
By Lemma~\ref{lem:allowed_eig_Lambda}, $p_\psi(\lambda)\neq 0$ implies $\lambda = 8k$ for some integer $k$. Since $\nGMax(\psi) = \lambda/2$ for nonzero weight, we obtain $\nGMax(\psi)=4k$. The upper bound follows by noting that the largest eigenvalue of $\Lambda$ is $2n$.
\end{proof}

A key property is that the bridge degree vanishes if and only if the state is an FGS, as proven in the following lemma.
 
\begin{lem}[Faithfulness]\label{lem:faithfulness}
    $\nGMax(\psi) = 0$ if and only if $\psi$ is a pure FGS.
\end{lem}
\begin{proof}
    This is a restatement of the pure-state Gaussianity condition $\Lambda \ket{\psi}^{\otimes 2} = 0$ from Eq.~\eqref{eq:Lambda_condition}: $\nGMax(\psi) = 0$ iff all weight is concentrated at $\lambda = 0$ iff $\psi$ is an FGS.
\end{proof}

An immediate consequence of the interplay between discreteness and faithfulness is that non-Gaussianity cannot exist on fewer than four qubits.
 Indeed, for $n\leq 3$, the upper bound in Lemma~\ref{lem:values_max_Lambda} immediately yields $\nGMax(\psi)=0$. By faithfulness, this implies that $\psi$ is an FGS. This is consistent with the known fact that
all even pure states of $n\leq 3$ qubits are FGSs~\cite{bravyi2005classicalcapacityfermionicproduct}.

Additionally, the bridge degree exhibits particularly simple behavior for systems with up to seven qubits: it takes only two possible values.

\begin{lem}[Bridge degree on $\leq 7$ qubits]\label{lem:max_Lambda_4_qubit}
Any even non-Gaussian pure state $\psi$ on $4 \leq n \leq 7$ qubits satisfies $\nGMax(\psi) = 4$.
\end{lem}

\begin{proof}
By Lemma~\ref{lem:values_max_Lambda}, the bridge degree for states in the range $4 \leq n \leq 7$ can only take the values $0$ and $4$. By faithfulness (Lemma~\ref{lem:faithfulness}), any non-Gaussian state must therefore have $\nGMax(\psi) = 4$.
\end{proof}

Moreover, the forthcoming lemma establishes that the bridge degree is additive under the tensor product of even pure states.
 
\begin{lem}[Additivity]\label{lem:additive}
For any even pure states $\psi, \phi$,
\begin{equation}
    \nGMax(\psi \otimes \phi) = \nGMax(\psi) + \nGMax(\phi).
\end{equation}
\end{lem}
\begin{proof}
Expanding $(\ket{\psi}\otimes\ket{\phi})^{\otimes 2}$ in the tensor-product eigenbasis of $\Lambda$, each component lives in $V_\lambda \otimes V_{\lambda'} \subset V_{\lambda+\lambda'}$. The maximal $\lambda + \lambda'$ with nonzero weight is $\nGMax(\psi)+\nGMax(\phi)$, which is attained by the fact that weights come in $\pm\lambda$ pairs (Lemma~\ref{lem:bell_weight}).
\end{proof}

We also record a bound relating bridge degree to the \emph{Gaussian nullity} $\nu(\psi)$ of Ref.~\cite{mele2025efficient}, defined as the smallest number of qubits to which the non-Gaussianity of $\psi$ can be compressed via a matchgate unitary.

\begin{lem}[Relation between bridge degree and Gaussian nullity]\label{lem:bound_nullity}
    For any $n$-qubit even pure state $\psi$,
    \begin{equation}\label{eq:bound_nullity}
        \nGMax(\psi) \leq 4 \lfloor\nu(\psi)/4 \rfloor.
    \end{equation}
\end{lem}
\begin{proof}
    By definition of the Gaussian nullity, there exists a matchgate $U$ with $U\ket{\psi} = \ket{0}^{\otimes (n-\nu(\psi))} \otimes \ket{\phi}$ for some $\nu(\psi)$-qubit state $\phi$~\cite{mele2025efficient}. Invariance under matchgates (Lemma~\ref{lem:invariance_weight}) and additivity under tensor products (Lemma~\ref{lem:additive}) give $\nGMax(\psi) = \nGMax(\phi)$, and the bound $\nGMax(\phi) \leq 4 \lfloor\nu(\psi)/4 \rfloor$ follows from Lemma~\ref{lem:values_max_Lambda}.
\end{proof}

Finally, we remark that the bridge degree can be witnessed directly from the Bell-sampling primitive (Algorithm~\ref{alg:bell_sampling}): given $N$ i.i.d.\ samples $\lambda_1,\dots,\lambda_N\sim p_\psi$, the empirical estimator
\begin{equation} \label{eq:bridge_degree_witness}
    \widehat\Lambda_{\mathrm{d}} \;\coloneqq\; \tfrac{1}{2}\,\max_{k}|\lambda_k|
\end{equation}
is, by construction, always a lower bound on $\nGMax(\psi)$. It saturates the true bridge degree once the extremal sector is observed --- which happens with probability $1-(1-2p_\psi(2\nGMax))^N$. This empirical estimator is tight when $p_\psi$ has appreciable mass at its extremal value, but can require exponentially many samples for states whose extremal sector carries only exponentially small weight, motivating the $\epsilon$-approximate variant developed in Sec.~\ref{sec:approx}.

\subsection{Monotonicity under post-selected Gaussian protocols}
\label{sec:monotonicity}

We now come to the central result of this section: the bridge degree is non-increasing under the post-selected Gaussian protocols of Definition~\ref{def:gaussian_protocol}.

More precisely, since the bridge degree is naturally defined for pure states, we consider a pure state $\ket{\psi}$ together with a post-selected Gaussian protocol $\mathcal{E}$ for which the output state is pure,
\begin{equation}
    \mathcal{E}(\ket{\psi}\!\bra{\psi}) = \ket{\phi}\!\bra{\phi}.
\end{equation}
This is a pointwise condition on the pair $(\mathcal{E},\psi)$ rather than a global property of $\mathcal{E}$: we do \emph{not} demand that $\mathcal{E}$ map every pure input to a pure output, only that the specific input $\ket{\psi}$ does so.

\begin{thm}[Monotonicity of the bridge degree]\label{thm:max_Lambda_monotone}
Let $\ket{\psi}$ be an even pure state and let $\mathcal{E}$ be a post-selected Gaussian protocol such that $\mathcal{E}(\ket{\psi}\!\bra{\psi}) = \ket{\phi}\!\bra{\phi}$ is pure. Then
\begin{equation}
    \nGMax(\phi) \;\leq\; \nGMax(\psi).
\end{equation}
\end{thm}

\begin{proof}
It suffices to prove the claim for each elementary operation in Definition~\ref{def:gaussian_protocol}. In particular, since the output $\ket{\phi}$ is pure, every Kraus branch of a Gaussian protocol $\mathcal{E}$ maps $\ket{\psi}$ to the same state, which can be seen as follows (see, e.g.,~\cite{tarabunga2025nonstabilizerness}). Writing $\mathcal{E}(\cdot)=\sum_k K_k(\cdot)K_k^\dagger$, the identity $\sum_k(K_k\ket{\psi})(K_k\ket{\psi})^\dagger=\ket{\phi}\!\bra{\phi}$ expresses a rank-one operator as a sum of positive semidefinite operators. Hence, each nonzero branch $K_k\ket{\psi}$ must be proportional to $\ket{\phi}$, i.e., $K_k\ket{\psi}\propto\ket{\phi}$ for every $k$. Therefore, every nonzero branch of $\mathcal{E}$ produces the same pure output state $\ket{\phi}$, so it is sufficient to consider a single branch, corresponding to post-selecting on the measurement outcomes. It therefore suffices to verify monotonicity under (i) applying a fermionic Gaussian unitary, (ii) appending or discarding a pure FGS, and (iii) performing an occupation-number measurement with post-selection.

\emph{(i) Fermionic Gaussian unitaries.} This is immediate from the matchgate-invariance of the bridge spectral distribution (Lemma~\ref{lem:invariance_weight}).

\emph{(ii) Appending/discarding a pure FGS.} It is immediate by additivity under tensor product (Lemma~\ref{lem:additive}) and faithfulness (Lemma~\ref{lem:faithfulness}).

\emph{(iii) Occupation-number measurement with post-selection.} Without loss of generality, consider the measurement of $Z_1$, with projector $P_\mu = (I + \mu Z_1)/2$ for $\mu=\pm 1$. The post-measurement state is $\ket{\psi_\mu} = P_\mu \ket{\psi}/\sqrt{p_\mu}$ where $p_{\mu} = \bra{\psi} P_{\mu} \ket{\psi}$. One can verify that the operator $P_\mu \otimes P_\mu$ commutes with $\Lambda$; consequently, each eigenspace $V_\lambda$ admits a decomposition $V_\lambda = V_\lambda^0 \oplus V_\lambda^1$,  where $V_\lambda^{\nu}$ denotes the subspace of $V_\lambda$ with eigenvalue $\nu\in\{0,1\}$ of $P_\mu \otimes P_\mu$. Expanding $\ket{\psi}^{\otimes 2}$ in this common eigenbasis gives
\begin{equation}
\begin{split}
\ket{\psi}^{\otimes 2}
&= \sum_{\lambda,m_\lambda} c_{\lambda,m_\lambda} \ket{\Psi_{\lambda,m_\lambda}} \\
&= \sum_{\lambda} \left( \sum_{m_\lambda:\Psi\in V_\lambda^{0}} c_{\lambda,m_\lambda} \ket{\Psi_{\lambda,m_\lambda}}
+ \sum_{m_\lambda:\Psi\in V_\lambda^{1}} c_{\lambda,m_\lambda} \ket{\Psi_{\lambda,m_\lambda}} \right).
\end{split}
\end{equation}
Applying $P_1\otimes P_1$ selects only the components in $V_\lambda^{1}$:
\begin{equation}
(P_1\otimes P_1)\ket{\psi}^{\otimes 2}
= \sum_{\lambda}\sum_{m_\lambda:\,\Psi\in V_\lambda^{1}} c_{\lambda,m_\lambda} \ket{\Psi_{\lambda,m_\lambda}}.
\end{equation}
After renormalization, the post-measurement state  $\ket{\psi_{1}}\ket{\psi_{1}}$ therefore has support only on those eigenvalues whose eigenvectors lie in the subspace $V_\lambda^1$, i.e., a subset of those present in the original state. It follows that the largest eigenvalue with nontrivial support cannot increase, that is,
$\nGMax(\psi_{1}) \le \nGMax(\psi)$. Therefore, $\nGMax$ is non-increasing under occupation-number measurements with post-selection.

\end{proof}

An immediate consequence of monotonicity is that bridge degree lower-bounds the number of non-Gaussian gates required to prepare a state. In the fermionic setting, it is natural to consider $t$-doped Gaussian protocols~\cite{mele2025efficient,trigueros2026unitarydesigns} in which at most $t$ non-Gaussian gates are interspersed with Gaussian protocols. Here we consider $4$-local non-Gaussian operations, which are generated by at most $4$ Majorana operators. 
\begin{defn}[$t$-doped Gaussian protocols]
    A $t$-doped Gaussian protocol is a protocol
    \begin{equation}\label{eq:tdoped_def}
        \mathcal{E} = \mathcal{E}_t W_t \cdots \mathcal{E}_1 W_1 \mathcal{E}_0,
    \end{equation}
    where each $\mathcal{E}_i$ is an arbitrary fermionic Gaussian protocol, possibly with post-selection, and each $W_i$ is a $4$-local non-Gaussian gate.
\end{defn}
We refer to states that can be prepared from the vacuum by a $t$-doped Gaussian protocol as \emph{$t$-doped Gaussian states}. The protocol is illustrated in Fig.~\ref{fig:t_doped}. This model is universal: any quantum computation can be implemented by injecting SWAP gates --- a $4$-local non-Gaussian operation --- on top of Gaussian circuits~\cite{jozsa2008matchgates,bravyi2002fermionic}, and each SWAP can in turn be realized via a Gaussian protocol consuming one copy of the magic state~\cite{hebenstreit2019all}
\begin{equation}\label{eq:magic_swap}
    \ket{M} = \frac{1}{2}\bigl(\ket{0000}+\ket{0101}+\ket{1010}+\ket{1111}\bigr).
\end{equation}

\begin{cor}[Non-Gaussian gate lower bound for state preparation]\label{cor:bound_dope}
    If an even pure state $\psi$ can be prepared from the vacuum by a $t$-doped Gaussian protocol, then
    \begin{equation}\label{eq:bound_dope}
        \nGMax(\psi) \leq 4t.
    \end{equation}
\end{cor}
\begin{proof}
    Any 4-local non-Gaussian gate $W$ can be implemented by injecting a four-qubit magic state $M_W$ via a Gaussian protocol~\cite{hebenstreit2019all}. By monotonicity (Theorem~\ref{thm:max_Lambda_monotone}) and additivity (Lemma~\ref{lem:additive}),
    \begin{equation}
        \nGMax(W\psi W^\dagger) \leq \nGMax(\psi \otimes M_W) = \nGMax(\psi) + \nGMax(M_W).
    \end{equation}
    By Lemma~\ref{lem:max_Lambda_4_qubit}, $\nGMax(M_W) = 4$ for any non-Gaussian $M_W$. Iterating over the $t$ non-Gaussian gates, starting from the vacuum (for which $\nGMax = 0$), gives $\nGMax(\psi) \leq 4t$.
\end{proof}

Furthermore, the bridge degree induces a natural hierarchy of even pure states through the nested sets $\mathcal{F}_k \coloneqq \{\psi : \nGMax(\psi)\leq k\}$, with $\mathcal{F}_0$ the pure FGSs and $\mathcal{F}_{4\lfloor n/4\rfloor}$ the full even Hilbert space. A structurally similar hierarchy has been introduced in the \emph{particle-preserving} setting by Ref.~\cite{semenyakin2025classifying} through the generalized Pl\"ucker relations, naturally defining a particle-preserving counterpart of the bridge degree which we call the \emph{Pl\"ucker degree}. A proof strategy similar to above yields an analogous monotonicity result there: the Pl\"ucker degree is a monotone under particle-preserving Gaussian protocols. We discuss this connection in Appendix~\ref{app:plucker}.

Before leaving this section, we note that the bridge degree extends to mixed states by the standard convex-roof construction: for any even state $\rho$, define the \emph{convex-roof extended bridge degree}
\begin{equation} \label{eq:cr_bridge_degree}
    \nGMax^{\,\mathrm{cr}}(\rho) \;\coloneqq\; \min_{\{p_i,\psi_i\}\,:\,\rho=\sum_i p_i\,\ketbra{\psi_i}{\psi_i}}\, \max_i\, \nGMax(\psi_i),
\end{equation}
where the minimum runs over all pure-state decompositions of $\rho$. This is the smallest $r$ such that $\rho$ admits a pure-state decomposition into states all with bridge degree at most $r$. On pure $\psi$, $\nGMax^{\,\mathrm{cr}}(\psi) = \nGMax(\psi)$. By the general theory of convex-roof monotones~\cite{vidal2000entanglement,terhal2000schmidt}, $\nGMax^{\,\mathrm{cr}}$ is automatically a non-Gaussianity monotone under post-selected Gaussian protocols, and reduces to $\nGMax$ on pure states. While it yields a valid monotone, its computation requires optimization over all decompositions and is generally intractable. In Sec.~\ref{sec:mixed} we develop a different mixed-state extension based on the Choi representation, which is computationally tractable at the cost of being monotone only under a smaller class of post-selected Gaussian operations (Definition~\ref{def:gaussian_operations}).

\section{Consequences of bridge degree monotonicity}
\label{sec:consequences}

The monotonicity of the bridge degree under post-selected Gaussian protocols, established in Sec.~\ref{sec:monotonicity}, has several operational consequences that we develop in turn. First, in Sec.~\ref{sec:nogo_conversion}, we derive explicit no-go theorems for state conversion under Gaussian protocols. Then, in Sec.~\ref{sec:irreversibility}, we establish that the resulting resource theory is irreversible in the exact-conversion setting. Finally, in Sec.~\ref{sec:designs_lower}, we use the bridge degree to derive lower bounds on the number of non-Gaussian gates required by any circuit to generate a state design.

\subsection{No-go theorems for Gaussian conversion}
\label{sec:nogo_conversion}

In this section, we show that the bridge degree yields no-go theorems for state conversion under Gaussian protocols. Since it is monotone under post-selected Gaussian protocols (Theorem~\ref{thm:max_Lambda_monotone}), any Gaussian conversion $\rho\to\sigma$ achievable with nonzero probability requires $\nGMax(\rho)\ge\nGMax(\sigma)$, and is therefore impossible whenever $\nGMax(\rho)<\nGMax(\sigma)$ (see Sec.~\ref{sec:nongaussianity_monotones}).

As an illustration, consider the family of $n$-qubit GHZ states
\begin{equation}\label{eq:ghz}
    \ket{\mathrm{GHZ}_n} = \frac{1}{\sqrt{2}}\bigl(\ket{0}^{\otimes n} + \ket{1}^{\otimes n}\bigr),
\end{equation}
for even $n$ (the state is not even for odd $n$). It is worth noting that $\ket{\mathrm{GHZ}_4}$ is matchgate-equivalent to the canonical magic state $\ket{M}$ in Eq.~\eqref{eq:magic_swap}. We compute the bridge degree of $\ket{\text{GHZ}_n}$ in the following lemma.

\begin{lem}[Bridge degree of GHZ states]\label{lem:max_Lambda_ghz}
For any even $n \geq 4$,
\begin{equation}
    \nGMax(\mathrm{GHZ}_n) = 
    \begin{cases}
        n, & n = 4k, \\
        n-2, & n = 4k+2.
    \end{cases}
\end{equation}
\end{lem}
\begin{proof}
For $n = 4k$, the largest eigenvalue of $\Lambda$ is $\lambda = 2n = 8k$, whose eigenspace is one-dimensional and spanned by $(\ket{\sigma_{01}}\ket{\sigma_{11}})^{\otimes 2k}$. The overlap with $\ket{\mathrm{GHZ}_n}^{\otimes 2}$ is
\begin{equation}
    (\bra{\sigma_{01}}\bra{\sigma_{11}})^{\otimes 2k} \ket{\mathrm{GHZ}_n}^{\otimes 2} = (-1)^k \bra{\mathrm{GHZ}_n} X_1 Y_2 \cdots X_{n-1} Y_n \ket{\mathrm{GHZ}_n},
\end{equation}
which is nonzero because $X_1 Y_2 \cdots X_{n-1} Y_n$ stabilizes $\ket{\mathrm{GHZ}_n}$. Hence $\nGMax(\mathrm{GHZ}_n) = n$.

For $n = 4k+2$, Lemma~\ref{lem:allowed_eig_Lambda} forbids support at $\lambda = 2n = 8k+4$, so the largest admissible eigenvalue is $\lambda = 8k$. Consider the state $(\ket{\sigma_{01}}\ket{\sigma_{11}})^{\otimes 2k} \otimes \ket{\sigma_{01}}^{\otimes 2}  \in V_{8k}$. Similarly as above, we have
    \begin{equation}
    (\bra{\sigma_{01}}\bra{\sigma_{11}})^{\otimes 2k} \otimes \bra{\sigma_{01}}^{\otimes 2} \ket{\text{GHZ}_n}\ket{\text{GHZ}_n} = (-1)^k \bra{\text{GHZ}_n} X_1 Y_2 \dots X_{4k-1} Y_{4k} X_{4k+1} X_{4k+2} \ket{\text{GHZ}_n}.
    \end{equation}
    The operator $X_1 Y_2 \cdots X_{4k-1} Y_{4k} X_{4k+1} X_{4k+2}$ is again a stabilizer of $\ket{\mathrm{GHZ}_n}$, so the overlap is nonzero.  
Thus the largest occupied eigenvalue is $8k = 2(n-2)$, which gives $\nGMax(\mathrm{GHZ}_n)=n-2$.
\end{proof}

As an immediate application, we consider Gaussian conversions between tensor products of GHZ states. The bridge degree detects that two copies of $\ket{\text{GHZ}_6}$ cannot be converted into three copies of $\ket{\text{GHZ}_4}$, since $2  \nGMax(\mathrm{GHZ}_6) = 8 < 12 = 3\nGMax(\mathrm{GHZ}_4)$. This no-go cannot be obtained, for example, from the relative entropy of non-Gaussianity $\nGR$ (Eq.~\eqref{eq:ngr_def}). We have that $\nGR(\text{GHZ}_n)=n$ for even $n\geq 4$~\cite{tarabunga2026}. Therefore, a deterministic Gaussian conversion above cannot be ruled out by $\nGR$ since $2\nGR(\text{GHZ}_6)=3\nGR(\text{GHZ}_4)$ (here we have used that $\nGR$ is additive under tensor products). For the Gaussian rank, certifying the no-go would require computing $\chi_\G(\mathrm{GHZ}_4^{\otimes 3})$ and $\chi_\G(\mathrm{GHZ}_6^{\otimes 2})$ and showing the former to be larger; but the Gaussian rank is not known to be efficiently computable on these tensor-product states.

Notably, with the bridge degree we can obtain no-go theorems for larger $\mathrm{GHZ}$ states as easily. Consider now four copies of $\ket{\mathrm{GHZ}_6}$ and three copies of $\ket{\mathrm{GHZ}_8}$. The bridge degree rules out the conversion $\mathrm{GHZ}_6^{\otimes 4} \to \mathrm{GHZ}_8^{\otimes 3}$, since $4\,\nGMax(\mathrm{GHZ}_6) = 16 < 24 = 3\,\nGMax(\mathrm{GHZ}_8)$, whereas neither $\nGR$ nor $\chi_\G$ can. For the relative entropy, $\nGR(\mathrm{GHZ}_6^{\otimes 4}) = 24 = \nGR(\mathrm{GHZ}_8^{\otimes 3})$ as before. For the Gaussian rank, certifying the no-go is again obstructed by the difficult optimization in its computation. %
The remaining monotones similarly fail to provide no-go theorems: the second moment of $\Lambda$, like $\nGR$, is additive and computable but gives no obstruction, since $\nGLambda(\mathrm{GHZ}_n)=n$ implies $\nGLambda(\mathrm{GHZ}_6^{\otimes 4})=\nGLambda(\mathrm{GHZ}_8^{\otimes 3})$, whereas the Gaussian extent and fidelity, like the Gaussian rank, are not known to be efficiently computable on these tensor-product states. This no-go statement can therefore be reliably certified only by the bridge degree, which does so efficiently. It is straightforward to construct further no-go theorems involving larger GHZ states that can likewise be certified only by the bridge degree.

More fundamentally, the bridge degree is, to our knowledge, the first known non-Gaussianity monotone that is both additive and monotone under the strictly larger class of \emph{post-selected} Gaussian protocols~\footnote{In subsequent work with collaborators~\cite{turkeshi2026williamsonmajorizationtheoryfermionic}, we show that the Gaussian nullity shares both properties.}. Additivity yields the bridge degree of the tensor-product states from single-copy values, as used throughout the examples above, while monotonicity under post-selection further extends all these no-go theorems to probabilistic conversions. Among previously known monotones, the Gaussian rank is monotone under post-selection but not additive, whereas the relative entropy of non-Gaussianity and the second moment of $\Lambda$ are additive but not monotone under post-selection (Tab.~\ref{tab:known_monotones}); the no-go theorems from the latter two therefore apply only to deterministic conversions.

Beyond such specific conversions, the bridge degree yields a general no-go theorem: no probabilistic Gaussian protocol can transform an $n$-qubit non-Gaussian pure state into a tensor product of $m$ $4$-qubit non-Gaussian pure states when $n < 4m$, even when post-selection is allowed. This mirrors a recent analogous result for bosonic Gaussian conversion~\cite{mele2025symplecticrank}.

\begin{cor}[Impossibility of separately spreading non-Gaussianity via post-selected Gaussian protocols]\label{cor:spr}
Let $n < 4m$, and let $\rho$ be an $n$-qubit even non-Gaussian pure state. Let $\phi_1,\dots,\phi_m$ be $4$-qubit even non-Gaussian pure states. Then no post-selected Gaussian protocol can exactly map $\rho$ to $\bigotimes_{i=1}^m \phi_i$.
\end{cor}
\begin{proof}
This follows directly from three facts:
(i) the monotonicity of the bridge degree (Theorem~\ref{thm:max_Lambda_monotone});
(ii) the additivity of the bridge degree for tensor products (Lemma~\ref{lem:additive}), together with the fact that each $4$-qubit non-Gaussian pure state has bridge degree $4$ (Lemma~\ref{lem:max_Lambda_4_qubit}), implying a total of $4m$; and
(iii) the upper bound that the bridge degree of any state on $n$ modes is at most $n$ (Lemma~\ref{lem:values_max_Lambda}).
Combining these, the initial state has bridge degree at most $n$, while the target state has bridge degree exactly $4m$. When $n < 4m$, this violates monotonicity, and the conversion is therefore impossible.
\end{proof}

Remarkably, this impossibility persists even when the initial state is arbitrarily far from the set of Gaussian states while each of the $m$ target factors is arbitrarily close to it. This robustness places the no-go beyond the reach of \emph{any} previously known non-Gaussianity monotone. As the proof shows, the bridge degree achieves this through two properties: it assigns every $4$-qubit non-Gaussian state the \emph{discrete} value $4$, however close that state is to the Gaussian set, and it is additive. No previously known monotone combines these two features. The relative entropy of non-Gaussianity, the second moment of $\Lambda$, the Gaussian extent, and the Gaussian fidelity are continuous and can take arbitrarily small values on near-Gaussian $4$-qubit states, so the target value can be driven below that of the input, leaving no obstruction. The Gaussian rank, by contrast, is discrete, taking the value $2$ on every $4$-qubit non-Gaussian state~\footnote{This follows from the fact that any $4$-qubit even pure state can be transformed by a matchgate unitary into the form $\ket\psi = \alpha \ket{0000} + \beta\ket{1111}$~\cite{hebenstreit2019all}.}, but it lacks additivity, so its value on tensor products of $4$-qubit states cannot be inferred from the single-factor values.

\subsection{Irreversibility of the resource theory}
\label{sec:irreversibility}

The no-go theorems above rule out individual state conversions. A fundamental question in the manipulation of quantum resources is whether the resource theory itself is \emph{reversible}~\cite{Gour2025, regula2024reversibility}: given two states $\rho$ and $\sigma$, can one always convert $n$ copies of $\rho$ into $m$ copies of $\sigma$ at the same rate as the reverse transformation? In this section, we show that the resource theory of fermionic non-Gaussianity is, in fact, \emph{irreversible} under exact post-selected Gaussian protocols.

To make this precise, define the \emph{exact asymptotic distillation rate} and \emph{exact cost rate}~\cite{Gour2025}:
\begin{align}
\mathrm{Distill}(\rho \to \sigma) &= \sup_{m,n \in \mathbb{N}} \left\{\frac{m}{n} : \exists\, \mathcal{E} \in \G,\ \mathcal{E}(\rho^{\otimes n}) = \sigma^{\otimes m}\right\}, \\
\mathrm{Cost}(\rho \to \sigma) &= \inf_{m,n \in \mathbb{N}} \left\{\frac{n}{m} : \exists\, \mathcal{E} \in \G,\ \mathcal{E}(\rho^{\otimes n}) = \sigma^{\otimes m}\right\},
\end{align}
where $\G$ denotes the set of post-selected Gaussian protocols. These rates are related by $\mathrm{Distill}(\rho \to \sigma) = 1/\mathrm{Cost}(\rho \to \sigma)$~\cite{Gour2025}. A resource theory is said to be \emph{reversible} if 
\begin{equation}
    \mathrm{Distill}(\rho \to \sigma) = \mathrm{Cost}(\sigma \to \rho),
\end{equation}
for all states $\rho, \sigma$. Otherwise, the resource theory is called \emph{irreversible}.

We stress that here we are considering the \emph{exact setting} where free operations map $\rho^{\otimes n}$ into $\sigma^{\otimes m}$ \emph{exactly}, which differs from the \emph{approximate setting}, where free operations map $\rho^{\otimes n}$ into $\sigma^{\otimes m}$ \emph{approximately} with an error that vanishes in the asymptotic limit of many copies~\cite{Gour2025}.

In the following theorem, we prove that the resource theory of fermionic non-Gaussianity is irreversible under the set of all post-selected Gaussian protocols $\mathcal{G}$ in the exact setting.
\begin{thm}[Irreversibility of the resource theory of fermionic non-Gaussianity]\label{thm:irreversibility}
The resource theory of fermionic non-Gaussianity is irreversible under exact post-selected Gaussian protocols. That is, there exist two states $\rho$ and $\sigma$ such
that the exact distillable rate and exact cost rate satisfy
\begin{equation}
    \text{Distill}(\rho \to \sigma) \neq \text{Cost}(\sigma \to \rho).
\end{equation}
\end{thm}
\begin{proof}
If there exists a post-selected Gaussian protocol which maps $n$ copies of $\rho$ to $m$ copies of $\sigma$, then for any non-Gaussianity monotone $M$ under post-selected Gaussian protocols, it must hold that
\begin{equation} \label{eq:conversion_bound}
    M(\rho^{\otimes n}) \geq M(\sigma^{\otimes m}).
\end{equation}
Now, consider the states $\rho = \ket{\text{GHZ}_6}$ and $\sigma = \ket{\text{GHZ}_4}$. We will give bounds for the distillation rate and cost rate of the two states.

We will first bound the distillation rate of $\rho \to \sigma$ using the bridge degree. Because $\nGMax$ is (i) non-increasing under post-selected Gaussian protocols and
(ii) additive on tensor products of pure states, Eq.~\eqref{eq:conversion_bound} gives
\begin{equation}
        \frac{m}{n} \leq \frac{\nGMax(\rho)}{\nGMax(\sigma)}.
\end{equation}
Since $\nGMax(\text{GHZ}_6)=\nGMax(\text{GHZ}_4)=4$ (Lemma~\ref{lem:max_Lambda_ghz}), we obtain
\begin{equation} \label{eq:distill_rate_ghz}
        \text{Distill}(\text{GHZ}_6 \to \text{GHZ}_4) \leq \frac{\nGMax(\text{GHZ}_6)}{\nGMax(\text{GHZ}_4)}=1.
\end{equation}

Next, we bound the cost rate of $\sigma \to \rho$ using the Gaussian nullity $\nu$~\cite{mele2025efficient}, which is also non-increasing under post-selected Gaussian protocols~\cite{turkeshi2026williamsonmajorizationtheoryfermionic} and additive under tensor products. The Gaussian nullity of $\ket{\text{GHZ}_k}$ is $\nu(\text{GHZ}_k)=k$~\cite{tarabunga2026}. Applying Eq.~\eqref{eq:conversion_bound} to $\nu$ therefore gives
\begin{equation} \label{eq:cost_rate_ghz}
        \text{Cost}(\text{GHZ}_4 \to \text{GHZ}_6) \geq \frac{\nu(\text{GHZ}_6)}{\nu(\text{GHZ}_4)}=\frac{3}{2}.
\end{equation}
Combining Eqs.~\eqref{eq:cost_rate_ghz} and \eqref{eq:distill_rate_ghz}, we obtain
\begin{equation}
    \text{Distill}(\text{GHZ}_6 \to \text{GHZ}_4) < \text{Cost}(\text{GHZ}_4 \to \text{GHZ}_6),
\end{equation}
which concludes the proof.

\end{proof}

Theorem~\ref{thm:irreversibility} establishes a strict gap between the exact distillable rate and the exact cost rate. In particular, the gap in the proof above is given by
\begin{equation}
    \frac{\mathrm{Distill}(\rho \to \sigma)}{\mathrm{Cost}(\sigma \to \rho)} \leq \frac{2}{3}.
\end{equation}
Our next result shows that this ratio can be made arbitrarily small, i.e., the irreversibility of the resource theory of non-Gaussianity is arbitrarily strong. We first prove a preliminary lemma, showing that the bridge degree is also constrained by the particle number, independently of the system size.

\begin{lem}[Bound from particle number]\label{lem:bound_particle_number}
Let $\psi$ be an even pure state of $n$ qubits with definite particle number $N$. Then 
\begin{equation}\label{eq:bound_particle_number}
    \nGMax(\psi) \;\leq\; 4\left\lfloor \tfrac{1}{2}\min(N,\,n-N)\right\rfloor .
\end{equation}
\end{lem}

\begin{proof}
Writing
$\gamma_{2j-1}=a_j+a_j^\dagger$ and
$\gamma_{2j}=-i(a_j-a_j^\dagger)$ on each copy, the bridge operator becomes
\begin{equation}
\Lambda=2\sum_{j=1}^{n}T_j ,
\end{equation}
where
\begin{equation}
T_j
\coloneqq
a_j^\dagger\otimes a_j
+
a_j\otimes a_j^\dagger .
\end{equation}
The canonical anticommutation relations imply that the operators $T_j$ mutually commute. Moreover,
\begin{equation}
T_j^2
=
n_j\otimes(1-n_j)
+(1-n_j)\otimes n_j
=:Q_j ,
\end{equation}
where $n_j=a_j^\dagger a_j$. Thus $Q_j$ is the projector onto the subspace in which mode $j$ is occupied in exactly one of the two copies. Consequently, the eigenvalues of $T_j$ are $0,\pm1$.

Let $N_1$ and $N_2$ denote the particle-number operators of the two copies. Both $\Lambda$ and every $T_j$ commute with the total particle number $N_1+N_2$. In the sector $N_1+N_2=M$, the number of modes occupied in exactly one copy is bounded both by the total number of particles and by the total number of holes. Equivalently,
\begin{equation}
\sum_{j=1}^{n}Q_j
\leq M,
\qquad
\sum_{j=1}^{n}Q_j
\leq 2n-M .
\end{equation}
Since the $T_j$ commute, they admit a simultaneous eigenbasis. If $t_j\in\{0,\pm1\}$ are their eigenvalues on a joint eigenvector, then $|t_j|$ is the corresponding eigenvalue of $Q_j$. Hence every eigenvalue $\lambda$ of $\Lambda$ in the $N_1+N_2=M$ sector satisfies
\begin{equation}
|\lambda|
=
2\left|\sum_{j=1}^{n}t_j\right|
\leq
2\sum_{j=1}^{n}|t_j|
\leq
2\min(M,2n-M).
\end{equation}

Since $\ket{\psi}$ has particle number $N$, the state $\ket{\psi}^{\otimes2}$ belongs to the particle-number sector $M=2N$. Therefore every eigenvalue in the support of $p_\psi$ satisfies $|\lambda|\leq4\min(N,n-N)$.
Combining with the restriction $\lambda=8k$ for an  integer $k$ (Lemma~\ref{lem:allowed_eig_Lambda}) gives Eq.~\eqref{eq:bound_particle_number}.
\end{proof}

\begin{thm}[The irreversibility of the resource theory of fermionic non-Gaussianity is arbitrarily strong]\label{thm:unbounded_irreversibility}
For any $\epsilon > 0$, there exist even pure states $\rho, \sigma$ such that
\begin{equation}
    \frac{\mathrm{Distill}(\rho \to \sigma)}{\mathrm{Cost}(\sigma \to \rho)} < \epsilon.
\end{equation}
\end{thm}
\begin{proof}
The argument follows the proof of Theorem~\ref{thm:irreversibility}, by bounding the distillable rate with the bridge degree and the cost rate with the Gaussian nullity. For any two even non-Gaussian pure states $\rho$ and $\sigma$,
\begin{equation}\label{eq:ratio_bound}
    \frac{\mathrm{Distill}(\rho \to \sigma)}{\mathrm{Cost}(\sigma \to \rho)} \;\leq\; \frac{\nGMax(\rho)}{\nu(\rho)}\cdot\frac{\nu(\sigma)}{\nGMax(\sigma)} .
\end{equation}
As in the proof of Theorem~\ref{thm:irreversibility}, we take $\sigma = \ket{\mathrm{GHZ}_4}$ for which $\nu(\sigma) = \nGMax(\sigma) = 4$, so that the right-hand side reduces to $\nGMax(\rho)/\nu(\rho)$.

GHZ states cannot serve as $\rho$ to prove the claim here, since by Lemma~\ref{lem:max_Lambda_ghz} together with $\nu(\mathrm{GHZ}_n)=n$ they satisfy $\nGMax(\mathrm{GHZ}_n)/\nu(\mathrm{GHZ}_n)\geq 1-2/n$, which tends to $1$. What is needed instead is a state whose bridge degree grows much slower than the nullity.

To this end, consider the two-fermion pairing state on $n=2L$ modes,
\begin{equation}\label{eq:pc_state}
    \ket{\mathrm{P}_L} \;=\; \frac{1}{\sqrt L}\sum_{j=1}^{L} a^\dagger_{2j-1}a^\dagger_{2j}\ket{0} .
\end{equation}
Direct computation shows that its single-particle density matrix has all $n$ eigenvalues equal to $1/L$, so for $L\ge2$ none lies in $\{0,1\}$ and the Gaussian nullity is maximal, $\nu(\mathrm{PC}_L)=n=2L$. On the other hand, since the state has a definite particle number $N=2$, Lemma~\ref{lem:bound_particle_number} gives $\nGMax(\mathrm{P}_L)\leq 4$ and since it is non-Gaussian,
\begin{equation}\label{eq:pc_bridge_degree}
    \nGMax(\mathrm{P}_L) \;=\; 4 \qquad\text{for all } L\geq 2.
\end{equation}
Setting $\rho = \ket{\mathrm{PC}_L}$ in Eq.~\eqref{eq:ratio_bound},
\begin{equation}
    \frac{\mathrm{Distill}(\mathrm{P}_L \to \mathrm{GHZ}_4)}{\mathrm{Cost}(\mathrm{GHZ}_4 \to \mathrm{P}_L)} \;\leq\; \frac{\nGMax(\mathrm{P}_L)}{\nu(\mathrm{P}_L)} \;=\; \frac{2}{L},
\end{equation}
which is smaller than any $\epsilon > 0$ for $L > 2/\epsilon$.

\end{proof}

 These results establish this phenomenon in the fermionic setting, in direct analogy with the recent result for bosonic non-Gaussianity in Ref.~\cite{mele2025symplecticrank}, and in contrast with entanglement theory, which is reversible under exact post-selected LOCCs for pure states, where both the exact distillable entanglement and the exact entanglement cost are given by the logarithm of the Schmidt rank~\cite{Gour2025}.

\subsection{Lower bounds for state designs}
\label{sec:designs_lower}

We now apply the bridge degree framework to obtain lower bounds on the non-Gaussian gate count required to generate state $k$-designs. First, we compute the Haar bridge spectral distribution $p_H$ explicitly and establish its characteristic property. Then, we use this to derive lower bounds on the non-Gaussian gate-count to generate state $k$-designs in both the exact and approximate settings.

\subsubsection{The Haar bridge spectral distribution}
\label{sec:design_haar}

Consider the second Haar moment $\Psi_{\mathrm{Haar}}^{(2)}$. By the left–right invariance of the Haar measure, $\Psi_{\mathrm{Haar}}^{(2)}$ is invariant under the matchgate 2-twirl. Consequently, $\Psi_{\mathrm{Haar}}^{(2)}$ admits a spectral decomposition over the eigenspaces of $\Lambda$ (Lemma~\ref{lem:two_copy_ave_matchgate}):
\begin{equation}\label{eq:Haar_twirl_decomp}
    \Psi_{\mathrm{Haar}}^{(2)} \;=\; \sum_{\lambda\in\mathcal{A}} p_{\mathrm H}(\lambda)\,\frac{\Pi_\lambda}{\binom{2n}{n+\lambda/2}}.
\end{equation}
Here, $p_{\mathrm H}(\lambda)$ denotes the probability of obtaining eigenvalue $\lambda$ when performing Bell sampling on two copies of a Haar-random pure even state. We refer to $p_{\mathrm H}(\lambda)$ as the \emph{Haar bridge spectral distribution}. The goal of this section is to compute this distribution in closed form and to analyze its characteristic shape.

\begin{lem}[Haar bridge spectral distribution]\label{lem:weight_haar}
    The Haar-averaged bridge spectral distribution is
    \begin{equation}\label{eq:p_H}
        p_H(\lambda) \;=\; \frac{4\binom{2n}{n+\lambda/2}}{2^n(2^n+2)},
    \end{equation}
    supported on $\mathcal{A}$.
\end{lem}
\begin{proof}
By Schur--Weyl duality, $\Psi_{\mathrm{Haar}}^{(2)}$ is proportional to the projector onto the symmetric subspace of the parity-even sector, $P_{\mathrm{sym}}^{2^n}=(Q+Q\mathbb{S})/2$ with $Q=(I+P\otimes P)/2$ and $\mathbb{S}$ the swap operator, with normalization $\Tr[P_{\mathrm{sym}}^{2^n}]=2^n(2^n+2)/4$. Using $Q\Pi_\lambda  = \Pi_\lambda$ and $\mathbb{S}\Pi_\lambda  = \Pi_\lambda$ for $\lambda\in \mathcal{A}$, we compute $p_H(\lambda) = \Tr[P_{\mathrm{sym}}^{2^n} \Pi_\lambda]/\Tr[P_{\mathrm{sym}}^{2^n}] = \tfrac{2}{2^n(2^n+2)}\Tr[\Pi_\lambda + \mathbb{S}\,\Pi_\lambda]$,
which yields Eq.~\eqref{eq:p_H}.
\end{proof}

The closed form of $p_H$ is binomial in $\lambda$ and peaked at $\lambda=0$, with support extending to the full range $|\lambda|\leq 8 \lfloor n /4\rfloor$. More quantitatively, the distribution is well-approximated by a Gaussian of width $\sqrt{n}$:

\begin{lem}[Concentration of $p_{\mathrm H}^{(n)}$]\label{lem:p_H_concentration}
For each $n\geq 2$, let $X_n$ be a random variable distributed according to the Haar
bridge spectral distribution $p_{\mathrm H}^{(n)}$ on $n$ modes.
Then
\begin{equation}\label{eq:p_H_moments}
    \mathbb{E}[X_n] \;=\; 0, \qquad
    \mathrm{Var}(X_n) \;=\; 2n - \frac{2n(2n-1)}{2^{n-1}+1}.
\end{equation}
Moreover, the rescaled variable $X_n/\sqrt{2n}$ converges in distribution to a standard Gaussian:
\begin{equation}\label{eq:p_H_CLT}
    \frac{X_n}{\sqrt{2n}} \;\xrightarrow{d}\; Z, \qquad Z \sim \mathcal{N}(0,1) \quad \text{as } n\to\infty.
\end{equation}
\end{lem}

\begin{proof}
The mean vanishes by the symmetry $\binom{2n}{n+k}=\binom{2n}{n-k}$ of $p_{\mathrm H}^{(n)}$
about $\lambda=0$. The second moment, for $n\geq 2$, is given by
\begin{equation}
    \mathbb{E}[X_n^2] \;=\; \!\!\!\sum_{\lambda\in \mathcal{A}}\!\lambda^2\,p_{\mathrm H}^{(n)}(\lambda)
    \;=\; 2n - \frac{2n(2n-1)}{2^{n-1}+1}.
\end{equation}
where we used the expression in Eq.~\eqref{eq:p_H} and computed the resulting binomial sum over the restricted range $\mathcal{A}=8\mathbb{Z}\cap[-2n,2n]$ with a roots-of-unity filter technique~\cite[Sec. 1.2.9D]{knuth1997taocp1} and otherwise standard algebra.

A similar technique can be used to compute the characteristic function $\varphi_{X_n}(t)\coloneqq\mathbb{E}[e^{itX_n}]$, which we use for the limit theorem:
\begin{equation}\label{eq:char_fn_X_n}
    \varphi_{X_n}(t) \;=\; \frac{2^n}{2^n+2}\Bigl(\cos(t)^{2n} + \sin(t)^{2n}
    + \cos\!\bigl(t+\tfrac{\pi}{4}\bigr)^{2n} + \cos\!\bigl(t-\tfrac{\pi}{4}\bigr)^{2n}\Bigr),
\end{equation}
where the four terms correspond to the fourth roots of unity. The characteristic function of $X/\sqrt{2n}$ is $\varphi_{X/\sqrt{2n}}(t)=\varphi_{X_n}(t/\sqrt{2n})$. We have
\begin{equation}
    \lim_{n\to\infty}\varphi_{X/\sqrt{2n}}(t) \;=\; e^{-t^2/2} \;=\; \varphi_Z(t)
    \qquad \forall t\in\mathbb{R}.
\end{equation}
Pointwise convergence of characteristic functions to the characteristic function
of $Z\sim\mathcal{N}(0,1)$ implies convergence in distribution by Lévy's continuity
theorem, establishing Eq.~\eqref{eq:p_H_CLT}.
\end{proof}

Beyond convergence in distribution, the Haar bridge spectral distribution satisfies a quantitative Berry--Esseen-type bound for the Kolmogorov distance:

\begin{lem}[Berry--Esseen bound for $p_{\mathrm H}$]\label{lem:berry_esseen_pH}
There exists an absolute constant $C>0$ such that, for every $n\geq 1$,
\begin{equation}\label{eq:berry_esseen_pH}
    \sup_{x\in\mathbb{R}}\,
    \Bigl|\Pr\!\Bigl(\tfrac{X_n}{\sqrt{2n}}\leq x\Bigr) - \Phi(x)\Bigr|
    \;\leq\; \frac{C}{\sqrt{n}},
\end{equation}
where $\Phi$ denotes the standard normal CDF.
\end{lem}

\begin{proof}
We apply Esseen's smoothing inequality~\cite{feller1971} to $Y_n\coloneqq X_n/\sqrt{2n}$: for any $T>0$,
\begin{equation}\label{eq:esseen_smoothing}
    \sup_{x\in\mathbb{R}}|F_{Y_n}(x) - \Phi(x)|
    \;\leq\; \frac{1}{\pi}\!\int_{-T}^{T}\!
    \frac{|\varphi_{X_n}(t/\sqrt{2n})-e^{-t^2/2}|}{|t|}\,dt
    \;+\; \frac{24\,m_\Phi}{\pi T},
\end{equation}
where $F_{Y_n}(x)$ denotes the cumulative distribution of $Y_n$, and $m_\Phi=(2\pi)^{-1/2}$. We choose $T=\tfrac{\pi}{8}\sqrt{2n}$, so the residual contributes $\mathcal{O}(1/\sqrt n)$. Decompose the characteristic function Eq.~\eqref{eq:char_fn_X_n} as $\varphi_{X_n}(s) = A(s) + B(s)$ with
\begin{equation}
    A(s) \coloneqq \frac{2^n}{2^n+2}\bigl(\cos(s)^{2n} + 2^{1-n}\bigr),
    \qquad
    B(s) \coloneqq  \frac{2^n}{2^n+2}\Bigl[\sin(s)^{2n} \,+\, \cos\!\bigl(s+\tfrac{\pi}{4}\bigr)^{2n} + \cos\!\bigl(s-\tfrac{\pi}{4}\bigr)^{2n} \,-\, 2^{1-n}\Bigr].
\end{equation}
By construction, $A(0)=1$ and $B(0)=0$, so the integrands $|A(t/\sqrt{2n})-e^{-t^2/2}|/|t|$ and $|B(t/\sqrt{2n})|/|t|$ are each integrable at $t=0$.

We bound the two contributions in turn, beginning with the dominant term. Writing
\begin{equation}\label{eq:A_split}
    A(t/\sqrt{2n}) - e^{-t^2/2} \;=\; \frac{2^n}{2^n+2}\bigl[\cos(t/\sqrt{2n})^{2n} - e^{-t^2/2}\bigr] \;+\; \frac{2}{2^n+2}\bigl[1 - e^{-t^2/2}\bigr],
\end{equation}
the bound $\cos(s)^{2n}\leq e^{-ns^2}$ (from $\log\cos s\leq -s^2/2$ for $|s|\leq\pi/8$) gives $\cos(t/\sqrt{2n})^{2n}\leq e^{-t^2/2}$ for $|t|\leq T$. Combining with the Taylor remainder $\log\cos s + s^2/2 = -s^4/12 + \mathcal{O}(s^6)$ yields
\begin{equation}\label{eq:cos_taylor}
    \bigl|\cos(t/\sqrt{2n})^{2n} - e^{-t^2/2}\bigr| \;\leq\; K_1\, e^{-t^2/2}\,t^4/n \qquad \text{for } |t|\leq T,
\end{equation}
for an absolute constant $K_1>0$. Moreover, the second bracket in Eq.~\eqref{eq:A_split} satisfies $|1-e^{-t^2/2}|\leq \min(1,t^2/2)$. Integrating,
\begin{equation}
    \int_{-T}^{T}\!\frac{|A(t/\sqrt{2n})-e^{-t^2/2}|}{|t|}\,dt
    \;\leq\; \frac{K_1}{n}\!\int_{\mathbb R}\!|t|^3 e^{-t^2/2}dt
    \;+\; \frac{2}{2^n+2}\!\int_{-T}^{T}\!\frac{\min(1,t^2/2)}{|t|}dt
    \;=\; \mathcal{O}(1/n).
\end{equation}

For the subdominant terms, the identity $\cos(s\pm\pi/4)=(\cos s\mp\sin s)/\sqrt{2}$ together with $\sin^2+\cos^2=1$ gives
\begin{equation}\label{eq:B_identity}
    \cos\!\bigl(s+\tfrac{\pi}{4}\bigr)^{2n} + \cos\!\bigl(s-\tfrac{\pi}{4}\bigr)^{2n} - 2^{1-n}
    \;=\; 2^{-n}\bigl[(1+\sin 2s)^n + (1-\sin 2s)^n - 2\bigr],
\end{equation}
in which the bracket admits 
\begin{equation}
   (1+\sin 2s)^n + (1-\sin 2s)^n - 2 =  2\sum_{k= 1}^{\lfloor n/2 \rfloor}\binom{n}{2k}(\sin 2s)^{2k}\leq 2\sum_{k= 1}^{\infty}\frac{1}{(2k)!}(n\sin 2s)^{2k} = 2(\cosh(n|\sin 2s|)-1)\leq 8\, n^2 s^2,
\end{equation}
for $|s|\leq 1/(2n)$ (in which range $n|\sin 2s|\leq 1$). The first inequality uses $\binom{n}{2k}\leq n^{2k}/(2k)!$, while the second uses $\cosh(x)-1\leq x^2$ on $x\in[0,1]$ together with $|\sin 2s|\leq 2|s|$. Combined with $|\sin s|^{2n}\leq |s|^{2n}\leq 2^{-2(n-1)}s^2$ for $|s|\leq 1/(2n)$, we obtain
\begin{equation}\label{eq:B_local}
    |B(s)| \;\leq\; K_2\, n^2\, 2^{-n}\, s^2 \qquad \text{for } |s|\leq 1/(2n),
\end{equation}
with $K_2>0$ absolute. A uniform bound on $|s|\leq\pi/8$ follows from $|\sin s|,|\cos(s\pm\pi/4)|\leq r\coloneqq \cos(\pi/8)$, so that $|B(s)|\leq K_3\,r^{2n}$ with $K_3>0$ absolute. Splitting the integration range at $|t|=1/\sqrt{2n}$,
\begin{align}
    \int_{-T}^{T}\!\frac{|B(t/\sqrt{2n})|}{|t|}\,dt
    &\leq \frac{K_2\,n^2\,2^{-n}}{2n}\!\int_{|t|\leq 1/\sqrt{2n}}\!|t|\,dt \;+\; K_3\,r^{2n}\!\int_{1/\sqrt{2n}\leq |t|\leq T}\!\frac{dt}{|t|} \notag\\
    &= \mathcal{O}(2^{-n}) + \mathcal{O}(r^{2n}\log n).
\end{align}

Combining the two contributions with the residual in Eq.~\eqref{eq:esseen_smoothing} yields $\sup_x|F_{Y_n}(x)-\Phi(x)|=\mathcal{O}(1/\sqrt n)$, which is Eq.~\eqref{eq:berry_esseen_pH}.
\end{proof}

The rate in Eq.~\eqref{eq:berry_esseen_pH} is sharp in Kolmogorov distance, as seen by evaluating at $x=0$: the symmetry $p_{\mathrm H}(\lambda)=p_{\mathrm H}(-\lambda)$ gives $F_{Y_n}(0) = \tfrac{1}{2} + \tfrac{1}{2}p_{\mathrm H}(0)$, while $\Phi(0)=\tfrac{1}{2}$, so $|F_{Y_n}(0)-\Phi(0)| = \tfrac{1}{2}p_{\mathrm H}(0) = \Theta(1/\sqrt n)$ by Stirling's approximation applied to Eq.~\eqref{eq:p_H}. 

\begin{lem}[Tail bound of $p_{\mathrm H}$]\label{lem:p_H_subgauss}
For every $T\geq 0$ and every $n\geq 1$, the Haar bridge spectral distribution satisfies
\begin{equation}\label{eq:pH_subgauss}
    p_{\mathrm H}\bigl(|\lambda|>T\bigr) \;\leq\; \bigl(4 + 2^{2-n}\bigr)\,\exp\!\Bigl(-\tfrac{T^2}{8n}\Bigr).
\end{equation}
\end{lem}

\begin{proof}
The moment generating function $M_{X_n}(t) = \mathbb{E}[e^{tX_n}]$ is obtained by analytic continuation of the characteristic function in Eq.~\eqref{eq:char_fn_X_n} through the substitution $t\mapsto -it$:
\begin{equation}\label{eq:MGF_X_n}
    M_{X_n}(t) \;=\; \frac{2^n}{2^n+2}\Bigl(\cosh^{2n}(t) \,+\, (-1)^n \sinh^{2n}(t) \,+\, \tfrac{2}{2^n}\,\Re\!\bigl[(\cosh t + i\sinh t)^{2n}\bigr]\Bigr).
\end{equation}
The last term has modulus $|(\cosh t + i\sinh t)^{2n}| = \cosh(2t)^{n}$, since $|\cosh t+i\sinh t|^2 = \cosh(2t)$. Combined with the elementary inequality $\sinh(s)\leq\cosh(s)\leq e^{s^2/2}$,
\begin{equation}\label{eq:MGF_bound}
    M_{X_n}(t) \;\leq\; \bigl(2+2^{1-n}\bigr)\,e^{2n t^2} \qquad\text{for all } t\in\mathbb{R}.
\end{equation}

By Markov's inequality applied to $e^{tX_n}$, for any $t>0$,
\begin{equation}
    \Pr(X_n>T) \;\leq\; e^{-tT}\,M_{X_n}(t)
    \;\leq\; \bigl(2+2^{1-n}\bigr)\,e^{2n t^2 - tT}.
\end{equation}
The exponent $2n t^2 - tT$ is minimized at $t^* = T/(4n)$, where it takes the value $-T^2/(8n)$. Hence
\begin{equation}
    \Pr(X_n>T) \;\leq\; \bigl(2+2^{1-n}\bigr)\,\exp\!\Bigl(-\tfrac{T^2}{8n}\Bigr).
\end{equation}
The symmetry $p_{\mathrm H}(\lambda)=p_{\mathrm H}(-\lambda)$ (Lemma~\ref{lem:p_H_concentration}) finally yields Eq.~\eqref{eq:pH_subgauss}.
\end{proof}

\subsubsection{Non-Gaussian gate lower bounds}

We now state the central result of this section: a lower bound on the number of non-Gaussian gates any $t$-doped Gaussian protocol must contain to form state designs.

We begin with the exact-design case, where $t$-doped Gaussian protocols must  contain $t=\Omega(n)$ non-Gaussian gates.

\begin{cor}[Non-Gaussian gate lower bound for exact $2$-designs]\label{cor:exact_design}
Any family of $t$-doped Gaussian protocols that prepares an exact state $2$-design must contain at least one state requiring $t \geq \lfloor n/4 \rfloor$ non-Gaussian gates.
\end{cor}
\begin{proof}
Any state $\psi$ prepared by a $t$-doped Gaussian protocol has $\nGMax(\psi)\leq 4t$ (Corollary~\ref{cor:bound_dope}), and therefore $p_\psi(\lambda) = 0$ for $|\lambda|>8t$. An exact $2$-design satisfies $\mathbb{E}_\psi[p_\psi(\lambda)] = p_H(\lambda)$ for all $\lambda$. In particular, the state ensemble must populate $\lambda=\pm 8 \lfloor n/4 \rfloor$, where $p_H(\lambda)>0$ but $p_\psi(8 \lfloor n/4 \rfloor)=0$ whenever $t< \lfloor n/4 \rfloor$. Hence at least one state in the ensemble requires $t\geq \lfloor n/4 \rfloor$ non-Gaussian gates.
\end{proof}

Corollary~\ref{cor:exact_design} is an existence statement: there exists at least one state in the ensemble that attains maximum bridge degree, and hence requires $\Omega(n)$ non-Gaussian gates, while leaving open the possibility that such ``hard'' states comprise an exponentially small fraction of $\mathcal{E}$. The approximate-design analysis below replaces this existential conclusion with a high-probability statement for approximate state design.

\begin{thm}[Bridge degree lower bound on approximate $2$-designs]\label{thm:approx_design}
Let $\mathcal{E}$ be an $\epsilon$-approximate state $2$-design of even pure states.
\begin{enumerate}
    \item[(a)] \emph{Constant-$\epsilon$ regime.} For every $\epsilon\in(0,1)$, there exist constants $c_\epsilon, p_\epsilon>0$ (depending only on $\epsilon$) such that, for $n$ sufficiently large,
    \begin{equation}\label{eq:approx_design_bound_constant}
        \Pr_{\psi\sim\mathcal{E}}\!\Bigl[\,\nGMax(\psi) \;\geq\; c_\epsilon\sqrt{n}\,\Bigr] \;\geq\; p_\epsilon.
    \end{equation}
    \item[(b)] \emph{Decaying-$\epsilon$ regime.} There exist absolute constants $C_*,c,\epsilon_0>0$ with $\epsilon_0<1/4$ such that, for every $\epsilon\in(0,\epsilon_0]$ with $\epsilon\geq C_*/\sqrt n$,
    \begin{equation}\label{eq:approx_design_bound_decaying}
        \Pr_{\psi\sim\mathcal{E}}\!\Bigl[\,\nGMax(\psi) \;\geq\; c\sqrt{n\,\log(1/\epsilon)}\,\Bigr] \;\geq\; 2\epsilon.
    \end{equation}
\end{enumerate}
\end{thm}
\begin{proof}
Let $\Phi$ denote the standard normal CDF, and fix any $a>0$ to be chosen later.

For any even pure state $\psi$, the support of $p_\psi$ is contained in $[-2\nGMax(\psi),2\nGMax(\psi)]$, so
\begin{equation}\label{eq:supp_inclusion}
    \nGMax(\psi)\leq a\sqrt{n/2} \;\Longrightarrow\; p_\psi(\lambda)=0 \text{ for all } |\lambda|>a\sqrt{2n}.
\end{equation}
Define the tail probabilities
\begin{equation}
    r_a \;\coloneqq\; \Pr_{\psi\sim\mathcal{E}}\!\Bigl[\nGMax(\psi)>a\sqrt{n/2}\Bigr],
    \qquad
    s_a \;\coloneqq\; \!\!\!\sum_{|\lambda|>a\sqrt{2n}}\!p_{\mathcal{E}}(\lambda),
\end{equation}
so that $s_a\leq r_a$ by Eq.~\eqref{eq:supp_inclusion}, and let
$s_a^{(H)} \coloneqq \sum_{|\lambda|>a\sqrt{2n}} p_{\mathrm H}(\lambda)$
denote the Haar tail. The $\epsilon$-approximate $2$-design condition yields
\begin{equation}\label{eq:tail_bound}
    |s_a - s_a^{(H)}|\leq D\bigl(\mathcal{T}_{M_n}^{(2)}(\mathcal{E})\bigr)\leq D(\mathcal{E})\leq\epsilon,
\end{equation}
In the first inequality, we used the variational definition of trace distance applied to the two-outcome POVM $\{P_{\leq},\,I-P_{\leq}\}$ with $P_{\leq}\coloneqq\sum_{|\lambda|\leq a\sqrt{2n}}\Pi_\lambda$, as well as Eq.~\eqref{eq:design_distance_TV}; the second inequality uses monotonicity of trace distance under quantum channels.

In particular, Eq.~\eqref{eq:tail_bound} implies $s_a\geq s_a^{(H)}-\epsilon$, and hence
\begin{equation}\label{eq:r_a_chain}
    r_a \;\geq\; s_a^{(H)} - \epsilon.
\end{equation}
The two regimes differ in how the Haar tail $s_a^{(H)}$ is controlled and how $a$ is chosen.

\begin{enumerate}
    \item[(a)] \emph{Constant-$\epsilon$ regime.} For any fixed $\epsilon\in(0,1)$, choose $a_\epsilon\in(0,\Phi^{-1}(1-\epsilon/2))$, so that $2(1-\Phi(a_\epsilon))-\epsilon>0$. By the qualitative convergence $s_a^{(H)}\to 2(1-\Phi(a))$ from Lemma~\ref{lem:p_H_concentration}, the relation $r_{a_\epsilon}\geq s_{a_\epsilon}^{(H)}-\epsilon$ from Eq.~\eqref{eq:r_a_chain} yields $r_{a_\epsilon}\geq 2(1-\Phi(a_\epsilon))-\epsilon - o(1)$ as $n\to\infty$. Hence, for $n$ sufficiently large, $r_{a_\epsilon}\geq p_\epsilon\coloneqq 2(1-\Phi(a_\epsilon))-\epsilon-\delta_\epsilon>0$ for any $\delta_\epsilon\in(0,2(1-\Phi(a_\epsilon))-\epsilon)$, and the threshold is $a_\epsilon\sqrt{n/2}=c_\epsilon\sqrt n$ with $c_\epsilon\coloneqq a_\epsilon/\sqrt 2$. This is Eq.~\eqref{eq:approx_design_bound_constant}. 

\item[(b)] \emph{Decaying-$\epsilon$ regime.} We use the Berry--Esseen-type bound: by Lemma~\ref{lem:berry_esseen_pH},
\begin{equation}\label{eq:Haar_tail_BE_decaying}
    s_a^{(H)}\geq 2(1-\Phi(a))-\frac{2C}{\sqrt n}, \qquad \forall a>0,
\end{equation}
with $C$ the absolute constant of Lemma~\ref{lem:berry_esseen_pH}. Combining with Eq.~\eqref{eq:r_a_chain},
\begin{equation}\label{eq:r_a_simplified_decaying}
    r_a\geq 2(1-\Phi(a))-\epsilon-\frac{2C}{\sqrt n}\geq 2(1-\Phi(a))-2\epsilon,
\end{equation}
where the last inequality uses the assumption $\epsilon\geq C_*/\sqrt n$ with $C_*\coloneqq 2C$. Setting $a^*\coloneqq\Phi^{-1}(1-2\epsilon)$ gives $r_{a^*}\geq 2\epsilon$. It remains to bound $a^*$ in terms of $\epsilon$. The Mills-ratio lower bound $1-\Phi(a)\geq \tfrac{1}{\sqrt{2\pi}}(\tfrac{1}{a}-\tfrac{1}{a^3})e^{-a^2/2}$ for $a>0$~\cite{feller1971}, applied at $a=a^*$ with $1-\Phi(a^*)=2\epsilon$, gives $a^*\geq c'_0\sqrt{\log(1/\epsilon)}$ for an absolute $c'_0>0$, provided $\epsilon\leq\epsilon_0$ for some absolute $\epsilon_0<1/4$ (so that $a^*> 0$). Setting $c\coloneqq c'_0/\sqrt 2$, the threshold is $a^*\sqrt{n/2}\geq c\sqrt{n\log(1/\epsilon)}$, which is Eq.~\eqref{eq:approx_design_bound_decaying}.
\end{enumerate}
\end{proof}

Combining the theorem above with Corollary~\ref{cor:bound_dope} yields the following corollary.

\begin{cor}[Non-Gaussian gate lower bound for approximate $2$-designs]\label{cor:approx_design_gates}
    Let $C_*>0$ be the absolute constant from Theorem~\ref{thm:approx_design}. For any $C_*/\sqrt n \leq \epsilon < 1$, any family of $t$-doped Gaussian protocols preparing an $\epsilon$-approximate state $2$-design satisfies
    \begin{equation}
        t \;=\; \Omega\bigl(\sqrt{n\,\log(1/\epsilon)}\bigr)
    \end{equation}
    with non-negligible probability over the ensemble.
\end{cor}

Theorem~\ref{thm:approx_design} is the central structural result of this section. It shows that a size-dependent number of non-Gaussian gates --- at least $\Omega(\sqrt{n\log(1/\epsilon)})$ --- is required of any ensemble approximately reproducing Haar statistics. This stands in sharp contrast to the qubit case (where the relevant designs are on the full qubit Hilbert space): random Clifford circuits already form an exact state $3$-design and need only an $n$-independent number of non-Clifford gates to generate an approximate state $k$-design for any $k>3$~\cite{leone2026nonclifford,haferkamp2022efficient,yuzhen2026designs}. On the matchgate side, a growing number of non-Gaussian gates is required already at the state $2$-design, the first nontrivial design (the matchgate orbit forms an exact state $1$-design). Restated at the single-state level, the bridge degree itself controls randomness generated from the matchgate orbit of a state: any state $\psi$ with $\nGMax(\psi) = o(\sqrt{n\log(1/\epsilon)})$ has its matchgate orbit at trace distance $>\epsilon$ from the Haar-random state.

\section{The $\epsilon$-approximate bridge degree}
\label{sec:approx}

The bridge degree introduced in Sec.~\ref{sec:maxlambda} is, by construction, a discrete quantity that reflects the support of the bridge spectral distribution $\{p_\psi(\lambda)\}_{\lambda\in\mathcal{A}}$. While this yields a clean theory, it comes at a cost. 
Consider, for instance, an arbitrarily small non-Gaussian perturbation of a fermionic Gaussian state, $\ket{\psi_\epsilon} \propto \ket{\psi_\mathrm{FGS}} + \epsilon\ket{\phi}$ with $\epsilon\to 0$. In this case, each previously unoccupied eigenvalue sector acquires non-zero weight, and the bridge degree jumps to its maximal value irrespective of how small $\epsilon$ is. Yet, for all practical purposes, the perturbed state remains operationally indistinguishable from a Gaussian state for sufficiently small $\epsilon$.

We address this by introducing an \emph{$\epsilon$-approximate} version of the bridge degree, which tolerates a fidelity error $\epsilon$ in the target state. Crucially, this approximate variant inherits the monotonicity of its exact counterpart, albeit under the smaller class of Gaussian protocols that excludes post-selection on measurement outcomes. This monotonicity translates into concrete bounds on approximate Gaussian conversion, extending the exact-conversion bounds of Sec.~\ref{sec:nogo_conversion} to the approximate setting. The price is that the approximate bridge degree is, in general, hard to compute --- its definition involves an optimization over the full set of even pure states within fidelity $\epsilon$ of the target. We circumvent this by constructing an efficiently computable lower bound, again built from the bridge spectral distribution and therefore directly estimable from Bell-sampling data. This lower bound in turn renders the conversion bounds experimentally accessible, and in particular establishes an efficiently computable, single-shot lower bound on the non-Gaussian magic-state cost of preparing any even pure state.

\subsection{Monotonicity under Gaussian protocols}
\label{sec:approx_monotonicity}

We begin with the formal definition.

\begin{defn}[$\epsilon$-approximate bridge degree]\label{def:approx_maxlambda}
Let $\psi$ be an even pure state and $\epsilon \in [0,1]$. The \emph{$\epsilon$-approximate bridge degree} is
\begin{equation}
    \Lambda_{\mathrm{d},\epsilon}(\psi) \;\coloneqq\; \min_{\ket{\phi}\,:\, 1-F(\psi,\phi)\leq\epsilon}\, \nGMax(\phi),
\end{equation}
where $F(\psi,\phi) = \lvert\braket{\psi}{\phi}\rvert^2$ and the minimum runs over all even pure states $\ket{\phi}$.
\end{defn}

Since $\nGMax$ is integer-valued and bounded by $4\lfloor n/4\rfloor$, the minimum is attained and we may write $\min$ rather than $\inf$. The quantity $\Lambda_{\mathrm{d},\epsilon}(\psi)$ answers the question: \emph{what is the lowest bridge degree of any pure state within fidelity $1-\epsilon$ of $\psi$?}

We now show that $\Lambda_{\mathrm{d},\epsilon}$ inherits monotonicity from the exact bridge degree, under the (smaller) class of Gaussian protocols (Definition~\ref{def:gaussian_protocol}). 
We first introduce an intermediate quantity: for any even pure state $\psi$ and $\epsilon\in[0,1]$, define
\begin{equation} \label{eq:approx_cr_maxLambda}
    \Lambda^{\mathrm{cr}}_{\mathrm{d},\epsilon}(\psi) \;\coloneqq\; \min_{\tau\,:\, F(\psi,\tau)\geq 1-\epsilon}\, \nGMax^{\,\mathrm{cr}}(\tau),
\end{equation}
where, for pure $\psi$ and an arbitrary (possibly mixed) state $\tau$,
$F(\psi,\tau)=\bra{\psi}\tau\ket{\psi}$ denotes the fidelity, $\Lambda^{\mathrm{cr}}_{\mathrm{d},\epsilon}$ is the convex-roof extended bridge degree (defined in Eq.~\eqref{eq:cr_bridge_degree}), and the minimum runs over all even states $\tau$ on $n$ qubits.

The next lemma shows that the optimization in $\Lambda^{\mathrm{cr}}_{\mathrm{d},\epsilon}$ is always achieved on pure states:

\begin{lem}[Pure-state achievability]\label{lem:Lambda_nm_eq_Lambda}
For every even pure state $\psi$ and $\epsilon\in[0,1]$,
\begin{equation}
    \Lambda^{\mathrm{cr}}_{\mathrm{d},\epsilon}(\psi) \;=\; \Lambda_{\mathrm{d},\epsilon}(\psi).
\end{equation}
\end{lem}
\begin{proof}
The ``$\leq$'' direction is immediate: restricting the outer minimization to pure $\tau$ recovers the definition of $\Lambda_{\mathrm{d},\epsilon}$, since $\nGMax^{\,\mathrm{cr}}(\phi)=\nGMax(\phi)$ for pure $\phi$.

For the ``$\geq$'' direction, let $\tau^*$ attain the minimum in $\Lambda^{\mathrm{cr}}_{\mathrm{d},\epsilon}(\psi)$ and let $\tau^* = \sum_i p_i\,\ketbra{\xi_i}{\xi_i}$ be a pure-state decomposition attaining the inner minimum in $\nGMax^{\,\mathrm{cr}}(\tau^*)$, so that
\begin{equation}\label{eq:nm_bound_xi}
    \max_i\, \nGMax(\xi_i) \;=\; \nGMax^{\,\mathrm{cr}}(\tau^*) \;=\; \Lambda^{\mathrm{cr}}_{\mathrm{d},\epsilon}(\psi).
\end{equation}
Since $\psi$ is pure, the fidelity constraint expands as
\begin{equation}
    1-\epsilon \;\leq\; F(\psi,\tau^*) \;=\; \bra{\psi}\tau^*\ket{\psi} \;=\; \sum_i p_i\, F(\psi,\xi_i),
\end{equation}
so there exists $i^\star$ with $F(\psi,\xi_{i^\star}) \geq 1-\epsilon$. Hence, $\xi_{i^\star}$ is admissible in the optimization defining $\Lambda_{\mathrm{d},\epsilon}(\psi)$. Combined with Eq.~\eqref{eq:nm_bound_xi},
\begin{equation}
    \Lambda_{\mathrm{d},\epsilon}(\psi) \;\leq\; \nGMax(\xi_{i^\star}) \;\leq\; \Lambda^{\mathrm{cr}}_{\mathrm{d},\epsilon}(\psi). \qedhere
\end{equation}
\end{proof}

We can now state and prove the main result of this subsection.

\begin{thm}[Monotonicity of the $\epsilon$-approximate bridge degree]\label{thm:approx_monotone}
For every even pure state $\psi$, every $\epsilon\in[0,1]$, and every Gaussian protocol $\mathcal{E}$ such that $\mathcal{E}(\psi)$ is pure,
\begin{equation}
    \Lambda_{\mathrm{d},\epsilon}\bigl(\mathcal{E}(\psi)\bigr) \;\leq\; \Lambda_{\mathrm{d},\epsilon}(\psi).
\end{equation}
\end{thm}
\begin{proof}
Let $\tau_\epsilon$ be a pure state achieving the minimum in $\Lambda_{\mathrm{d},\epsilon}(\psi)$, so $F(\psi,\tau_\epsilon)\geq 1-\epsilon$ and $\nGMax(\tau_\epsilon)=\Lambda_{\mathrm{d},\epsilon}(\psi)$. By monotonicity of the fidelity under any quantum channel~\cite{nielsen2011quantum},
\begin{equation}
    F\bigl(\mathcal{E}(\psi),\mathcal{E}(\tau_\epsilon)\bigr) \;\geq\; F(\psi,\tau_\epsilon) \;\geq\; 1-\epsilon,
\end{equation}
and by monotonicity of $\nGMax^{\,\mathrm{cr}}$ under Gaussian protocols
\begin{equation}
    \nGMax^{\,\mathrm{cr}}\bigl(\mathcal{E}(\tau_\epsilon)\bigr) \;\leq\; \nGMax^{\,\mathrm{cr}}(\tau_\epsilon) \;=\; \nGMax(\tau_\epsilon) \;=\; \Lambda_{\mathrm{d},\epsilon}(\psi).
\end{equation}
 Thus $\mathcal{E}(\tau_\epsilon)$ is admissible in the optimization defining $\Lambda^{\mathrm{cr}}_{\mathrm{d},\epsilon}(\mathcal{E}(\psi))$. Combining with the pure-state equivalence (Lemma~\ref{lem:Lambda_nm_eq_Lambda}) applied at $\mathcal{E}(\psi)$,
\begin{equation}
    \Lambda_{\mathrm{d},\epsilon}\bigl(\mathcal{E}(\psi)\bigr) \;=\; \Lambda^{\mathrm{cr}}_{\mathrm{d},\epsilon}\bigl(\mathcal{E}(\psi)\bigr) \;\leq\; \nGMax^{\,\mathrm{cr}}\bigl(\mathcal{E}(\tau_\epsilon)\bigr) \;\leq\; \Lambda_{\mathrm{d},\epsilon}(\psi). \qedhere
\end{equation}
\end{proof}

\subsection{A computable lower bound from the bridge spectral distribution}
\label{sec:approx_eigenvalues}

Recall that the bridge degree is naturally defined in terms of the bridge spectral distribution. Here, we generalize the relation between the bridge spectral distribution and the bridge degree in the approximate setting. This also provides an efficient way to bound the approximate bridge degree, where taking an extremum over all even pure states $\ket{\phi}$ with $F(\psi,\phi)\geq 1-\epsilon$ makes it not obviously tractable, unlike in the exact setting. 
 The main idea is to simply restrict attention to states that already live in a low-eigenvalue sector of $\Lambda$, and ask how much of $\ket{\psi}^{\otimes 2}$ overlaps with such sectors.

Recall from Sec.~\ref{sec:eigenstructure} that $\mathcal{V}_\lambda = \bigoplus_{|j|\leq \lambda} V_j$ is the span of the $\Lambda$-eigenspaces with eigenvalue at most $\lambda$ in absolute value, and an even pure state $\phi$ satisfies $\nGMax(\phi)\leq \alpha$ if and only if $\ket{\phi}^{\otimes 2} \in \mathcal{V}_{2\alpha}$. In particular, the closest low-bridge degree state to $\psi$ in fidelity is closely tied to the overlap of $\ket{\psi}^{\otimes 2}$ with $\mathcal{V}_{2\alpha}$. This is exactly what the following quantity measures.

\begin{defn}[Bridge-fidelity]\label{def:bridge_fidelity}
Let $\psi$ be an even pure state, and let $k$ be an integer satisfying $0 \leq k \leq 4\lfloor n/4 \rfloor$. The bridge fidelity is
\begin{equation}
    F^\Lambda_{k}(\psi) \;\coloneqq\; \max_{\ket{\Psi} \in \mathcal{V}_{2k}:\,\braket{\Psi}{\Psi}=1} \sqrt{F(\psi^{\otimes 2},\Psi)}.
\end{equation}
\end{defn}

Notably, the bridge fidelity admits a simple closed-form expression in terms of the bridge spectral distribution.

\begin{lem}[Closed form for the bridge fidelity]\label{lem:lambda_fidelity_closed}
    For any even pure state $\psi$ and any integer $k$ satisfying $0 \leq k \leq 4\lfloor n/4 \rfloor$,
    \begin{equation}
        F^\Lambda_{k}(\psi) = \sqrt{\,\sum_{|j|\leq 2k} p_\psi(j)\,}.
    \end{equation}
\end{lem}
\begin{proof}
    For a subspace $V\subset \mathbb{C}^d$ with projector $\Pi_V$ and a unit vector $\ket{v}$, the standard identity $\max_{\ket{w}\in V:\,\braket{w}{w}=1} \lvert\braket{w}{v}\rvert^2 = \bra{v}\Pi_V \ket{v}$ holds. Applying this with $V = \mathcal{V}_{2k}$ and $\ket{v} = \ket{\psi}^{\otimes 2}$ gives
    \begin{equation}
        \max_{\ket{\Psi}\in\mathcal{V}_{2k}:\,\braket{\Psi}{\Psi}=1} \lvert\braket{\Psi}{\psi}^{\otimes 2}\rvert^2 = \sum_{|j|\leq 2k} \bra{\psi}^{\otimes 2}\Pi_j\ket{\psi}^{\otimes 2} = \sum_{|j|\leq 2k} p_\psi(j). \qedhere
    \end{equation}
\end{proof}

We now extend the bridge fidelity to the approximate setting, analogously to the bridge degree.

\begin{defn}[Approximate bridge fidelity]\label{def:approx_lambda_fidelity}
Let $\psi$ be an even pure state and $\epsilon > 0$. The $\epsilon$-approximate bridge fidelity is the smallest integer $\alpha\geq 0$ such that there exists a normalized state $\ket{\Psi}\in\mathcal{V}_{2\alpha}$ with $\sqrt{F(\psi^{\otimes 2},\Psi)} \geq 1-\epsilon$:
\begin{equation}
    \nGF_\epsilon(\psi) \;\coloneqq\; \min \bigl\{ \alpha\geq 0 \,:\, F^\Lambda_{\alpha}(\psi) \geq 1-\epsilon \bigr\}.
\end{equation}
\end{defn}

Intuitively, $\nGF_\epsilon(\psi)$ is the smallest $\alpha$ such that a fraction $(1-\epsilon)^2$ or more of the bridge spectral distribution is concentrated within $|\lambda|\leq 2\alpha$. Importantly, $\nGF_\epsilon(\psi)$ provides an efficient lower bound to the approximate bridge degree, as proven in the following lemma.

\begin{lem}[Approximate bridge fidelity lower-bounds the approximate bridge degree]\label{lem:relation_max_Lambda_fidelity}
For any even pure state $\psi$ and any $\epsilon>0$,
\begin{equation}
    \nGF_\epsilon(\psi) \;\leq\; \Lambda_{\mathrm{d},\epsilon}(\psi).
\end{equation}
\end{lem}
\begin{proof}
By definition of the approximate bridge degree, there exists an even pure state $\ket{\phi}$ with $F(\psi,\phi)\geq 1-\epsilon$ and $\nGMax(\phi)=\Lambda_{\mathrm{d},\epsilon}(\psi)$. Since $\ket{\phi}^{\otimes 2} \in \mathcal{V}_{2\Lambda_{\mathrm{d},\epsilon}(\psi)}$, we obtain
\begin{equation}
    F^\Lambda_{\Lambda_{\mathrm{d},\epsilon}(\psi)}(\psi) \;\geq\; \sqrt{F(\psi^{\otimes 2},\phi^{\otimes 2})} = F(\psi,\phi) \geq 1-\epsilon,
\end{equation}
and the conclusion follows from the definition of the $\epsilon$-approximate bridge fidelity.
\end{proof}

The lower bound of Lemma~\ref{lem:relation_max_Lambda_fidelity} naturally extends to arbitrary mixed states through the convex-roof extension $\nGMax^{\,\mathrm{cr}}$ [Eq.~\eqref{eq:cr_bridge_degree}]. This yields a \emph{fidelity witness}, which certifies the non-Gaussianity of a mixed state from its fidelity with a known pure state.

\begin{cor}[Fidelity witness for mixed states]\label{cor:fidelity_witness}
Let $\ket{\phi}$ be an even pure state and let $\rho$ be an even state with fidelity $F(\phi,\rho)\geq 1-\epsilon$. Then
\begin{equation}\label{eq:fidelity_witness}
    \nGMax^{\,\mathrm{cr}}(\rho)\;\geq\;\Lambda_{\mathrm{d},\epsilon}(\phi)\;\geq\;\nGF_\epsilon(\phi).
\end{equation}
\end{cor}

\begin{proof}
Consider the approximate convex-roof extended bridge degree $\Lambda^{\mathrm{cr}}_{\mathrm{d},\epsilon}(\phi)$ [Eq.~\eqref{eq:approx_cr_maxLambda}]. Since $F(\phi,\rho)\geq 1-\epsilon$, the state $\rho$ is feasible in the optimization of $\Lambda^{\mathrm{cr}}_{\mathrm{d},\epsilon}(\phi)$, so $\nGMax^{\,\mathrm{cr}}(\rho)\geq\Lambda^{\mathrm{cr}}_{\mathrm{d},\epsilon}(\phi)$. Lemma~\ref{lem:Lambda_nm_eq_Lambda} gives $\Lambda^{\mathrm{cr}}_{\mathrm{d},\epsilon}(\phi)=\Lambda_{\mathrm{d},\epsilon}(\phi)$, and Lemma~\ref{lem:relation_max_Lambda_fidelity} gives $\Lambda_{\mathrm{d},\epsilon}(\phi)\geq\nGF_\epsilon(\phi)$.
\end{proof}

Operationally, Corollary~\ref{cor:fidelity_witness} witnesses the mixed-state non-Gaussianity $\nGMax^{\,\mathrm{cr}}(\rho)$ from two ingredients: the fidelity $F(\phi,\rho)$ to the target $\ket{\phi}$, which is directly measurable~\cite{huang2020predicting}, and the approximate bridge fidelity $\nGF_\epsilon(\phi)$, computed from the classical description of $\phi$. An experimentalist preparing a state intended to be $\ket{\phi}$ and measuring $F(\phi,\rho)\geq 1-\epsilon$ may thus certify $\nGMax^{\,\mathrm{cr}}(\rho)\geq\nGF_\epsilon(\phi)$ for the actual laboratory state. This strategy of estimating the fidelity to a classically described target and using it to certify a resource of the laboratory state has been demonstrated experimentally for entanglement in Ref.~\cite{shaw2024benchmarking}.

We next describe how the approximate bridge fidelity $\nGF_\epsilon$ can be estimated from Bell-sampling data on a pure state $\psi$. 
Concretely, $\nGF_\epsilon(\psi)$ is the smallest $\alpha\in\mathcal{A}/2$ for which the bridge fidelity $F^\Lambda_\alpha(\psi)^2$ crosses the threshold $(1-\epsilon)^2$, and the sample complexity of locating this crossing point from Bell-sampling data is governed by the gap
\begin{equation} \label{eq:Delta_def}
    \Delta \;\coloneqq\; \min\!\Bigl\{\, F^\Lambda_{\nGF_\epsilon(\psi)}(\psi)^2 - (1-\epsilon)^2\,,\; (1-\epsilon)^2 - F^\Lambda_{\nGF_\epsilon(\psi)-1}(\psi)^2 \,\Bigr\}
\end{equation}
between the bridge fidelity and the threshold at the two adjacent values of $\alpha$. Since $\Delta$ is not known \emph{a priori}, we adopt an adaptive stopping rule based on the Dvoretzky--Kiefer--Wolfowitz (DKW) inequality~\cite{Dvoretzky1956,Massart1990}, summarized as Algorithm~\ref{alg:algo_Lambda_fid}. The following lemma certifies its sample complexity.

\begin{lem}[Sample complexity of $\nGF_\epsilon$]\label{lem:sample_complexity_approx}
Let $\psi$ be an even pure state with $\nGF_\epsilon(\psi)=\alpha$. There is a Bell-sampling-based algorithm that outputs $\widehat\alpha = \nGF_\epsilon(\psi)$ with probability at least $1-\delta$ using
\begin{equation}\label{eq:adaptive_F_complexity}
    N \;=\; \mathcal{O}\!\left(\frac{1}{\Delta^2}\bigl(\log(1/\Delta)+ \log(1/\delta)\bigr)\right)
\end{equation}
Bell samples, where $\Delta$ is defined in Eq.~\eqref{eq:Delta_def}. The algorithm requires no a priori knowledge of $\Delta$.
\end{lem}

\begin{proof}
Write $F(\alpha)\coloneqq F^\Lambda_\alpha(\psi)^2 = \sum_{|\lambda|\leq 2\alpha}p_\psi(\lambda)$ for the squared bridge fidelity, and let $\widehat F_k(\alpha') \coloneqq k^{-1}\sum_{i=1}^k \mathbf{1}\{|\lambda_i|\leq 2\alpha'\}$ for $\alpha'\in\mathcal{A}/2$ denote the algorithm's empirical CDF estimate of $F(\alpha')$, with confidence half-width
\begin{equation}\label{eq:dkw_hw}
    \eta_k \;\coloneqq\; \sqrt{\log(2/\delta_k)/(2k)}, \qquad \delta_k \;\coloneqq\; 6\delta/(\pi^2 k^2),
\end{equation}
and stopping criterion: stop at the first $k$ for which some $\alpha'\in\mathcal{A}/2$ satisfies
\begin{equation} \label{eq:stopping_conditions}
    \widehat F_k(\alpha') - \eta_k \geq (1-\epsilon)^2 \quad \mathrm{and} \quad \widehat F_k(\alpha'-1) + \eta_k < (1-\epsilon)^2,
\end{equation}
returning $\widehat\alpha \coloneqq \alpha'$. The DKW inequality states that for any fixed $k$,
\begin{equation} \label{eq:dkw_adaptive}
    \Pr(\max_{\alpha'} |\widehat F_k(\alpha') - F(\alpha')| > \eta_k) \leq \delta_k.
\end{equation}
 With $\delta_k$ chosen so that $\sum_{k\geq 1}\delta_k = \delta$, the union bound over $k$ ensures that Eq.~\eqref{eq:dkw_adaptive} holds for all $k$, with probability at least $1-\delta$. Call this event $E$. Suppose now that the algorithm stops at some round $k$ with output $\alpha'$, so the stopping conditions Eq.~\eqref{eq:stopping_conditions} hold; on $E$ we have $\widehat F_k(\alpha')\in[F(\alpha')-\eta_k,\,F(\alpha')+\eta_k]$, and Eq.~\eqref{eq:stopping_conditions} therefore implies
\begin{equation}
    F(\alpha') \,\geq\, \widehat F_k(\alpha') - \eta_k \,\geq\, (1-\epsilon)^2 \quad\text{and}\quad F(\alpha'-1) \,\leq\, \widehat F_k(\alpha'-1) + \eta_k \,<\, (1-\epsilon)^2.
\end{equation}
Since $F$ is non-decreasing, $\alpha'$ is the smallest value with $F(\alpha')\geq (1-\epsilon)^2$, which by definition equals $\nGF_\epsilon(\psi) = \alpha$. The algorithm therefore returns the correct value whenever it stops. To bound the stopping time, observe that Eq.~\eqref{eq:Delta_def} gives $F(\alpha)\geq (1-\epsilon)^2 + \Delta$ and $F(\alpha-1)\leq (1-\epsilon)^2 - \Delta$; on $E$ this implies $\widehat F_k(\alpha)\geq F(\alpha) - \eta_k$ and $\widehat F_k(\alpha-1)\leq F(\alpha-1) + \eta_k$, so the stopping conditions at $\alpha' = \alpha$ are satisfied as soon as $\eta_k\leq \Delta/2$. Substituting Eq.~\eqref{eq:dkw_hw} and solving self-consistently for the smallest such $k$ yields
\begin{equation}
    k \;=\; \mathcal{O}\!\left(\frac{1}{\Delta^2}\bigl(\log(1/\Delta)+ \log(1/\delta)\bigr)\right),
\end{equation}
which is Eq.~\eqref{eq:adaptive_F_complexity}.
\end{proof}

\begin{algorithm}[H]
\caption{Adaptive estimation of $\nGF_\epsilon(\psi)$}
\label{alg:algo_Lambda_fid}
\begin{flushleft}
\textbf{Input:} threshold $\epsilon\in(0,1)$, failure probability $\delta\in(0,1)$. \\
\textbf{Output:} $\widehat\alpha = \nGF_\epsilon(\psi)$ (with probability at least $1-\delta$).
\end{flushleft}
\begin{algorithmic}[1]
\State $N\gets 0$
\Repeat
    \State $N\gets N+1$
    \State Run Algorithm~\ref{alg:bell_sampling} once to obtain $\lambda_N\sim p_\psi$.
    \State Update the empirical CDF $\widehat F(\alpha) \gets N^{-1}\sum_{i=1}^N \mathbf{1}\{|\lambda_i|\leq 2\alpha\}$ for all $\alpha\in\mathcal{A}/2$.
    \State $\delta'_N \gets 6\delta/(\pi^2 N^2)$;\quad $\eta_N \gets \sqrt{\log(2/\delta'_N)/(2N)}$.
\Until{$\exists\, \alpha'\in\mathcal{A}/2$ such that $\widehat F(\alpha') - \eta_N \geq (1-\epsilon)^2$ \textbf{and} $\widehat F(\alpha'-1) + \eta_N < (1-\epsilon)^2$}
\State \textbf{return} $\widehat\alpha \gets \alpha'$.
\end{algorithmic}
\end{algorithm}

\subsection{Approximate Gaussian conversion}
\label{sec:approx_conversion}

We have discussed in Sec.~\ref{sec:consequences} how the monotonicity of the bridge degree translates into concrete bounds on Gaussian conversion.
In this section, we show that the monotonicity of the approximate bridge degree allows us to extend these bounds to the case of approximate Gaussian conversion.

Before stating the main theorem, we record two elementary perturbation bounds relating the trace distance and the fidelity that will be used in the proof.

\begin{lem}\label{lem:diffFid2}
For a density matrix $\rho$ and a pure state $\psi$, we have
\begin{equation}\label{eq:diffF2_helper}
    1-F(\psi,\rho)\leq D(\psi,\rho).
\end{equation}
\end{lem}
\begin{proof}
    We make use of the variational definition of the trace distance:
    \begin{equation} \label{eq:trace_distance_var}
        D(\rho,\sigma) = \max_P \Tr(P(\rho-\sigma)),
    \end{equation}
    where the maximization is over all projectors $P$. Taking $P = \psi$, we have
    \begin{equation}
    \begin{split}
        D(\psi,\rho) &\geq \Tr(\psi(\psi-\rho))  \\
        &= 1 - \bra{\psi} \rho \ket{\psi} \\
        &= 1 - F(\psi,\rho),
        \end{split}
    \end{equation}
    which concludes the proof.
\end{proof}

\begin{lem}\label{lem:diffFid}
For density matrices $\rho$, $\sigma$ and a pure state $\psi$, we have
\begin{equation}\label{eq:diffF}
    \left|F(\rho,\psi)-F(\sigma,\psi)\right|\leq D(\rho,\sigma),
\end{equation}
and
\begin{equation}\label{eq:diffF2_b}
    \left|F(\psi,\sigma)-F(\rho,\sigma)\right|\leq \sqrt{D(\psi,\rho)}.
\end{equation}
\end{lem}
\begin{proof}
We again make use of the variational definition of the trace distance in Eq.~\eqref{eq:trace_distance_var}, setting $P=\psi$,
    \begin{equation}
    \begin{split}
        D(\rho,\sigma) &\geq \Tr(\psi(\rho-\sigma))  \\
        &= \bra{\psi} \rho \ket{\psi} - \bra{\psi} \sigma \ket{\psi} \\
        &= F(\rho,\psi) - F(\sigma,\psi).
        \end{split}
    \end{equation}
    Then, exchanging $\rho$ and $\sigma$, we arrive at
    \begin{equation}
        D(\rho,\sigma) \geq  \lvert F(\rho,\psi)-F(\sigma,\psi)\rvert.
    \end{equation}

    To show the second inequality, we have~\cite{rastegin2003lower}
    \begin{equation}
        \left|F(\psi,\sigma)-F(\rho,\sigma)\right| \leq \sqrt{1-F(\psi,\rho)},
    \end{equation}
    which, combined with Lemma~\ref{lem:diffFid2}, concludes the proof.
\end{proof}

\begin{thm}[Gaussian-conversion bounds from the approximate bridge degree]\label{thm:approx_conversion}
Let $\psi,\phi$ be even pure states. If there exists a Gaussian protocol $\mathcal{E}$ such that $\mathcal{E}(\psi)$ is $\delta$-close in trace distance to $\phi$, then for all $\epsilon>0$:
\begin{align}
    \Lambda_{\mathrm{d},\sqrt{\epsilon}+\delta}(\phi) &\leq \Lambda_{\mathrm{d},\epsilon}(\psi), \label{eq:approx_bound_1}\\
    \Lambda_{\mathrm{d},\epsilon+\sqrt{\delta}}(\phi) &\leq \Lambda_{\mathrm{d},\epsilon}(\psi). \label{eq:approx_bound_2}
\end{align}
Moreover, if post-selected Gaussian protocols are allowed and $\mathcal{E}(\psi)$ is $\delta$-close to $\phi$ with any nonzero probability, then
\begin{equation}\label{eq:approx_bound_3}
    \Lambda_{\mathrm{d},\delta}(\phi) \leq \nGMax(\psi).
\end{equation}
\end{thm}
\begin{proof}
For a given $\epsilon$, let $\tau_\epsilon$ be an even pure state such that $F(\psi,\tau_\epsilon)\geq 1- \epsilon$ and $\nGMax(\tau_\epsilon)=\Lambda_{\mathrm{d},\epsilon}(\psi)$. We have
\begin{equation} \label{eq:fid_a}
    F\left(\phi, \mathcal{E}(\tau_\epsilon)\right) \geq F(\phi,\mathcal{E}(\psi)) - D\left(\mathcal{E}(\psi), \mathcal{E}(\tau_\epsilon)\right) \geq F(\phi,\mathcal{E}(\psi)) - D(\psi, \tau_\epsilon) \geq 1- \delta - \sqrt{\epsilon}.
\end{equation}
We used Lemma~\ref{lem:diffFid} in the first inequality. In the second inequality, we used that the trace distance does not increase under quantum channels. In the third inequality, we used
\begin{equation}
    D(\psi,\tau_\epsilon) = \sqrt{1-F(\psi,\tau_\epsilon)} \leq \sqrt{\epsilon},
\end{equation}
 since both $\psi$ and $\tau_\epsilon$ are pure states, and
 \begin{equation}
     F(\phi,\mathcal{E}(\psi)) \geq 1-D(\phi,\mathcal{E}(\psi))\geq 1-\delta,
 \end{equation}
  where we used Lemma~\ref{lem:diffFid2}. Eq.~\eqref{eq:fid_a} implies
\begin{equation} \label{eq:fid_b}
    \Lambda_{\mathrm{d},\sqrt{\epsilon}+\delta}(\phi) =\Lambda^{\mathrm{cr}}_{\mathrm{d},\sqrt{\epsilon}+\delta}(\phi)\leq \nGMax^{\,\mathrm{cr}}(\mathcal{E}(\tau_\epsilon)) \leq \nGMax^{\,\mathrm{cr}}(\tau_\epsilon) =  \Lambda_{\mathrm{d},\epsilon}(\psi),
\end{equation}
where we used the pure-state equivalence (Lemma~\ref{lem:Lambda_nm_eq_Lambda}) at $\phi$ in the first equality and the definition of the approximate bridge degree in the first inequality. In the second inequality, we used that the convex-roof extension $\nGMax^{\,\mathrm{cr}}$ is non-increasing under Gaussian protocols. The final equality uses $\nGMax^{\,\mathrm{cr}}(\tau_\epsilon) = \nGMax(\tau_\epsilon)$ on the pure state $\tau_\epsilon$. This establishes Eq.~\eqref{eq:approx_bound_1}.

Now, for Eq.~\eqref{eq:approx_bound_2},
\begin{equation}
    F\left(\phi, \mathcal{E}(\tau_\epsilon)\right) \geq F\left(\mathcal{E}(\psi), \mathcal{E}(\tau_\epsilon)\right) - \sqrt{D(\phi,\mathcal{E}(\psi))} \geq F(\psi, \tau_\epsilon) - \sqrt{D(\phi,\mathcal{E}(\psi))} \geq 1- \epsilon - \sqrt{\delta}.
\end{equation}
We used Lemma~\ref{lem:diffFid} in the first inequality. In the second inequality, we used that the fidelity does not decrease under quantum channels. Thus, using the same steps as Eq.~\eqref{eq:fid_b} shows Eq.~\eqref{eq:approx_bound_2}.

Finally, setting $\epsilon=0$ in Eq.~\eqref{eq:approx_bound_1} gives Eq.~\eqref{eq:approx_bound_3}, with the same proof. As the monotonicity of trace distance under quantum channels is not required in this case, the Gaussian protocol $\mathcal{E}$ also allows post-selection. This concludes the proof.
\end{proof}

As an immediate application, we obtain lower bounds on the $\epsilon$-single-shot cost of preparing a non-Gaussian state from the canonical magic state $\ket{M}$ of Eq.~\eqref{eq:magic_swap}. Defining the $\epsilon$-single-shot distillation and cost resources~\cite{Gour2025}
\begin{align}
\mathrm{Distill}^\epsilon(\rho) &= \max \bigl\{ m\in\mathbb{N} :\, \exists \mathcal{E}\in\G,\, D(\mathcal{E}(\rho),M^{\otimes m})\leq\epsilon \bigr\}, \\
\mathrm{Cost}^\epsilon(\rho) &= \min \bigl\{ m\in\mathbb{N} :\, \exists \mathcal{E}\in\G,\, D(\mathcal{E}(M^{\otimes m}),\rho)\leq\epsilon \bigr\},
\end{align}
we have the following.

\begin{thm}[Approximate cost lower bound]\label{thm:lower_bound_cost}
For any even pure state $\psi$ and any $\epsilon\geq 0$,
\begin{equation} \label{eq:lower_bound_cost}
    \mathrm{Cost}^\epsilon(\psi) \;\geq\; \frac{\Lambda_{\mathrm{d},\epsilon}(\psi)}{4} \;\geq\; \frac{\nGF_\epsilon(\psi)}{4}.
\end{equation}
\end{thm}
\begin{proof}
If a post-selected Gaussian protocol $\mathcal{E}$ maps $M^{\otimes m}$ to a state $\epsilon$-close to $\psi$ in trace distance, then by Theorem~\ref{thm:approx_conversion} and additivity under tensor product (Lemma~\ref{lem:additive}),
\begin{equation}
    \Lambda_{\mathrm{d},\epsilon}(\psi) \leq \nGMax(M^{\otimes m}) = m\,\nGMax(M) = 4m,
\end{equation}
where we used $\nGMax(M)=4$ (Lemma~\ref{lem:max_Lambda_4_qubit}). This directly establishes the first inequality in Eq.~\eqref{eq:lower_bound_cost}. The second inequality follows from Lemma~\ref{lem:relation_max_Lambda_fidelity}.
\end{proof}

Theorem~\ref{thm:lower_bound_cost} has a direct operational significance: it says that the number of magic states required to prepare $\psi$ within $\epsilon$ in trace distance is at least $\nGF_\epsilon(\psi)/4$. Crucially, $\nGF_\epsilon(\psi)/4$ is \emph{efficiently estimable} from Bell sampling data via Algorithm~\ref{alg:algo_Lambda_fid}.

\section{Mixed-state extension}
\label{sec:mixed}

The bridge degree introduced in Sec.~\ref{sec:maxlambda} and its approximate variant of Sec.~\ref{sec:approx} are defined on even \emph{pure} states. In this section we extend the framework to mixed states. First, we recast the mixed-state Gaussianity condition of Ref.~\cite{bravyi2005lagrangian} --- originally a commutator condition $[\Lambda,\rho^{\otimes 2}]=0$  --- as a kernel condition on the doubled Hilbert space, by introducing the operator
\begin{equation}\label{eq:def_mixed_Lambda}
    \tilde{\Lambda} \;\coloneqq\; \Lambda\otimes I \;-\; I\otimes\Lambda,
\end{equation}
and derive the analogues of the pure-state spectral and symmetry results for it. Then, we define the \emph{mixed-state bridge degree} as the largest eigenvalue of $\tilde{\Lambda}$ whose eigenspace is occupied by the two-copy Choi state, and we prove that it is non-increasing under the smaller class of post-selected Gaussian operations (Definition~\ref{def:gaussian_operations}).

\subsection{The bridge operator on the doubled Hilbert space}
\label{sec:choi_setup}

We work with the Choi representation of density matrices. Let $\ket{\Omega} = \sum_{j=1}^{2^n} \ket{j}\otimes\ket{j}$ be the unnormalized maximally entangled state on two copies of $\mathcal{H}_n$. For any operator $A\in\mathcal{B}(\mathcal{H}_n)$, we define its Choi state $\kket{A} \coloneqq (A\otimes I)\ket{\Omega}$. The inner product $\bbrakket{A}{B} = \Tr(A^\dagger B)$ is the Hilbert--Schmidt inner product on operators, and we record the standard \emph{ABC-rule}:
\begin{equation} \label{eq:ABC_rule}
    \kket{ABC} = (A\otimes C^T)\kket{B}.
\end{equation}

The first observation is that the mixed-state Gaussianity condition of Eq.~\eqref{eq:Lambda_condition_mixed} is equivalent to a kernel condition on the Choi space:

\begin{lem}[Mixed-state Gaussianity via Choi kernel]\label{lem:mixed_gauss_choi}
    An even state $\rho$ is an FGS if and only if
    \begin{equation}\label{eq:mixed_gauss_choi}
        \tilde{\Lambda}\, \kket{\rho^{\otimes 2}} = 0.
    \end{equation}
\end{lem}
\begin{proof}
By the ABC-rule, $\kket{\Lambda\rho^{\otimes 2} - \rho^{\otimes 2}\Lambda} = (\Lambda\otimes I - I\otimes \Lambda^T)\kket{\rho^{\otimes 2}}$. Since $\Lambda^T = \Lambda$ (as $\Lambda$ is real symmetric in the computational basis), the commutator condition $[\Lambda,\rho^{\otimes 2}]=0$ is equivalent to $\tilde{\Lambda}\kket{\rho^{\otimes 2}} = 0$.
\end{proof}

The invariance of $\tilde{\Lambda}$ under matchgates, which will be essential in what follows, follows from the analogous property of $\Lambda$:

\begin{lem}[Matchgate invariance of $\tilde\Lambda$]\label{lem:tLambda_matchgate}
    For any matchgate $U\in M_n$, let $\mathcal{U} \coloneqq U\otimes U \otimes U^* \otimes U^*$. Then
    \begin{equation}
        [\tilde{\Lambda},\, \mathcal{U}] = 0.
    \end{equation}
\end{lem}
\begin{proof}
The ABC-rule gives $\kket{U\rho U^\dagger}=U \otimes U^*\kket{\rho} $, so the action of $\mathcal{U}$ on the two-copy Choi space corresponds to matchgate conjugation on the two copies of $\rho$. The claim then follows from $[\Lambda,U^{\otimes 2}]=0$.
\end{proof}

We now turn to the spectrum of $\tilde{\Lambda}$.  Since $\tilde{\Lambda}$ is a difference of two copies of $\Lambda$ (each acting on a different factor of the Choi space), its spectrum is easily read off from those of $\Lambda$. The spectrum of $\tilde{\Lambda}$ ranges over $\lambda \in \{-4n,-4n+2,\ldots,4n\}$; note in particular that the range of allowed eigenvalues is now $4n$ rather than $2n$, reflecting the doubled Hilbert space on which $\tilde{\Lambda}$ acts. We denote by $\tilde{V}_\lambda$ the eigenspace with eigenvalue $\lambda$ and by $\tilde{\Pi}_\lambda$ the corresponding projector. In analogy with the pure-state case, the eigenspaces of $\tilde{\Lambda}$ admit an explicit basis.

\begin{lem}[Eigenstructure of $\tilde{\Lambda}$]\label{lem:eigenbasis_bell_mixed}
    The eigenstates of $\tilde{\Lambda}$ are tensor products of two Bell-state vectors,
    \begin{equation}
        \ket{\Psi_{\boldsymbol{\mu}^{(1)},\boldsymbol{\mu}^{(2)}}} \coloneqq \ket{\Psi_{\boldsymbol{\mu}^{(1)}}} \otimes \ket{\Psi_{\boldsymbol{\mu}^{(2)}}},
    \end{equation}
    with eigenvalue
    \begin{equation}
        \lambda = \sum_{j=1}^{2n}\mu_j^{(1)} - \sum_{j=1}^{2n}\mu_j^{(2)}.
    \end{equation}
    Moreover, the states with fixed $\lambda$ span the corresponding eigenspace $\tilde{V}_\lambda$.
\end{lem}
\begin{proof}
Using Lemma~\ref{lem:eigenbasis_bell} and the definition of $\tilde{\Lambda}$,
\begin{equation}
\tilde{\Lambda}\,\ket{\Psi_{\boldsymbol{\mu}^{(1)}}}\otimes\ket{\Psi_{\boldsymbol{\mu}^{(2)}}} = \Bigl(\sum_j\mu_j^{(1)} - \sum_j\mu_j^{(2)}\Bigr)\ket{\Psi_{\boldsymbol{\mu}^{(1)}}}\otimes\ket{\Psi_{\boldsymbol{\mu}^{(2)}}}.
\end{equation}
Hence, $\ket{\Psi_{\boldsymbol{\mu}^{(1)}}}\otimes\ket{\Psi_{\boldsymbol{\mu}^{(2)}}}$ is an eigenstate of $\Tilde{\Lambda}$
with eigenvalue $\lambda=\sum_j\mu_j^{(1)}-\sum_j\mu_j^{(2)}$, as claimed.
Since the states $\{\ket{\Psi_{\boldsymbol{\mu}^{(1)},\boldsymbol{\mu}^{(2)}}}\}$ form an orthonormal basis of $\mathcal{H}_n^{\otimes 4}$, the subset with fixed $\lambda$ spans the eigenspace.
\end{proof}

The central observation now is that the symmetries that constrained the pure-state bridge spectral distribution to $\mathcal{A}$ have natural analogues in the mixed-state setting. Define the mixed-state Bell-state weights
\begin{equation} \label{eq:mixed_bell_weights}
    \tilde{p}_\rho(\boldsymbol{\mu}^{(1)},\boldsymbol{\mu}^{(2)}) \coloneqq |\langle\Psi_{\boldsymbol{\mu}^{(1)},\boldsymbol{\mu}^{(2)}}\kket{\rho^{\otimes 2}}|^2 = \left\lvert \bra{\Psi_{\boldsymbol{\mu}^{(1)}}} \rho^{\otimes 2} \ket{\Psi_{\boldsymbol{\mu}^{(2)}}} \right\rvert^2 .
\end{equation}
Moreover, the mixed-state bridge spectral distribution is
\begin{equation} \label{eq:mixed_bridge_spectral}
    \tilde{p}_\rho(\lambda) \coloneqq \bbra{\rho^{\otimes 2}} \tilde{\Pi}_\lambda \kket{\rho^{\otimes 2}} 
    =  \sum_{\boldsymbol{\mu}:\sum_j\mu_j^{(1)} - \mu_j^{(2)}=\lambda} \tilde{p}_\rho(\boldsymbol{\mu}^{(1)},\boldsymbol{\mu}^{(2)}) .
\end{equation}
The weight $\tilde{p}_\rho(\lambda)$ is the mixed-state analogue of $p_\psi(\lambda)$, and will play the same role in the mixed-state theory.

\begin{lem}[Parity symmetry of the mixed-state Bell-state weights]\label{lem:bell_weight_mixed}
For any even state $\rho$,
\begin{equation}
    \tilde{p}_\rho(\boldsymbol{\mu}^{(1)},\boldsymbol{\mu}^{(2)}) = \tilde{p}_\rho(-\boldsymbol{\mu}^{(1)},-\boldsymbol{\mu}^{(2)}).
\end{equation}
\end{lem}
\begin{proof}
For an even state $\rho$, the ABC-rule gives 
\begin{equation} \label{eq:parity_mixed}
        (P^{(1)}\otimes P^{(1)})\kket{\rho^{\otimes2}}= \kket{P \rho P \otimes \rho}   =\kket{\rho^{\otimes 2}},
    \end{equation}
     since $P\rho P = \rho$ for an even state. The proof now follows the steps of the proof of Lemma~\ref{lem:bell_weight}, combined with Lemma~\ref{lem:inverting_bell} applied to each Bell-state $\ket{\Psi_{\boldsymbol{\mu}^{(1)}}}$ and $\ket{\Psi_{\boldsymbol{\mu}^{(2)}}}$.
\end{proof}

\begin{cor}\label{cor:pm_Lambda_mixed}
For any even state $\rho$, $\tilde{p}_\rho(\lambda) = \tilde{p}_\rho(-\lambda)$.
\end{cor}

The exchange symmetry has an analogue here as well, this time coming from the symmetry of $\kket{\rho^{\otimes 2}}$ under swapping the two copies of the Choi state. Unlike the pure-state case, however, the exchange-symmetry sector is not determined by $\lambda\bmod 8$~\footnote{This can be seen already for one qubit, where $\ket{\sigma_{01}}\ket{\sigma_{00}}$ and $\ket{\sigma_{00}}\ket{\sigma_{11}}$ are both in $\tilde{V}_2$, yet the former is symmetric while the latter is antisymmetric under exchange.}. We therefore impose the exchange and parity symmetries directly on the Bell-state factors to restrict the possible eigenvalues, deviating from the pure-state strategy of Lemma~\ref{lem:allowed_eig_Lambda}, which simply combines two separate results from the parity and exchange symmetries.

\begin{lem}[Eigenvalue restriction (mixed-state)]\label{lem:allowed_eig_Lambda_mixed}
    For any even state $\rho$, $\tilde{p}_\rho(\lambda)\neq 0$ implies $\lambda\in 8\mathbb{Z}$.
\end{lem}
\begin{proof}
Consider the expansion of $\kket{\rho^{\otimes 2}}$ in the eigenbasis of Bell states (Lemma~\ref{lem:eigenbasis_bell_mixed}). The state $\kket{\rho^{\otimes 2}}$ is symmetric under the swap of the two copies of the Choi state, $\mathbb{S}\otimes\mathbb{S}\kket{\rho^{\otimes 2}} = \kket{\rho^{\otimes 2}}$. This symmetry condition forces the total number of $\ket{\sigma_{11}}$ factors to be even. Furthermore, Lemma~\ref{lem:bell_weight_mixed} forces the number of $\ket{\sigma_{01}}$ factors to be even as well.

We can now apply the run-length counting argument of Lemma~\ref{lem:symmetry_eigenspace_Lambda} separately to each of the two factors. 
We encode each eigenstate by two binary sequences  
\begin{equation}
x_1,\dots,x_m \in \{0,1\}, \qquad y_1,\dots,y_{m'} \in \{0,1\}
\end{equation}
where $x_j=0$ corresponds to $\ket{\sigma_{01}}$ and $x_j=1$ to $\ket{\sigma_{11}}$.   

    Let $m_0,m_1$ be the numbers of $0$'s and $1$'s in the $x$-sequence, and $m'_0,m'_1$ the corresponding numbers for the $y$–sequence. Define
    \begin{equation}
        T = \sum_{j=1}^m (-1)^{j+1} x_j, \qquad T' = \sum_{j=1}^{m'} (-1)^{j+1} y_j.
    \end{equation}

Let $A$ be the sum of $(-1)^{s+1}$ over odd-length $0$-runs and let $B$ be the sum of $(-1)^{s+1}$ over odd-length $1$-runs in the $x$-sequence, and $A',B'$ the analogous sums for the $y$–sequence. Then the total eigenvalue is
\begin{equation}
    \lambda=2\alpha, \quad \alpha = A - B - A' + B'.
\end{equation}

As in Lemma~\ref{lem:eigenbasis_bell_mixed}, we have
\begin{equation}
    T=B, \qquad T'=B',
\end{equation}
and 
\begin{equation}
    \frac{1-(-1)^{m}}{2} - T = A, \qquad \frac{1-(-1)^{m'}}{2} - T' = A'. 
\end{equation}
Therefore,
\begin{equation}
    \alpha = A - B - A' + B' = \frac{1-(-1)^{m}}{2} - \frac{1-(-1)^{m'}}{2}  - 2(T - T')
\end{equation}
Since $m_1+m_1'$ is even, we have 
\begin{equation}
    T-T' \equiv m_1 - m_1' \pmod{2} \equiv 0 \pmod{2}.
\end{equation}
Furthermore, since $m_0+m_0'$ is also even, we get $m+m'$ even, and therefore
\begin{equation}
    \frac{1-(-1)^m}{2} - \frac{1-(-1)^{m'}}{2}
    =0
\end{equation}
 Combining these, we obtain $\alpha \equiv 0 \pmod{4}$, and therefore $\lambda = 2\alpha \equiv 0 \pmod{8}$, as claimed.

\end{proof}

Furthermore, the mixed-state bridge spectral distribution is invariant under matchgate unitaries, in exact analogy with the pure-state case.

\begin{lem}[Matchgate invariance of the mixed-state bridge spectral distribution]\label{lem:invariance_weight_mixed}
For any matchgate $U$ and any even state $\rho$, let $\sigma = U\rho U^\dagger$. Then $\tilde{p}_\sigma(\lambda) = \tilde{p}_\rho(\lambda)$ for all $\lambda$.
\end{lem}
\begin{proof}
By Lemma~\ref{lem:tLambda_matchgate}, $\mathcal{U}$ commutes with $\tilde{\Lambda}$, hence with each spectral projector $\tilde{\Pi}_\lambda$. Since $\kket{\sigma} = U \otimes U^*\kket{\rho}$, we have
\begin{equation}
    \tilde{p}_\sigma(\lambda) = \bbra{\sigma^{\otimes 2}}\tilde{\Pi}_\lambda \kket{\sigma^{\otimes 2}} = \bbra{\rho^{\otimes 2}}(\mathcal{U}^\dagger)\tilde{\Pi}_\lambda \mathcal{U}\kket{\rho^{\otimes 2}} = \bbra{\rho^{\otimes 2}}\tilde{\Pi}_\lambda \kket{\rho^{\otimes 2}} = \tilde{p}_\rho(\lambda). \qedhere
\end{equation}
\end{proof}

Finally, we remark on the experimental accessibility of the mixed-state bridge spectral distribution. By construction, $\tilde p_\rho$ is the bridge spectral distribution of the pure Choi state $\kket{\rho}$ on a doubled Hilbert space; whenever $\kket{\rho}$ can be experimentally prepared, Algorithm~\ref{alg:bell_sampling} extends directly and yields $\tilde p_\rho$ from two copies of $\kket{\rho}$. Preparing $\kket{\rho}$ from copies of $\rho$ alone is, however, generally not possible, and consequently the mixed-state bridge spectral distribution $\tilde p_\rho$ is not directly accessible via Bell sampling on $\rho$. Nevertheless, in many experimental settings, the mixedness of $\rho$ arises from noise acting on an underlying pure state. In this case, error-mitigation techniques can be used to reconstruct the bridge spectral distribution of the underlying pure state, rather than the mixed-state distribution $\tilde p_\rho$. In Appendix~\ref{app:depol}, we present such a protocol for global depolarizing noise and show how the reconstructed distribution can be used to certify the bridge degree of the ideal state.

\subsection{Definition and basic properties}
\label{sec:mixed_maxlambda_def}

With the mixed-state bridge spectral distribution in hand, we can define the mixed-state bridge degree in complete analogy with its pure-state counterpart.

\begin{defn}[Mixed-state bridge degree]\label{def:mixed_maxlambda}
    Let $\rho$ be an even state of $n$ qubits. The mixed-state bridge degree $\tilde{\Lambda}_{\mathrm{d}}(\rho)$ is the largest $\alpha\geq 0$ such that $\tilde{p}_\rho(4\alpha)\neq 0$.
\end{defn}

Writing $\tilde{\mathcal{V}}_\lambda \coloneqq \bigoplus_{|j|\leq\lambda}\tilde{V}_j$, an equivalent formulation is that $\tilde{\Lambda}_{\mathrm{d}}(\rho)$ is the smallest $\alpha\geq 0$ such that $\kket{\rho^{\otimes 2}}\in\tilde{\mathcal{V}}_{4\alpha}$.

The mixed-state bridge degree inherits the structural features of the pure-state one. Specifically:
\begin{itemize}[leftmargin=2em]
\item[(i)] it takes discrete integer values in $\{0, 2,4, \ldots\}$;
\item[(ii)] it is \emph{faithful} on the set of mixed FGSs;
\item[(iii)] it is \emph{additive} under tensor products of even states;
\item[(iv)] on pure states, it reduces to the pure-state bridge degree;
\item[(v)] is \emph{non-increasing under post-selected Gaussian operations}.
\end{itemize}

We now establish each of these properties.

\begin{lem}[Discreteness and boundedness (mixed-state)]\label{lem:values_maxlambda_mixed}
For any $n$-qubit even state $\rho$, the mixed-state bridge degree $\tilde{\Lambda}_{\mathrm{d}}(\rho)$ takes values of the form $2k$ for some integer $k$ and satisfies
\begin{equation}
    0 \leq \tilde{\Lambda}_{\mathrm{d}}(\rho) \leq 2\lfloor n/2 \rfloor.
\end{equation}
\end{lem}
\begin{proof}
By Lemma~\ref{lem:allowed_eig_Lambda_mixed}, $\tilde{p}_\rho(\lambda)\neq 0$ requires $\lambda = 8k$, so $\alpha = \lambda/4 = 2k$. The upper bound follows from the largest eigenvalue of $\tilde{\Lambda}$ being $4n$.
\end{proof}

\begin{lem}[Faithfulness (mixed-state)]\label{lem:faithfulness_mixed}
$\tilde{\Lambda}_{\mathrm{d}}(\rho)=0$ if and only if $\rho$ is a mixed FGS.
\end{lem}
\begin{proof}
This is immediate from Lemma~\ref{lem:mixed_gauss_choi}: $\tilde{\Lambda}_{\mathrm{d}}(\rho)=0$ iff all mixed-state bridge spectral distribution is concentrated at $\lambda=0$ iff $\tilde{\Lambda}\kket{\rho^{\otimes 2}}=0$ iff $\rho$ is an FGS.
\end{proof}

\begin{lem}[Additivity (mixed-state)]\label{lem:additive_mixed}
For all even states $\rho,\sigma$,
\begin{equation}
    \tilde{\Lambda}_{\mathrm{d}}(\rho\otimes\sigma) = \tilde{\Lambda}_{\mathrm{d}}(\rho) + \tilde{\Lambda}_{\mathrm{d}}(\sigma).
\end{equation}
\end{lem}
\begin{proof}
The argument is parallel to the pure-state case (Lemma~\ref{lem:additive}). Expanding $\kket{(\rho\otimes\sigma)^{\otimes 2}}$ in the eigenbasis of $\tilde{\Lambda}$, each component has eigenvalue equal to the sum of the component eigenvalues. The symmetry $\tilde{p}_\rho(\lambda)=\tilde{p}_\rho(-\lambda)$ (Corollary~\ref{cor:pm_Lambda_mixed}) ensures that the maximal eigenvalue component of the tensor product state $\rho\otimes\sigma$ is attained as the sum of the corresponding maximal eigenvalue component of each state.
\end{proof}

Notably, the mixed-state bridge degree reduces to the pure-state one when applied to pure states. %

\begin{lem}[Reduction to pure states]\label{lem:pure_state_reduction}
For any even pure state $\psi$,
\begin{equation}
    \tilde{\Lambda}_{\mathrm{d}}(\psi) = \nGMax(\psi).
\end{equation}
\end{lem}
\begin{proof}
 A direct calculation using Eq.~\eqref{eq:mixed_bell_weights} shows that the mixed-state Bell-state weights factorize for a pure state:
\begin{equation}\label{eq:mixed_weight_factorize}
    \tilde{p}_\psi(\boldsymbol{\mu}^{(1)},\boldsymbol{\mu}^{(2)}) = |\langle\Psi_{\boldsymbol{\mu}^{(1)}}|\psi\rangle^{\otimes 2}|^2 |\langle\Psi_{\boldsymbol{\mu}^{(2)}}|\psi\rangle^{\otimes 2}|^2 = p_\psi(\boldsymbol{\mu}^{(1)}) \,p_\psi(\boldsymbol{\mu}^{(2)}).
\end{equation} 
 Therefore, if $\tilde{p}_\psi(\boldsymbol{\mu}^{(1)},\boldsymbol{\mu}^{(2)})\neq 0$, then both pure-state weights $p_\psi(\boldsymbol{\mu}^{(1)})$ and $p_\psi(\boldsymbol{\mu}^{(2)})$ are nonzero. Consequently, the corresponding mixed-state eigenvalue $\lambda = \sum_j\mu_j^{(1)} - \sum_j\mu_j^{(2)}$ satisfies $|\lambda|\leq |\lambda^{(1)}|+|\lambda^{(2)}|\leq 4\nGMax(\psi)$, where $\lambda^{(1,2)}=\sum_j\mu_j^{(1,2)}$. This gives $\tilde{\Lambda}_{\mathrm{d}}(\psi)\leq\nGMax(\psi)$.

For the reverse inequality, let $\boldsymbol{\mu}^{*}$ with $p_\psi(\boldsymbol{\mu}^{*})\neq0$ achieve the maximum pure-state eigenvalue, $\sum_j\mu_j^{*} = 2\nGMax(\psi)$. Hence, the state $\ket{\Psi_{\boldsymbol{\mu}^{*},-\boldsymbol{\mu}^{*}}}$ is an eigenstate of $\tilde{\Lambda}$ with eigenvalue $4\nGMax(\psi)$, and Eq.~\eqref{eq:mixed_weight_factorize} gives $\tilde{p}_\psi(\boldsymbol{\mu}^{*},-\boldsymbol{\mu}^{*}) = p_\psi(\boldsymbol{\mu}^{*})^2\neq 0$ (using the symmetry $p_\psi(-\boldsymbol{\mu}) = p_\psi(\boldsymbol{\mu})$ from Lemma~\ref{lem:bell_weight}). Hence $\tilde{\Lambda}_{\mathrm{d}}(\psi)\geq\nGMax(\psi)$.
\end{proof}

Lemma~\ref{lem:pure_state_reduction} guarantees that the pure-state theory of Sec.~\ref{sec:maxlambda} is recovered as a special case of the general mixed-state theory.

\subsection{Monotonicity under post-selected Gaussian operations}
\label{sec:mixed_maxlambda_monotone}

Having established that $\tilde{\Lambda}_{\mathrm{d}}$ is a natural extension of $\nGMax$ to mixed states, we now show that it is a valid non-Gaussianity monotone under the smaller class of post-selected Gaussian operations of Definition~\ref{def:gaussian_operations}.  We note that $\tilde{\Lambda}_{\mathrm{d}}$ cannot be a monotone under the full class of Gaussian protocols, since it is faithful only on the FGSs (Lemma \ref{lem:faithfulness_mixed}) and not on the larger set of convex-Gaussian states.

\begin{thm}[Monotonicity of the mixed-state bridge degree]\label{thm:mixed_max_Lambda_monotone}
For any even state $\rho$ and any post-selected Gaussian operation $\mathcal{E}$ with $\mathcal{E}(\rho) = \sigma$,
\begin{equation}
    \tilde{\Lambda}_{\mathrm{d}}(\sigma) \;\leq\; \tilde{\Lambda}_{\mathrm{d}}(\rho).
\end{equation}
\end{thm}
\begin{proof}
It suffices to check each elementary operation (a)--(d) in Definition~\ref{def:gaussian_operations}.

\emph{(a) Fermionic Gaussian unitaries.} This is immediate from the matchgate-invariance of the mixed-state bridge spectral distribution (Lemma~\ref{lem:invariance_weight_mixed}).

\emph{(b) Composition with an FGS.} This is immediate from additivity under tensor product (Lemma~\ref{lem:additive_mixed}) and faithfulness (Lemma~\ref{lem:faithfulness_mixed}).

\emph{(c) Partial trace.} It suffices to prove monotonicity under tracing out a single fermionic mode, which we may take without loss of generality to be the first mode, since any mode can be brought there by a matchgate --- under which $\tilde{\Lambda}_{\mathrm{d}}$ is invariant by (a). Under the Jordan--Wigner transformation, this is equivalent to tracing out the first qubit.

We split the $n$-qubit bridge operator as
\begin{equation}\label{eq:Lambda_split_ptrace}
    \Lambda \;=\; \Lambda_1 \;+\; G\,\Lambda_{\bar 1},
    \qquad G \coloneqq Z_1 \otimes Z_1,
\end{equation}
where $\Lambda_1$ is the bridge operator of qubit $1$ and $\Lambda_{\bar 1}$ the bridge operator of the remaining $n-1$ qubits. One can verify that $\Lambda_1$ and $G$ commute. Let $\{\ket{\alpha}\}$ be a joint eigenbasis of $(\Lambda_1,G)$,
\begin{equation}
    \Lambda_1\ket{\alpha}=m_\alpha\ket{\alpha},\qquad G\ket{\alpha}=g_\alpha\ket{\alpha},\qquad g_\alpha= \pm1.
\end{equation}
Then for any eigenvector $\ket{u}$ of $\Lambda_{\bar 1}$ with eigenvalue $t$, the product $\ket{\alpha}\otimes\ket{u}$ is an eigenvector of $\Lambda$ with eigenvalue $m_\alpha+g_\alpha t$.

By Eq.~\eqref{eq:mixed_bridge_spectral}, the value $\lambda$ is populated for $\rho$ (i.e.\ $\tilde{p}_\rho(\lambda)\neq0$) if and only if $\rho^{\otimes 2}$ has a nonzero matrix element between two eigenvectors of $\Lambda$ whose eigenvalues differ by $\lambda$. Applied to $\Tr_1\rho$ (with $\Lambda_{\bar 1}$ in place of $\Lambda$), this implies that for any $\lambda\in\operatorname{supp}\tilde{p}_{\Tr_1\rho}$, there exist eigenvectors $\ket{u},\ket{v}$ of $\Lambda_{\bar 1}$ with eigenvalues $t,t'$ satisfying $t-t'=\lambda$ and $\bra{u}\,(\Tr_1\rho)^{\otimes 2}\,\ket{v}\neq0$. Expanding the partial trace over qubit $1$ of both copies in the basis $\{\ket{\alpha}\}$,
\begin{equation}\label{eq:trace_matrix_elem}
    \bra{u}\,(\Tr_1\rho)^{\otimes 2}\,\ket{v}
    \;=\; \sum_\alpha \big(\bra{\alpha}\otimes\bra{u}\big)\,\rho^{\otimes 2}\,\big(\ket{\alpha}\otimes\ket{v}\big),
\end{equation}
we see that at least one summand is nonzero. It is a matrix element of $\rho^{\otimes 2}$ between eigenvectors of $\Lambda$ with eigenvalues $s=m_\alpha+g_\alpha t$ and $s'=m_\alpha+g_\alpha t'$. We have $s-s' = g_\alpha\,(t-t')$, which implies $|s-s'|=|t-t'|=|\lambda|$ since $g_\alpha=\pm1$. Therefore $\pm\lambda$ is populated for $\rho$, and the symmetry $\tilde{p}_\rho(\lambda)=\tilde{p}_\rho(-\lambda)$ (Corollary~\ref{cor:pm_Lambda_mixed}) then shows that $\lambda$ is populated for $\rho$. As this holds for any $\lambda\in\operatorname{supp}\tilde{p}_{\Tr_1\rho}$, we have $\operatorname{supp}\tilde{p}_{\Tr_1\rho}\subseteq\operatorname{supp}\tilde{p}_\rho$, which implies $\tilde{\Lambda}_{\mathrm{d}}(\Tr_1\rho)\leq\tilde{\Lambda}_{\mathrm{d}}(\rho)$.

\emph{(d) Occupation-number measurement with post-selection.} The argument follows the pure-state case. Consider measurement of $Z_1$ with projector $P_\mu = (I+\mu Z_1)/2$, yielding $\rho_\mu = P_\mu\rho P_\mu/\Tr(P_\mu\rho)$. In the Choi representation, $\kket{\rho_\mu}\propto (P_\mu\otimes P_\mu^T)\kket{\rho}$, so the two-copy post-measurement Choi state is related to $\kket{\rho^{\otimes 2}}$ by applying the operator $P_\mu\otimes P_\mu\otimes P_\mu^T\otimes P_\mu^T$, which commutes with $\tilde{\Lambda}$. By the same argument as in the pure-state case, the post-measurement Choi state has support only on a subset of the eigenvalues originally populated, giving $\tilde{\Lambda}_{\mathrm{d}}(\rho_\mu)\leq\tilde{\Lambda}_{\mathrm{d}}(\rho)$.

Composing these bounds yields the theorem.
\end{proof}

Theorem~\ref{thm:mixed_max_Lambda_monotone} completes the mixed-state extension of our framework. In particular, it implies that all the no-go theorems for Gaussian conversion developed in Secs.~\ref{sec:nogo_conversion} extend straightforwardly to mixed states with $\nGMax$ replaced by $\tilde{\Lambda}_{\mathrm{d}}$. Likewise, an approximate version can be developed by closely following the construction of Sec.~\ref{sec:approx}.

\section{Gaussianity testing}
\label{sec:measurements}

In this section, we develop sample-efficient Bell-sampling protocols for testing whether an unknown state is an FGS. We first derive two-sided bounds relating the bridge fidelity $F_0^\Lambda$ --- which is directly accessible from Bell sampling --- to the Gaussian fidelity $F_\G$ of an arbitrary even pure state. We then use these bounds to construct two Gaussianity tests: a one-sided test which is optimal among all two-copy tests with perfect-completeness; and a tolerant test that distinguishes states close to Gaussian from states far from it.

\subsection{Bounds on Gaussian fidelity}
\label{sec:gaussian_fidelity_bounds}

A key ingredient for the Gaussianity test is a quantitative bound relating the bridge fidelity $F_0^\Lambda(\psi)$ to the Gaussian fidelity $F_\G$ (Eq.\eqref{eq:gfid_def}), which quantifies the distance of a state to its closest Gaussian state. We derive this bound via the second moment of $\Lambda$ (Eq.\eqref{eq:m_Lambda_prelim}), which can be expressed in terms of the bridge spectral distribution as 
\begin{equation} \label{eq:m_Lambda}
    \nGLambda(\psi) \;=\; \frac{1}{2}\sum_\lambda p_\psi(\lambda)\,\lambda^2.
\end{equation}
 The first step is to bound $F_\G$ in terms of $\nGLambda$:

\begin{lem}[Gaussian fidelity bounds from $\nGLambda$~\cite{tarabunga2026}]\label{lem:bound_Fg_Lambda}
    For any even pure state $\psi$,
    \begin{equation}
        1-\frac{\nGLambda(\psi)}{2} \;\leq\; F_\G(\psi) \;\leq\; 1 - \frac{1}{4}\!\left(1-\sqrt{1-\frac{\nGLambda(\psi)}{n}}\right)^{\!2}.
    \end{equation}
\end{lem}

This bound can then be converted into the analogous one for $F_0^\Lambda$, which is what we use in the Gaussianity test.

\begin{lem}[Gaussian fidelity bounds from $F_0^\Lambda$]\label{lem:bound_Fg_F_Lambda}
    For any $n$-qubit even pure state $\psi$,
    \begin{equation}
        1 - n^2\!\left(1 - (F_0^\Lambda(\psi))^2\right) \;\leq\; F_\G(\psi) \;\leq\; F_0^\Lambda(\psi).
    \end{equation}
\end{lem}
\begin{proof}
    To show the upper bound, let $\phi$ be the FGS achieving $F_\G(\psi) = F(\psi,\phi)$. Since $\phi$ is Gaussian, $\ket{\phi}^{\otimes 2} \in V_0$, it is an admissible state in the optimization defining $F_0^\Lambda(\psi)$:
    \begin{equation}
        F_0^\Lambda(\psi) = \max_{\Phi\in V_0}\sqrt{F(\psi^{\otimes 2},\Phi)} \;\geq\; \sqrt{F(\psi^{\otimes 2},\phi^{\otimes 2})} = F(\psi,\phi) = F_\G(\psi).
    \end{equation}

    To show the lower bound, using $\lambda \leq 2n$ we have
    \begin{equation}
        \nGLambda(\psi) = \tfrac{1}{2}\sum_\lambda p_\psi(\lambda)\lambda^2 \leq 2n^2 \sum_{\lambda\neq 0} p_\psi(\lambda) = 2n^2\!\left(1 - (F_0^\Lambda(\psi))^2\right),
    \end{equation}
    using $\sum_\lambda p_\psi(\lambda)=1$ and $(F_0^\Lambda)^2 = p_\psi(0)$. Combining with the lower bound of Lemma~\ref{lem:bound_Fg_Lambda} gives the claim.
\end{proof}

Although it does not enter the Gaussianity test, the auxiliary quantity $\nGLambda$ is itself an efficiently estimable function of the Bell-sampling outcomes, which may be of independent interest:

\begin{lem}[Estimation of $\nGLambda$]\label{lem:estimate_M_Lambda}
    Let $\psi$ be an even pure state. $\nGLambda(\psi)$ can be estimated to accuracy $\epsilon$ with probability at least $1-\delta$ using
    \begin{equation}
        N = \mathcal{O}\!\left(\frac{n^3}{\epsilon^2} \,\log\!\left(\frac{1}{\delta}\right)\right)
    \end{equation}
    Bell samples from two copies of $\psi$.
\end{lem}
\begin{proof}
Run Algorithm~\ref{alg:bell_sampling} for $N$ rounds to obtain i.i.d.\ eigenvalues $\lambda_1,\dots,\lambda_N\sim p_\psi$, and form the unbiased estimator $\widehat M = N^{-1}\sum_m \widehat X_m$ with $\widehat X_m = \lambda_m^2/2$. Since $|\lambda_m|\leq 2n$, the random variables $\widehat X_m$ satisfy $\widehat X_m \leq 2n^2$, so $\widehat X_m^2 \leq 2n^2\,\widehat X_m$ and
\begin{equation}
    \mathrm{Var}(\widehat X_m) \leq \mathbb{E}[\widehat X_m^2] \leq 2n^2\,\mathbb{E}[\widehat X_m] = 2n^2\,\nGLambda(\psi) \leq 2n^3.
\end{equation}
Thus, by Bernstein's inequality
\begin{equation}
    \Pr\!\left(|\widehat M - \nGLambda(\psi)|\geq\epsilon\right) \leq 2\exp\!\left(-\frac{N\epsilon^2}{4n^3 + \frac{4}{3}n^2\epsilon}\right).
\end{equation}
Requiring the right-hand side to be at most $\delta$, it follows that $\nGLambda(\psi)$ can be estimated within $\epsilon$ accuracy using
\begin{equation}
    N \geq \frac{12n^3+4n^2\epsilon}{3\epsilon^2 } \log\left(\frac{2}{\delta}\right) ,
\end{equation}
with probability at least $1-\delta$.
\end{proof}

This worst-case bound can be improved adaptively when either $\Lambda_{\mathrm{d}}(\psi)$ or $\nGLambda(\psi)$ is small:

\begin{lem}[Adaptive estimation of $\nGLambda$]\label{lem:adaptive_estimate_M_Lambda}
Let $\psi$ be an even pure state. $\nGLambda(\psi)$ can be estimated to accuracy $\epsilon$ with high probability using 
\begin{equation}\label{eq:adaptive_M_Lambda_complexity}
    N \;=\; \tilde{\mathcal{O}}\!\left(\frac{\nGMax^{2}(\psi)\,\nGLambda(\psi)}{\epsilon^{2}} \,+\, \frac{n^{2}}{\epsilon}\right)
\end{equation}
by an adaptive algorithm, without requiring prior knowledge of $\nGMax(\psi)$. 
\end{lem}

\begin{proof}
The algorithm runs Algorithm~\ref{alg:bell_sampling} sequentially, computing after each round $k$ the running empirical mean $\widehat{\nGLambda}_k = (2k)^{-1}\sum_{i\leq k}\lambda_i^{2}$ and empirical variance $\widehat{\sigma}_k^{2}$ of the i.i.d.\ random variables $X_i \coloneqq \lambda_i^{2}/2$. Each sample $X_i$ takes values in $[0, R]$ with $R\coloneqq 2n^{2}$, since the bridge eigenvalues satisfy $|\lambda_i|\leq 2n$ (Lemma~\ref{lem:allowed_eig_Lambda}).

The empirical Bernstein inequality~\cite{maurer2009empirical} states that, for any fixed sample size $k$ and confidence parameter $\delta'\in(0,1)$,
\begin{equation}\label{eq:emp_bernstein}
    \Pr\!\Bigl(\,\bigl|\widehat{\nGLambda}_k - \nGLambda(\psi)\bigr| \;>\; \mathrm{HW}(k;\delta')\,\Bigr) \;\leq\; \delta',
    \qquad
    \mathrm{HW}(k;\delta') \;\coloneqq\; \sqrt{\frac{2\widehat{\sigma}_k^{2}\,\log(2/\delta')}{k}} \;+\; \frac{7R\,\log(2/\delta')}{3(k-1)}.
\end{equation}
To apply Eq.~\eqref{eq:emp_bernstein} at an adaptive stopping time, set $\delta_k\coloneqq 6\delta/(\pi^{2}k^{2})$ so that $\sum_{k\geq 1}\delta_k = \delta$; by a union bound, Eq.~\eqref{eq:emp_bernstein} then holds simultaneously for all $k$ with probability at least $1-\delta$. The algorithm stops at the first $k$ for which $\mathrm{HW}(k;\delta_k)\leq\epsilon$ and outputs $\widehat{\nGLambda}\coloneqq \widehat{\nGLambda}_k$, which by Eq.~\eqref{eq:emp_bernstein} satisfies $|\widehat{\nGLambda}-\nGLambda(\psi)|\leq\epsilon$ with probability at least $1-\delta$.

It remains to bound the stopping time. The empirical variance $\widehat{\sigma}_k^{2}$ is upper bounded by the true variance $\sigma^{2}\coloneqq \mathrm{Var}(X_i)$ as $\widehat{\sigma}_k^{2}\leq \sigma^{2}+2\sqrt{\tfrac{\sigma^{2}\log(1/\delta)}{2(k-1)}}+\log(1/\delta)/(k-1)$~\cite{maurer2009empirical}. The variance itself is controlled by the bound $|\lambda_i|\leq 2\Lambda_{\mathrm{d}}(\psi)$,
\begin{equation}\label{eq:variance_bound_M_Lambda}
    \sigma^{2} \;\leq\; \mathbb{E}[X_i^{2}] \;\leq\; \bigl(2\Lambda_{\mathrm{d}}^{2}(\psi)\bigr)\cdot \mathbb{E}[X_i] \;=\; 2\Lambda_{\mathrm{d}}^{2}(\psi)\,\nGLambda(\psi).
\end{equation}
Substituting into Eq.~\eqref{eq:emp_bernstein} and solving for the smallest $k$ yields the stopping time of Eq.~\eqref{eq:adaptive_M_Lambda_complexity}.
\end{proof}

\subsection{A Gaussianity test via Bell sampling}
\label{sec:gaussianity_test}

We now derive the central operational result of this section: a sample-efficient Bell-sampling protocol that decides whether an unknown fermionic pure state is Gaussian. The test exploits the symmetry characterization of FGSs: a pure state $\ket{\psi}$ is Gaussian if and only if its bridge spectral distribution is supported entirely at eigenvalue $\lambda=0$.

A two-copy Gaussianity test is characterized by a POVM acceptance operator $P$ with $0\leq P\leq I$ acting on $\mathcal{H}_n^{\otimes 2}$, with acceptance probability $p_{\mathrm{acc}}(\rho) = \Tr(\rho^{\otimes 2} P)$ on input $\rho$. The test has \emph{perfect completeness} if $p_{\mathrm{acc}}(\phi) = 1$ for every pure FGS $\phi$, which (since $0\leq P\leq I$) is equivalent to $P\ket{\phi}^{\otimes 2} = \ket{\phi}^{\otimes 2}$ for all FGS $\phi$.

We first consider the one-sided version of the testing problem.

\begin{defn}[Gaussianity testing]\label{def:gaussianity_test}
Given copies of an even pure state $\psi$ promised to be either
\begin{itemize}[leftmargin=2em]
\item[(a)] an FGS, or
\item[(b)] a state with Gaussian fidelity at most $1-\epsilon$,
\end{itemize}
the Gaussianity testing problem asks to distinguish the two cases with high probability.
\end{defn}

The testing algorithm is a simple protocol:
\begin{quote}\emph{Perform Bell sampling on two copies of $\psi$. Accept if every observed eigenvalue is $0$; reject if any nonzero eigenvalue is observed.}\end{quote}
The sample complexity of this protocol is the content of the next theorem.

\begin{thm}[Gaussianity test]\label{thm:gaussianity_test}
There exists a quantum algorithm that solves the fermionic Gaussianity testing problem (Definition~\ref{def:gaussianity_test}) using
\begin{equation}
    N = \mathcal{O}\!\left(\frac{n^2}{\epsilon}\,\log\!\left(\frac{1}{\delta}\right)\right)
\end{equation}
copies of $\psi$ with success probability at least $1-\delta$. The algorithm has perfect completeness: it always accepts Gaussian states.
\end{thm}
\begin{proof}
If $\psi$ is an FGS, then by Eq.~\eqref{eq:Lambda_condition} all Bell sample outcomes have eigenvalue $\lambda = 0$, and the algorithm always accepts: the test satisfies perfect completeness.

In case (b), we have $F_\G(\psi)\leq 1-\epsilon$, and by the lower bound of Lemma~\ref{lem:bound_Fg_F_Lambda} together with $p_\psi(0)=(F_0^\Lambda)^2$,
\begin{equation}
    p_\psi(0) = (F_0^\Lambda(\psi))^2 \;\leq\; 1 - \frac{\epsilon}{n^2},
\end{equation}
so each Bell sample produces a nonzero eigenvalue with probability at least $\epsilon/n^2$. The probability that all $N$ samples yield $\lambda=0$ (causing the algorithm to incorrectly accept) is therefore at most $(1-\epsilon/n^2)^N \leq e^{-N\epsilon/n^2}$, which is smaller than $\delta$ for $N\geq (n^2/\epsilon)\log(1/\delta)$.
\end{proof}

Theorem~\ref{thm:gaussianity_test} is the fermionic counterpart of the bosonic Gaussianity test of Ref.~\cite{girardi2025gaussian}. It has the advantage of being \emph{experimentally accessible}, since Bell sampling is implementable by a constant-depth two-copy Clifford or matchgate circuit, whereas the bosonic test is not known to admit a practical implementation.

Moreover, the test is \emph{optimal} among all two-copy tests with perfect completeness, in the following sense: for any perfectly complete acceptance operator $P$ and any even pure state $\psi$,
\begin{equation}\label{eq:pointwise_optimality}
    \Tr(\psi^{\otimes 2}\,\Pi_0) \;\leq\; \Tr(\psi^{\otimes 2}\,P).
\end{equation}
This follows from a general fact established in Ref.~\cite[Lemma 2.1]{girardi2025gaussian}: in a $k$-copy property-testing setting with perfect completeness on a class $\mathcal{P}$ of pure states, the optimal acceptance operator is the projection onto the span of $\{\ket{\phi}^{\otimes k} : \phi\in\mathcal{P}\}$. In our case ($k=2$, with $\mathcal{P}$ the set of pure FGSs), this span equals $V_0$: by Lemma~\ref{lem:twirl_vacuum}, $\mathcal{T}^{(2)}_{M_n}(\ket{\phi}\!\bra{\phi}^{\otimes 2}) = \Pi_0/\dim V_0$ for any pure FGS $\phi$, so the span of $\ket{\phi}^{\otimes 2}$ equals $V_0$. The optimal acceptance operator is therefore $\Pi_0$, which is precisely the POVM implemented by one round of Bell sampling on $\psi^{\otimes 2}$ since the latter accepts with probability $p_\psi(0) = \Tr(\psi^{\otimes 2}\Pi_0)$.

The factor $n^2$ in Theorem~\ref{thm:gaussianity_test} originates from the worst-case bound of Lemma~\ref{thm:gaussianity_test}. A particularly intriguing question is whether Gaussianity testing is possible with a sample complexity independent of the system size. Using the bridge degree, we next show that this is the case for states generated by few non-Gaussian gates.

\begin{thm}[Gaussianity testing for $t$-doped Gaussian states]\label{thm:test_tdoped}
Let $\psi$ be an even pure $t$-doped Gaussian state. Then the Gaussianity testing problem (Definition~\ref{def:gaussianity_test}) can be solved with success probability at least $1-\delta$ using
\begin{equation}
    N \;=\; \mathcal{O}\!\left(\frac{t^2}{\epsilon}\,\log\!\left(\frac{1}{\delta}\right)\right)
\end{equation}
copies of $\psi$, independently of the system size $n$. The algorithm has perfect completeness.
\end{thm}

\begin{proof}
By Corollary~\ref{cor:bound_dope} we have $\nGMax(\psi)\leq 4t$, so $p_\psi$ is supported on $|\lambda|\leq 2\nGMax(\psi)\leq 8t$. Hence, using Eq.~\eqref{eq:m_Lambda},
\begin{equation}
    \nGLambda(\psi) \;=\; \frac{1}{2}\sum_{\lambda\neq 0}\lambda^2\, p_\psi(\lambda)
    \;\leq\; \frac{(8t)^2}{2}\sum_{\lambda\neq 0}p_\psi(\lambda)
    \;=\; 32\,t^2\bigl(1-p_\psi(0)\bigr),
\end{equation}
while the lower bound of Lemma~\ref{lem:bound_Fg_Lambda} reads $\nGLambda(\psi)\geq 2\bigl(1-F_\G(\psi)\bigr)$. Combining the two inequalities,
\begin{equation}
    1-p_\psi(0) \;\geq\; \frac{1-F_\G(\psi)}{16\,t^2}.
\end{equation}
In case (b) of Definition~\ref{def:gaussianity_test} we have $F_\G(\psi)\leq 1-\epsilon$, so each Bell sample yields a nonzero eigenvalue with probability at least $\epsilon/(16t^2)$. Thus the testing algorithm of Theorem~\ref{thm:gaussianity_test}, which has perfect completeness, applies verbatim with $n$ replaced by $4t$, yielding the stated sample complexity.
\end{proof}

Note that since the non-Gaussian gates can be chosen to be ``weak'', Theorem~\ref{thm:test_tdoped} covers states arbitrarily close to the set of Gaussian states, which should intuitively be the most difficult to distinguish from Gaussian. We therefore conjecture that the restriction to few non-Gaussian gates is not necessary. This would make the sample complexity of the test independent of the system size in general, as is the case for stabilizer testing~\cite{gross2021schur}, and would yield a correspondingly improved version of Theorem~\ref{thm:gaussianity_test}. An analogous statement has been established in the bosonic setting for states close to the Gaussian set in fidelity~\cite{girardi2025gaussian}.

We next consider the tolerant version of the testing problem:

\begin{defn}[Tolerant Gaussianity testing]\label{def:gaussianity_test_tolerant}
Given parameters $0\leq\alpha<\beta\leq 1$ and copies of an even pure state $\psi$ promised to be either
\begin{itemize}[leftmargin=2em]
\item[(a)] $F_\G(\psi)\geq 1-\alpha$, or
\item[(b)] $F_\G(\psi)\leq 1-\beta$,
\end{itemize}
the tolerant Gaussianity testing problem asks to distinguish the two cases with high probability.
\end{defn}

\begin{thm}[Tolerant Gaussianity test]\label{thm:gaussianity_test_tolerant}
For any $0\leq\alpha<\beta\leq 1$ such that $\varepsilon \coloneqq \beta/n^2 - 2\alpha > 0$ and $\alpha = \mathcal{O}(\varepsilon)$, there exists a quantum algorithm that solves the tolerant Gaussianity testing problem (Definition~\ref{def:gaussianity_test_tolerant}) with success probability at least $1-\delta$, using
\begin{equation}\label{eq:tolerant_test_lemma9_complexity}
    N  \;=\; \mathcal{O}\!\left(\frac{n^2}{\beta-2\alpha n^2}\,\log(1/\delta)\right)
\end{equation}
copies of $\psi$.
\end{thm}
\begin{proof}
Run Algorithm~\ref{alg:bell_sampling} for $k$ rounds and let $X_i \coloneqq \mathbf{1}\{\lambda_i = 0\}\in\{0,1\}$ be the indicator of observing $\lambda = 0$ in round $i$. The $X_i$ are i.i.d.\ with $\mathbb{E}[X_i] = p_\psi(0) = (F_0^\Lambda(\psi))^2$. By Lemma~\ref{lem:bound_Fg_F_Lambda}, case (a) of Definition~\ref{def:gaussianity_test_tolerant} forces $p_\psi(0)\geq (1-\alpha)^2 \geq 1-2\alpha \eqqcolon p_A$, while case (b) forces $p_\psi(0)\leq 1-\beta/n^2 \eqqcolon p_B$. By assumption, the gap $p_A - p_B = \varepsilon$ is positive, and $1 - p_A = 2\alpha = \mathcal{O}(\varepsilon)$, placing the test in the high-acceptance regime. The test outputs \emph{accept} if and only if $\widehat{p}_0 \coloneqq k^{-1}\sum_i X_i \geq (p_A+p_B)/2$. Lemma~9 of Ref.~\cite{girardi2025gaussian} then gives that
\begin{equation}
    k \;=\; \mathcal{O}\!\left(\frac{1}{\varepsilon}\,\log\!\left(\frac{1}{\delta}\right)\right)
\end{equation}
samples suffice for success probability at least $1-\delta$; substituting $\varepsilon = \beta/n^2 - 2\alpha$ gives Eq.~\eqref{eq:tolerant_test_lemma9_complexity}.
\end{proof}

We note that an alternative tolerant Gaussianity test can be constructed using the second moment $\nGLambda(\psi)$ in place of $p_\psi(0)$, by combining the two-sided bound of Lemma~\ref{lem:bound_Fg_Lambda} with the adaptive estimator of Lemma~\ref{lem:adaptive_estimate_M_Lambda}. The $\nGLambda$-based variant, however, requires the stricter promise gap $\alpha < \beta^2/(4n^2)$ rather than $\alpha < \beta/(2n^2)$. Finally, we remark that the sample complexity of either variant could be further improved by sharpening the underlying inequalities relating $\nGLambda$ or $p_\psi(0)$ to the Gaussian fidelity $F_\G$. 

\section{Certifying matchgate-invariant state $2$-designs}
\label{sec:designs_cert}

The lower bounds of Sec.~\ref{sec:designs_lower} tell us how many non-Gaussian gates must be present in any state design. The complementary question is whether, given a candidate ensemble, one can \emph{certify} from experimental data that it is in fact an approximate $2$-design. In general the corresponding decision problem is not known to admit an efficient quantum algorithm; partial complexity-theoretic characterization in this direction is provided in Ref.~\cite{nakata2025computational}. In this section we show that, for a natural class of structured ensembles --- the matchgate-invariant ones --- the certification problem reduces to an efficient task, solvable from $\mathrm{poly}(n)$ Bell samples. The key technical object is the ensemble-averaged bridge spectral distribution
\begin{equation}\label{eq:p_E_def}
    p_{\mathcal{E}}(\lambda) \;\coloneqq\; \mathbb{E}_{\psi\sim\mathcal{E}}\bigl[p_\psi(\lambda)\bigr],
\end{equation}
which encodes the second-order matchgate-invariant content of $\mathcal{E}$ and is directly sampled by ensemble Bell sampling.

\begin{defn}[Matchgate-invariant ensemble]\label{def:matchgate_invariant_ensemble}
An ensemble $\mathcal{E}$ of even pure states on $n$ qubits is \emph{matchgate-invariant} if its ensemble-averaged two-copy state is invariant under the matchgate twirl:
\begin{equation}\label{eq:matchgate_invariant_def}
    \mathcal{T}^{(2)}_{M_n}\!\bigl(\mathbb{E}_{\psi\sim\mathcal{E}}[\psi^{\otimes 2}]\bigr) \;=\; \mathbb{E}_{\psi\sim\mathcal{E}}[\psi^{\otimes 2}].
\end{equation}
\end{defn}

A natural and broad family of matchgate-invariant ensembles is given by \emph{matchgate orbits}: starting from any ensemble $\mathcal{E}_0$ of even pure states, draw $\psi\sim\mathcal{E}_0$ and conjugate by an independent Haar-random matchgate. The resulting ensemble is matchgate-invariant by construction, since averaging over Haar-random matchgates is precisely the matchgate twirl $\mathcal{T}^{(2)}_{M_n}$ applied to $\mathbb{E}_{\psi_0\sim\mathcal{E}_0}[\psi_0^{\otimes 2}]$.

For any matchgate-invariant ensemble, Definition~\ref{def:matchgate_invariant_ensemble} together with Lemma~\ref{lem:two_copy_ave_matchgate} gives the expansion
\begin{equation}\label{eq:p_E_spectral}
    \mathbb{E}_{\psi\sim\mathcal{E}}\bigl[\psi^{\otimes 2}\bigr] \;=\; \mathcal{T}^{(2)}_{M_n}\!\bigl(\mathbb{E}_{\psi\sim\mathcal{E}}[\psi^{\otimes 2}]\bigr) \;=\; \sum_{\lambda\in\mathcal{A}} p_{\mathcal{E}}(\lambda)\,\frac{\Pi_\lambda}{\binom{2n}{n+\lambda/2}}.
\end{equation}
Combining with the analogous decomposition of the Haar-twirled state in Eq.~\eqref{eq:Haar_twirl_decomp} yields $D(\mathcal{E})$ in closed form:
\begin{equation}\label{eq:design_distance_TV}
    D(\mathcal{E}) \;=\; \mathrm{TV}\!\left(p_{\mathcal{E}},\, p_{\mathrm H}\right)
    \;=\; \tfrac{1}{2}\sum_{\lambda\in\mathcal{A}}\bigl|p_{\mathcal{E}}(\lambda) - p_{\mathrm H}(\lambda)\bigr|.
\end{equation}

Certifying that $\mathcal{E}$ is an $\epsilon$-approximate $2$-design therefore reduces to estimating $\mathrm{TV}(p_{\mathcal{E}}, p_{\mathrm H})$ from Bell-sampling data. The corresponding measurement primitive is the natural ensemble extension of Algorithm~\ref{alg:bell_sampling}: at each round, draw a fresh $\psi\sim\mathcal{E}$ and run Bell sampling on $\psi^{\otimes 2}$; marginalizing over the ensemble draw, the outcome is distributed as $p_{\mathcal{E}}$.

\begin{algorithm}[H]
\caption{Ensemble Bell sampling for the bridge spectral distribution}
\label{alg:ensemble_bell_sampling}
\begin{flushleft}
\textbf{Input:} $N\geq 1$, and oracle access to an ensemble $\mathcal{E}$ of even pure states.\\
\textbf{Output:} eigenvalues $(\lambda_1,\dots,\lambda_N)\in (8\mathbb{Z})^N$ drawn i.i.d.\ from $p_{\mathcal{E}}$. The empirical ensemble-averaged bridge spectral distribution is $\widehat p_{\mathcal{E}}(\lambda) \coloneqq N^{-1}\sum_{k=1}^N \mathbf{1}\{\lambda_k=\lambda\}$.
\end{flushleft}
\begin{algorithmic}[1]
\For{$k=1,\dots,N$}
    \State Draw $\psi_k \sim \mathcal{E}$ and prepare two fresh copies of $\psi_k$.
    \State Run Algorithm~\ref{alg:bell_sampling} once on $\psi_k^{\otimes 2}$ to obtain an eigenvalue $\lambda_k$.
\EndFor
\State \Return $(\lambda_1,\dots,\lambda_N)$.
\end{algorithmic}
\end{algorithm}

From this protocol we address two complementary certification tasks. The first is \emph{approximate $2$-design testing}: deciding whether $D(\mathcal{E})\leq \alpha$ or $D(\mathcal{E})\geq \alpha+\epsilon$. The second is \emph{exact $2$-design testing}: deciding whether $D(\mathcal{E})=0$ or $D(\mathcal{E})\geq \epsilon$, which admits a tighter sample bound. For the approximate task, our route is to first estimate $D(\mathcal{E})$ to additive accuracy $\epsilon$; the resulting sample complexity is the content of the next lemma.

\begin{lem}[Efficient estimation of $D(\mathcal{E})$]\label{lem:estimate_design_distance}
For any matchgate-invariant ensemble $\mathcal{E}$ of even pure states, $D(\mathcal{E})$
can be estimated to additive accuracy $\epsilon$ with probability at least $1-\delta$ using
\begin{equation}\label{eq:est_design_dist_N}
    N \;=\; \mathcal{O}\!\left(\frac{\sqrt {n\log(1/\epsilon)} + \log(1/\delta)}{\epsilon^2}\right)
\end{equation}
runs of Algorithm~\ref{alg:ensemble_bell_sampling}.
\end{lem}

The estimation procedure underlying Lemma~\ref{lem:estimate_design_distance} is summarized as Algorithm~\ref{alg:estimate_design_distance}.

\begin{algorithm}[H]
\caption{Estimator for $D(\mathcal{E})$}
\label{alg:estimate_design_distance}
\begin{flushleft}
\textbf{Input:} accuracy $\epsilon\in(0,1)$, failure probability $\delta\in(0,1)$, oracle access to a matchgate-invariant ensemble $\mathcal{E}$.\\
\textbf{Output:} $\widehat D$ satisfying $|\widehat D - D(\mathcal{E})|\leq\epsilon$ with probability at least $1-\delta$.
\end{flushleft}
\begin{algorithmic}[1]
\State Set bulk threshold $T \gets \bigl\lceil\sqrt{4n\log(9/\epsilon)} \bigr\rceil$ and partition $\mathcal{A}_{\mathrm{bulk}}\gets\mathcal{A}\cap[-T,T]$, $\mathcal{A}_{\mathrm{tail}}\gets\mathcal{A}\setminus\mathcal{A}_{\mathrm{bulk}}$, with bulk size $K\gets|\mathcal{A}_{\mathrm{bulk}}|=2\lfloor T/8\rfloor+1$.
\State Set sample count $N \gets \bigl\lceil (2K + 9\log(4/\delta))/\epsilon^2\bigr\rceil$.
\State Run Algorithm~\ref{alg:ensemble_bell_sampling} for $N$ rounds to obtain $\lambda_1,\dots,\lambda_N\sim p_{\mathcal{E}}$.
\State Form the empirical distribution $\widehat p_\lambda \gets N^{-1}\sum_{j=1}^N \mathbf{1}\{\lambda_j=\lambda\}$ for $\lambda\in\mathcal{A}_{\mathrm{bulk}}$, and $\widehat p(\mathcal{A}_{\mathrm{tail}}) \gets N^{-1}\sum_{j=1}^N \mathbf{1}\{\lambda_j\in\mathcal{A}_{\mathrm{tail}}\}$.
\State Compute the bulk estimate $\widehat{D}_{\mathrm{bulk}}\gets \tfrac{1}{2}\sum_{\lambda\in\mathcal{A}_{\mathrm{bulk}}} |\widehat p_\lambda - p_{\mathrm H}(\lambda)|$ using the closed form Eq.~\eqref{eq:p_H} for $p_{\mathrm H}$.
\State Compute the tail upper bound $\widehat{D}^{\mathrm{up}}_{\mathrm{tail}}\gets \tfrac{1}{2}\bigl(\widehat p(\mathcal{A}_{\mathrm{tail}}) + p_{\mathrm H}(\mathcal{A}_{\mathrm{tail}})\bigr)$ using the closed form Eq.~\eqref{eq:p_H} for $p_{\mathrm H}$.
\State \Return $\widehat D \gets \widehat{D}_{\mathrm{bulk}} + \widehat{D}^{\mathrm{up}}_{\mathrm{tail}}$.
\end{algorithmic}
\end{algorithm}

\begin{proof}

Algorithm~\ref{alg:ensemble_bell_sampling} provides i.i.d.\ samples from $p_{\mathcal{E}}$, suggesting an immediate strategy by estimating $D(\mathcal{E})=\mathrm{TV}(p_{\mathcal{E}},p_{\mathrm H})$ from its empirical version. By a standard result on learning discrete distributions in total variation distance~\cite{canonne2020shortnote}, this naive estimator has sample complexity $\mathcal{O}(n/\epsilon^2)$, due to the support size $|\mathcal{A}|=\mathcal{O}(n)$. We improve this to $\mathcal{O}(\sqrt {n\log(1/\epsilon)}/\epsilon^2)$ by exploiting the concentration of $p_{\mathrm H}$ on an interval of width $\mathcal{O}(\sqrt n)$ around the origin, as established in Lemma~\ref{lem:p_H_concentration}. It thus suffices to estimate the empirical distribution accurately only on this much smaller bulk.

From the tail bound of Lemma~\ref{lem:p_H_subgauss},
\begin{equation}
    p_{\mathrm H}\bigl(|\lambda|>T\bigr) \;\leq\; (4+2^{2-n})\,e^{-T^2/(4n)},
\end{equation}
we fix the truncation threshold
\begin{equation}
    T = \sqrt{4n\log(9/\epsilon)},
\end{equation}
so that $p_{\mathrm H}(|\lambda|>T) \leq \epsilon/2$ for sufficiently large $n$. 

Partition the support into
\begin{equation}
    \mathcal{A}_{\mathrm{bulk}} \coloneqq \mathcal{A}\cap[-T,T],
    \qquad
    \mathcal{A}_{\mathrm{tail}} \coloneqq \mathcal{A}\setminus\mathcal{A}_{\mathrm{bulk}},
    \qquad
    K \coloneqq |\mathcal{A}_{\mathrm{bulk}}| \;=\; \mathcal{O}\!\bigl(\sqrt{n\log(1/\epsilon)}\bigr).
\end{equation}
This partition decomposes the total variation distance as $D(\mathcal{E}) = D_{\mathrm{bulk}} + D_{\mathrm{tail}}$, where
\begin{equation}
    D_{\mathrm{bulk}} \coloneqq \tfrac{1}{2}\!\sum_{\lambda\in\mathcal{A}_{\mathrm{bulk}}}\!|p_{\mathcal{E}}(\lambda)-p_{\mathrm H}(\lambda)|,
    \qquad
    D_{\mathrm{tail}} \coloneqq \tfrac{1}{2}\!\sum_{\lambda\in\mathcal{A}_{\mathrm{tail}}}\!|p_{\mathcal{E}}(\lambda)-p_{\mathrm H}(\lambda)|.
\end{equation}
The bulk piece $D_{\mathrm{bulk}}$ will be estimated directly from the empirical distribution; the tail piece $D_{\mathrm{tail}}$ admits, via the identity $|a-b|=(a+b)-2\min(a,b)$, the exact decomposition
\begin{equation}\label{eq:tail_identity}
    D_{\mathrm{tail}}
    \;=\; \tfrac{1}{2}\bigl(p_{\mathcal{E}}(\mathcal{A}_{\mathrm{tail}}) + p_{\mathrm H}(\mathcal{A}_{\mathrm{tail}})\bigr)
    \;-\; R,
    \qquad R \coloneqq \!\!\sum_{\lambda\in\mathcal{A}_{\mathrm{tail}}}\!\!\min\bigl(p_{\mathcal{E}}(\lambda),p_{\mathrm H}(\lambda)\bigr),
\end{equation}
with $0 \leq R \leq p_{\mathrm H}(\mathcal{A}_{\mathrm{tail}}) \leq \epsilon/2$. The strategy is therefore to estimate the explicit upper bound $D^{\mathrm{up}}_{\mathrm{tail}}\coloneqq \tfrac{1}{2}\bigl(p_{\mathcal{E}}(\mathcal{A}_{\mathrm{tail}}) + p_{\mathrm H}(\mathcal{A}_{\mathrm{tail}})\bigr)$ from the empirical distribution, leaving a deterministic bias given by the residual $R$ of at most $\epsilon/2$.

Concretely, we run Algorithm~\ref{alg:ensemble_bell_sampling} for $N$ rounds to obtain $\lambda_1,\dots,\lambda_N\sim p_{\mathcal{E}}$, form the empirical distribution $\widehat p_\lambda \coloneqq N^{-1}\sum_j\mathbf{1}\{\lambda_j=\lambda\}$, and define the estimator
\begin{equation}
    \widehat D \;\coloneqq\; \widehat{D}_{\mathrm{bulk}} + \widehat{D}^{\mathrm{up}}_{\mathrm{tail}},
    \qquad
    \widehat{D}_{\mathrm{bulk}} \coloneqq \tfrac{1}{2}\!\!\!\sum_{\lambda\in\mathcal{A}_{\mathrm{bulk}}}\!\!\!|\widehat p_\lambda - p_{\mathrm H}(\lambda)|,
    \qquad
    \widehat{D}^{\mathrm{up}}_{\mathrm{tail}} \coloneqq \tfrac{1}{2}\bigl(\widehat p(\mathcal{A}_{\mathrm{tail}}) + p_{\mathrm H}(\mathcal{A}_{\mathrm{tail}})\bigr).
\end{equation}
Note that $p_{\mathrm H}(\mathcal{A}_{\mathrm{tail}})$ is a deterministic quantity, computed from Eq.~\eqref{eq:p_H} (or upper-bounded via Eq.~\eqref{eq:pH_subgauss}); only $\widehat p$ depends on the sampling data. We have the error decomposition
\begin{equation}\label{eq:est_decomp_design}
    \widehat D - D
    \;=\; \bigl(\widehat{D}_{\mathrm{bulk}} - D_{\mathrm{bulk}}\bigr)
    \;+\; \bigl(\widehat{D}^{\mathrm{up}}_{\mathrm{tail}} - D^{\mathrm{up}}_{\mathrm{tail}}\bigr)
    \;+\; R,
\end{equation}
which we bound term by term.

For the bulk, the triangle inequality yields
\begin{equation}
    |\widehat{D}_{\mathrm{bulk}} - D_{\mathrm{bulk}}|
    \;\leq\; \tfrac{1}{2}\!\!\!\sum_{\lambda\in\mathcal{A}_{\mathrm{bulk}}}\!\!\!|\widehat p_\lambda - p_{\mathcal{E}}(\lambda)|
    \;=:\; \mathrm{TV}_{\mathrm{bulk}}(\widehat p, p_{\mathcal{E}}).
\end{equation}
We bound the expected value of each summand as
\begin{equation}
    \mathbb{E}|\widehat p_\lambda - p_{\mathcal{E}}(\lambda)| \leq \sqrt{\mathbb{E}[(\widehat p_\lambda - p_{\mathcal{E}}(\lambda))^2]} = \sqrt{\frac{1}{N} p_{\mathcal{E}}(\lambda)(1-p_{\mathcal{E}}(\lambda))}\leq\sqrt{p_{\mathcal{E}}(\lambda)/N},
\end{equation}
where we used Jensen's inequality in the first inequality. By the Cauchy--Schwarz inequality,
\begin{equation}
    \mathbb{E}\bigl[\mathrm{TV}_{\mathrm{bulk}}(\widehat p, p_{\mathcal{E}})\bigr]
    \;\leq\; \tfrac{1}{2}\!\!\!\sum_{\lambda\in\mathcal{A}_{\mathrm{bulk}}}\!\!\!\sqrt{p_{\mathcal{E}}(\lambda)/N}
    \;\leq\; \tfrac{1}{2}\sqrt{K/N}.
\end{equation}
The map $(\lambda_1,\dots,\lambda_N)\mapsto \mathrm{TV}_{\mathrm{bulk}}(\widehat p,p_{\mathcal{E}})$ has bounded differences $\leq 1/N$, so McDiarmid's inequality~\cite{McDiarmid1989} gives, with probability at least $1-\delta/2$,
\begin{equation}\label{eq:bulk_concentration}
    |\widehat{D}_{\mathrm{bulk}} - D_{\mathrm{bulk}}|
    \;\leq\; \tfrac{1}{2}\sqrt{K/N} + \sqrt{\tfrac{\log(4/\delta)}{2N}}.
\end{equation}

For the tail, the deviation 
$\widehat{D}^{\mathrm{up}}_{\mathrm{tail}} - D^{\mathrm{up}}_{\mathrm{tail}} = \tfrac{1}{2}\bigl(\widehat p(\mathcal{A}_{\mathrm{tail}}) - p_{\mathcal{E}}(\mathcal{A}_{\mathrm{tail}})\bigr)$ can be bounded by Hoeffding's inequality, yielding 
\begin{equation}\label{eq:tail_concentration}
    |\widehat{D}^{\mathrm{up}}_{\mathrm{tail}} - D^{\mathrm{up}}_{\mathrm{tail}}|
    \;\leq\; \tfrac{1}{2}\sqrt{\tfrac{\log(4/\delta)}{2N}},
\end{equation}
with probability at least $1-\delta/2$. The residual term $R$ in Eq.~\eqref{eq:est_decomp_design} is deterministic and satisfies $0\leq R\leq \epsilon/2$.

Combining Eqs.~\eqref{eq:bulk_concentration} and \eqref{eq:tail_concentration}, by the triangle inequality, we obtain that with probability at least $1-\delta$,
\begin{equation}
    |\widehat D - D|
    \;\leq\; \tfrac{1}{2}\sqrt{K/N} \;+\; \tfrac{3}{2}\sqrt{\tfrac{\log(4/\delta)}{2N}} \;+\; \tfrac{\epsilon}{2}.
\end{equation}
Requiring the right-hand side to be at most $\epsilon$ and applying $\sqrt{x}+\sqrt{y}\leq\sqrt{2(x+y)}$ to the two square-root terms yields the sufficient condition
\begin{equation}
    N \;\geq\; \frac{2K + 9\log(4/\delta)}{\epsilon^2}
    \;=\; \mathcal{O}\!\left(\frac{\sqrt {n\log(1/\epsilon)} + \log(1/\delta)}{\epsilon^2}\right),
\end{equation}
which is Eq.~\eqref{eq:est_design_dist_N}.
\end{proof}

Exploiting the estimator from Lemma~\ref{lem:estimate_design_distance} immediately yields a tolerant tester for the approximate $2$-design property, with the same $\widetilde{\mathcal{O}}(\sqrt n/\epsilon^2)$ scaling.

\begin{thm}[Testing of approximate $2$-design]\label{thm:test_design_distance}
Let $\mathcal{E}$ be a matchgate-invariant ensemble of even pure states, where it is promised that
\begin{align*}
\mathrm{either}\quad (a)& \,\,D(\mathcal{E}) \leq \alpha \,,\\
\mathrm{or}\quad (b)& \,\, D(\mathcal{E}) \geq \beta\,,
\end{align*} 
for  $0\leq \alpha<\beta\leq 1$. Then, there exists a Bell-sampling-based algorithm that, using
\begin{equation}\label{eq:test_design_dist_N}
    N \;=\; \mathcal{O}\!\left(\frac{\sqrt{n\,\log\!\bigl(1/(\beta-\alpha)}\bigr) + \log(1/\delta)}{(\beta-\alpha)^2}\right)
\end{equation}
runs of Algorithm~\ref{alg:ensemble_bell_sampling}, distinguishes case ($a$) and ($b$) with probability at least $1-\delta$.

\end{thm}

\begin{proof}
Set $\epsilon =(\beta-\alpha)/2$ and run the estimator from Lemma~\ref{lem:estimate_design_distance} at additive accuracy $\epsilon$ and failure probability $\delta$; this requires $N$ runs of Algorithm~\ref{alg:ensemble_bell_sampling} as in Eq.~\eqref{eq:test_design_dist_N}. Let $\widehat D$ denote the resulting estimate. We then make the following decision: if $\widehat D \leq (\alpha+\beta)/2$ then we output (a), otherwise (b). 

By Lemma~\ref{lem:estimate_design_distance}, $|\widehat D - D| \leq \epsilon = (\beta-\alpha)/2$ with probability at least $1-\delta$. Conditioning on this event:
\begin{itemize}
    \item If $D \leq \alpha$, then $\widehat D \leq \alpha + \tfrac{\beta-\alpha}{2} = \tfrac{\alpha+\beta}{2}$, so the algorithm outputs (a);
    \item If $D \geq \beta$, then $\widehat D \geq \beta - \tfrac{\beta-\alpha}{2} = \tfrac{\alpha+\beta}{2}$, so the algorithm outputs (b).
\end{itemize}
This completes the proof.
\end{proof}

For the exact $2$-design test, we obtain a tighter bound, as shown below.

\begin{thm}[Testing of exact $2$-design]\label{thm:nontolerant_test_design}
Let $\mathcal{E}$ be a matchgate-invariant ensemble of even pure states where it is promised that
\begin{align*}
\mathrm{either}\quad (a)& \,\,D(\mathcal{E}) = 0 \,,\\
\mathrm{or}\quad (b)& \,\, D(\mathcal{E}) \geq \epsilon\,,
\end{align*} 
for $\epsilon>0$. Then, there exists a Bell-sampling-based algorithm that uses
\begin{equation}\label{eq:nontolerant_test_N}
    N \;=\; \mathcal{O}\!\left(\frac{n^{1/4}}{\epsilon^2}\,\log(1/\delta)\right)
\end{equation}
runs of Algorithm~\ref{alg:ensemble_bell_sampling} to distinguish case ($a$) and ($b$) with probability at least $1-\delta$.
\end{thm}

\begin{proof}
Algorithm~\ref{alg:ensemble_bell_sampling} provides i.i.d.\ samples from $p_{\mathcal{E}}$, so we can directly invoke the optimal identity tester of Ref.~\cite{valiant2017}, which distinguishes $p_{\mathcal{E}} = p_{\mathrm H}$ from $\mathrm{TV}(p_{\mathcal{E}}, p_{\mathrm H})\geq \epsilon$ with probability at least $1-\delta$ using
$\mathcal{O}\bigl(\|p_{\mathrm H}\|_{2/3}\,\epsilon^{-2}\,\log(1/\delta)\bigr)$ samples.
It therefore suffices to bound $\|p_{\mathrm H}\|_{2/3}=\mathcal{O}(n^{1/4})$, which we do in Lemma~\ref{lem:pH_2/3_norm} below.
\end{proof}

\begin{lem}[$\ell^{2/3}$ norm of the Haar bridge spectral distribution]\label{lem:pH_2/3_norm}
There is an absolute constant $C> 0$ such that, for sufficiently large $n$,
\begin{equation}\label{eq:pH_2/3_bound}
    \|p_{\mathrm H}\|_{2/3}^{2/3} \;\coloneqq\; \sum_{\lambda\in\mathcal{A}}\, p_{\mathrm H}(\lambda)^{2/3} \;\leq\; C\,n^{1/6},
\end{equation}
and consequently $\|p_{\mathrm H}\|_{2/3} = \mathcal{O}(n^{1/4})$.
\end{lem}

\begin{proof}
Set $w\coloneqq 2\sqrt{2n}$ and partition the support of $p_{\mathrm H}$ into intervals
\begin{equation}
    B_0 \;\coloneqq\; \{\lambda\in 8\mathbb{Z}\,:\,|\lambda|\leq w\},
    \qquad
    B_k \;\coloneqq\; \{\lambda\in 8\mathbb{Z}\,:\,kw<|\lambda|\leq (k+1)w\} \quad\text{for } k\geq 1,
\end{equation}
which form a disjoint cover of $8\mathbb{Z}\cap[-2n,2n]$. Each interval contains at most $|B_k|\leq \tfrac{2w}{8}+1 = \sqrt{n/2}+1 \leq \sqrt n$ points for $n\geq 12$.

By Lemma~\ref{lem:p_H_subgauss},
\begin{equation}
    p_{\mathrm H}(B_k) \;\leq\; p_{\mathrm H}(|\lambda|>kw) \;\leq\; (4+2^{2-n})\,e^{-k^2}
    \qquad\text{for } k\geq 1,
\end{equation}
while $p_{\mathrm H}(B_0)\leq 1$ trivially. By Hölder's inequality,
\begin{equation}
    \sum_{\lambda\in B_k} p_{\mathrm H}(\lambda)^{2/3}
    \;\leq\; |B_k|^{1/3}\,p_{\mathrm H}(B_k)^{2/3}
    \;\leq\; n^{1/6}\,p_{\mathrm H}(B_k)^{2/3}.
\end{equation}
Summing over $k\geq 0$,
\begin{equation}
    \sum_\lambda p_{\mathrm H}(\lambda)^{2/3}
    \;\leq\; n^{1/6}\!\left(1 \,+\, (4+2^{2-n})^{2/3}\!\!\sum_{k\geq 1}\!e^{-2k^2/3}\right)
    \;\leq\; C\,n^{1/6},
\end{equation}
since the series converges to an absolute constant. This proves Eq.~\eqref{eq:pH_2/3_bound}.
\end{proof}

These results, together with those of Sec.~\ref{sec:designs_lower}, give a complete operational picture of fermionic state designs in both the exact and approximate cases.

\emph{Exact designs.} Any $t$-doped Gaussian protocol generating an exact state $2$-design must produce at least one state requiring $\Omega(n)$ non-Gaussian gates (Corollary~\ref{cor:exact_design}), and the property of being an exact $2$-design can be certified from $\mathcal{O}(n^{1/4}/\epsilon^2)$ Bell samples (Theorem~\ref{thm:nontolerant_test_design}).

\emph{Approximate designs.} Any $t$-doped Gaussian protocol generating an $\epsilon$-approximate state $2$-design requires $\Omega(\sqrt{n\log(1/\epsilon)})$ non-Gaussian gates with probability at least $2\epsilon$ over the ensemble (Theorem~\ref{thm:approx_design}, valid for $\epsilon\geq C_*/\sqrt n$); equivalently, any state $\psi$ with $\nGMax(\psi)=o(\sqrt{n\log(1/\epsilon)})$ has its matchgate orbit at trace distance $>\epsilon$ from the Haar-random state. On the certification side, $D(\mathcal{E})$ can be estimated to additive accuracy $\epsilon$ from $\widetilde{\mathcal{O}}(\sqrt n/\epsilon^2)$ Bell samples (Lemma~\ref{lem:estimate_design_distance}), with a corresponding tolerant testing procedure (Theorem~\ref{thm:test_design_distance}).

\section{Discussion and open problems}
\label{sec:conclusion}

\subsection{Discussion}
We have developed a unified resource-theoretic and operational framework for fermionic non-Gaussianity, organized around two structural pillars: the \emph{bridge spectral distribution} $p_\psi(\lambda)$, which encodes the matchgate-invariant content of the two-copy state $\ket{\psi}^{\otimes 2}$, and \emph{Bell sampling}, which gives direct experimental access to this distribution from a constant-depth two-copy Clifford or matchgate circuit. Every empirically accessible non-Gaussianity quantifier in this framework is a closed-form functional of $p_\psi$, and each of our operational applications follows from a structural property of this single object. The second-order matchgate commutant, accessed through Bell sampling, thereby reduces a host of structural and computational questions about fermionic non-Gaussianity to statistical questions about $p_\psi$.

The main results of this work can be summarized as follows. We introduce the bridge degree and prove that it is a bona fide non-Gaussianity monotone under post-selected Gaussian protocols. The resulting structural constraints yield strong no-go theorems for Gaussian conversion and establish that the resource theory of fermionic non-Gaussianity is strictly irreversible in the exact-conversion setting. As a concrete operational consequence, we obtain non-Gaussian gate-count lower bounds for fermionic state designs: any exact state $2$-design contains a state requiring $\Omega(n)$ non-Gaussian gates, and any $\epsilon$-approximate state $2$-design contains states requiring $\Omega(\sqrt{n\log(1/\epsilon)})$ non-Gaussian gates with high probability over the ensemble. We further develop an approximate variant of the bridge degree that is robust under perturbation and admits an efficiently computable lower bound, through which the single-shot magic-state cost of any non-Gaussian state can be lower-bounded directly from Bell-sampling data. We then extend the framework to mixed states via a Choi-based construction, with structural properties analogous to those of the pure-state version. Finally, Bell sampling alone yields two concrete algorithmic protocols, both with polynomial sample complexity: (i) a Gaussianity test with perfect completeness, optimal among two-copy tests with this property; and (ii) a certification protocol for the state $2$-design property of matchgate-invariant ensembles, in both the exact and approximate cases.

\subsection{Open problems}
\label{sec:open_problems}

Our work opens several research directions.

\emph{Connections to classical simulability.} A natural direction concerns the connection of the bridge degree to \emph{classical simulability}. We have shown that the bridge degree lower-bounds the number of non-Gaussian gates required to prepare a state, but whether this gives an efficient simulation cost remains open. Establishing that a small bridge degree implies efficient classical simulation would identify the bridge degree as the unifying complexity parameter of fermionic non-Gaussian computation.

\emph{Experimental applications.} On the experimental side, the bridge degree is a compelling candidate for benchmarking and comparing the non-Gaussianity of fermionic platforms --- an area that has so far lacked such a tool. The hardware efficiency of Bell sampling makes the methods developed here directly applicable across a wide range of fermionic architectures, and we expect them to enable the first experiments witnessing non-Gaussianity on near-term fermionic devices, such as ultracold atoms in optical lattices~\cite{Koepsell2021,Brown2019,Hartke2023, Xu2025,vijayan2020timeresolved,jordens2008mott} and digital quantum computers~\cite{GonzlezCuadra2023,Chien2026,Bojovi2026}.

\emph{Mixed-state algorithm.} While our results provide a clean pure-state theory with efficient measurement protocols, the mixed-state counterpart is comparatively limited: the mixed-state bridge spectral distribution is not directly accessible from copies of the state via Bell sampling. Nevertheless, for global depolarizing noise, we show that error mitigation can reconstruct the ideal bridge spectral distribution from Bell-sampling data. Extending this approach to broader and more realistic noise models --- or developing an efficient quantum algorithm that samples from or estimates the mixed-state bridge spectral distribution --- would extend the experimental reach of the framework to noisy quantum devices. 

\emph{Mixed-state monotones and witnesses.} A central open question is to develop mixed-state non-Gaussianity monotones and witnesses, particularly ones that are experimentally accessible. On the monotone side, no efficiently computable mixed-state monotone is currently known. A related question is whether membership in the convex-Gaussian set can be decided efficiently; a negative answer would rule out any efficiently computable faithful monotone for this set. On the witnessing side, our fidelity witness applies to mixed states, but it requires prior knowledge of a state with which the experimentally prepared state is expected to have high fidelity, for example an intended pure target with an efficient classical description. Constructing state-agnostic witnesses that certify mixed-state non-Gaussianity without such prior information would therefore be especially valuable, following recent developments in non-stabilizerness~\cite{haug2026efficient,warmuz2024magic,macedo2025witnessing} and entanglement~\cite{guhne2009entanglement,tarabunga2025quantifying,rico2024entanglement,rico2025state}.

\emph{Sample-optimal Gaussianity testing.} A natural question in property testing is whether \emph{constant}-copy (independent of $n$) Gaussianity testing is possible. A hint in this direction comes from the parallel with stabilizer-state testing: a similar Bell-sampling-based strategy is known to test stabilizer states with a constant number of copies~\cite{gross2021schur}. Theorem~\ref{thm:test_tdoped} provides partial evidence that the same holds here, achieving such a rate for states generated by few non-Gaussian gates, and motivates the conjecture of that a system-size-independent sample complexity holds for general states. Whether this conjecture holds --- or whether some dependence on $n$ is unavoidable --- remains open. For the tolerant variant, an analogous question concerns the regime of validity: our test is applicable under the assumption that the input state is within $\mathcal{O}(1/n)$ trace distance of an FGS in the accept case. Whether the test can be extended to a constant-gap regime is an intriguing direction.

\emph{Extensions of design properties.} Our state-design results concern ensembles of states produced by $t$-doped Gaussian circuits; a natural extension is to study the unitary-design property of $t$-doped Gaussian circuits themselves, as well as alternative notions of approximation such as relative error. A further direction is to extend the analysis to higher-order matchgate commutants $\mathrm{Com}_k(M_n)$ for $k\geq 3$, whose structure has only recently been characterized~\cite{sierant2026theorymatchgatecommutant,LasMou26,BrDi26}; a natural question is a $k$-copy analogue of the bridge spectral distribution, which would open the door to a structural analysis of higher-order state and unitary designs.

{\it Note added}: While completing this manuscript, we became aware of the independent work of Ref.~\cite{haug2026practicaltestswitnessesfermionic}, which considers the same two-copy Bell-sampling Gaussianity test with perfect completeness and establishes $\mathcal{O}(n^2/\epsilon)$ sample complexity, with the bound subsequently improved to $\mathcal{O}(n/\epsilon)$ in v3 of their preprint. We additionally establish that this test is optimal among two-copy tests with perfect completeness, give system-size-independent bound for states generated by few non-Gaussian gates, and extend it to the tolerant regime. Beyond this overlap, the two works address different aspects of fermionic non-Gaussianity.

\begin{acknowledgments}
We thank Barbara Kraus for comments on the manuscript. We thank Bernhard Jobst, Marc Langer, Sheng-Hsuan Lin, Ra{\'u}l Morral-Yepes,  Frank Pollmann, Piotr Sierant, and Xhek Turkeshi for collaborations on related topics. We thank Oliver Reardon-Leaman for pointing out an error in an earlier version of our proofs of Theorems~\ref{thm:irreversibility} and~\ref{thm:unbounded_irreversibility}. We acknowledge funding from the European Research Council (ERC) under the European Union (ERC, DynaQuant, No. 101169765).
\end{acknowledgments}

\appendix

\section{Connection to generalized Pl\"ucker hierarchies}
\label{app:plucker}

The bridge degree induces a natural hierarchy of even pure states through the nested sets $\mathcal{F}_k \coloneqq \{\psi : \nGMax(\psi)\leq k\}$, with $\mathcal{F}_0$ the pure FGSs and $\mathcal{F}_{4\lfloor n/4\rfloor}$ the full even Hilbert space. A structurally similar hierarchy has been introduced in the \emph{particle-preserving} setting by Ref.~\cite{semenyakin2025classifying}, which we now recall and connect to our framework. Defining the operator~\cite{smirnov2016instanton}
\begin{equation}\label{eq:Omega_op}
    \Omega \;\coloneqq\; \sum_{r=1}^{n}\psi_r\otimes\psi_r^\dagger
\end{equation}
acting on two copies, a pure state $\ket{\psi}$ is a Slater determinant if and only if it satisfies the Pl\"ucker relations~\cite{smirnov2016instanton,miwa2000solitons}
\begin{equation}\label{eq:plucker_standard}
    \Omega\ket{\psi}^{\otimes 2} \;=\; 0.
\end{equation}
Ref.~\cite{semenyakin2025classifying} introduces the \emph{generalized Pl\"ucker relations}
\begin{equation} \label{eq:gen_plucker}
    \Omega^k\ket{\psi}^{\otimes 2} \;=\; 0, \qquad k\geq 1,
\end{equation}
whose solution sets $\mathcal{G}_k$ form a nested chain $\mathcal{G}_1\subseteq\mathcal{G}_2\subseteq\cdots$, recovering the Slater determinants at $k=1$ and exhausting the particle-preserving Hilbert space at $k = n+1$.

Motivated by the parallel with the bridge-degree hierarchy, we define the \emph{Pl\"ucker degree} of $\ket{\psi}$ as the smallest $k$ for which $\ket{\psi}\in\mathcal{G}_{k+1}$. Adapting the techniques developed here for the bridge degree (Theorem~\ref{thm:max_Lambda_monotone}), the Pl\"ucker degree can be shown to be a monotone under particle-preserving Gaussian protocols. In particular, monotonicity under post-selected occupation-number measurements follows from the commutativity $[\Omega,\,P_\mu\otimes P_\mu] = 0$, which implies that the post-selected state is always in $\mathcal{G}_k$ whenever $\ket{\psi}\in\mathcal{G}_k$ --- the particle-preserving counterpart of the relation driving the bridge-degree monotonicity proof. This endows the hierarchy of Ref.~\cite{semenyakin2025classifying} with a resource-theoretic interpretation.

\section{Certifying the bridge degree under global depolarizing noise}
\label{app:depol}

As discussed in the main text, the mixed-state bridge spectral distribution is not directly accessible from Bell sampling on copies of a mixed state $\rho$. Nevertheless, when the mixedness arises from global depolarizing noise acting on an underlying pure state, the bridge spectral distribution of the ideal state can be recovered by error mitigation. Global depolarizing noise is a standard and experimentally well-motivated model, as it captures the effective noise of many noisy quantum processors~\cite{haug2023scalable,urbanek2021mitigating,dalzell2024random,vovrosh2021simple}. In this appendix, we first show that Bell sampling on the noisy state yields a convex combination of the ideal bridge spectral distribution and a known noise contribution. This decomposition can then be inverted to reconstruct the ideal bridge spectral distribution from the measured data. Finally, we present a certification algorithm that uses the error-mitigated distribution to obtain a lower bound on the bridge degree of the underlying pure state.

\subsection{Error mitigation of the bridge spectral distribution}

We consider the globally depolarized state
\begin{equation}\label{eq:depol_model}
    \rho=(1-p)\,\ketbra{\psi}{\psi}+p\,\frac{I}{d},
\end{equation}
where $\ket{\psi}$ is an even pure state and $p\in[0,1]$. Our goal is to recover the bridge spectral distribution of the noise-free state
$\ket{\psi}$ from measurements on the noisy state $\rho$. The Bell-sampling primitive
(Algorithm~\ref{alg:bell_sampling}) applied to two copies of $\rho$ produces outcome
$\lambda$ with probability
\begin{equation}\label{eq:mixed_bell_distribution}
    p_\rho(\lambda)\;\coloneqq\;\Tr\!\bigl[\Pi_\lambda\,\rho^{\otimes 2}\bigr].
\end{equation}
We emphasize that $p_\rho(\lambda)$ is the probability distribution associated with the physical two-copy measurement, and should not be confused with the mixed-state bridge spectral distribution $\tilde p_\rho$ defined in Eq.~\eqref{eq:mixed_bridge_spectral}.

We now show that, for
every even $\lambda\in[-2n,2n]$,
\begin{equation}\label{eq:depol_mixture}
    p_\rho(\lambda)\;=\;\eta\,p_\psi(\lambda)\;+\;(1-\eta)\,p_I(\lambda),
    \qquad \eta\coloneqq(1-p)^{2},
\end{equation}
where $p_I$ is the bridge spectral distribution of the maximally mixed state $I/d$,
\begin{equation}\label{eq:p_I}
    p_I(\lambda)\;\coloneqq\;\Tr\!\bigl[\Pi_\lambda\,(I/d)^{\otimes 2}\bigr]\;=\;\frac{1}{4^{\,n}}\binom{2n}{\,n+\lambda/2\,}.
\end{equation}
To show this, we expand 
\begin{equation}
    \rho^{\otimes2}=(1-p)^2\,\psi^{\otimes2}+(1-p)p\,(\psi\otimes\tfrac{I}{d}+\tfrac{I}{d}\otimes\psi)+p^2\,(\tfrac{I}{d})^{\otimes2},    
\end{equation}
 and take the overlap with $\Pi_\lambda$. The first
term gives $\Tr[\Pi_\lambda\,\psi^{\otimes2}]=p_\psi(\lambda)$ and the last
term gives $p_I(\lambda)$ by definition, while the two cross terms are equal by the exchange symmetry of $\Pi_\lambda$ and reduce
to $\Tr[\Pi_\lambda(\psi\otimes\tfrac{I}{d})]=\tfrac1d\bra{\psi}(\Tr_2\Pi_\lambda)\ket{\psi}$.
Since the eigenbasis of $\Lambda$ consists of tensor products of Bell states
(Lemma~\ref{lem:eigenbasis_bell}), each of which is maximally entangled across the two copies with
single-copy reduced state $I/d$, thus yielding $\Tr_2\Pi_\lambda=(\dim V_\lambda/d)\,I$ where
$\dim V_\lambda=\binom{2n}{n+\lambda/2}$. Hence, the cross term equals
$\dim V_\lambda/d^2=p_I(\lambda)$. Collecting the three contributions then gives
Eq.~\eqref{eq:depol_mixture} with $\eta=(1-p)^2$.

Inverting Eq.~\eqref{eq:depol_mixture}, the bridge spectral distribution of the ideal state $\ket{\psi}$ can be reconstructed from the measured distribution $p_\rho$, yielding the error-mitigated distribution
\begin{equation}\label{eq:algo1_correction}
p_\psi^{\mathrm{mtg}}(\lambda)
=\frac{p_\rho(\lambda)-(1-\eta)\,p_I(\lambda)}{\eta},
\end{equation}
for every $\lambda\in\mathcal{A}$. The parameter $\eta$ can also be estimated directly from the same Bell-sampling data. Since the purity of the noisy state satisfies
$\Tr\rho^{2}=\eta+(1-\eta)/d$, we obtain
\begin{equation}\label{eq:eta_from_purity}
    \eta=\frac{d\,\Tr\rho^{2}-1}{d-1}.
\end{equation}
The purity can be written as $\Tr\rho^{2}=\Tr[\mathbb{S}\rho^{\otimes2}]$, with $\mathbb{S}$ the swap operator, whose expectation value is directly accessible through Bell sampling. Alternatively, $\eta$ can be extracted from any ``forbidden'' sector with $\lambda\notin 8\mathbb{Z}$. Indeed, for an even state, $p_\psi(\lambda)=0$ in these sectors by Lemma~\ref{lem:allowed_eig_Lambda}, and Eq.~\eqref{eq:depol_mixture} simplifies to
$p_\rho(\lambda)=(1-\eta)p_I(\lambda)$, giving $\eta=1-p_\rho(\lambda)/p_I(\lambda)$.
Thus, the entire mitigation procedure consists solely of classical post-processing of the Bell-sampling data. The resulting error-mitigated version of Algorithm~\ref{alg:bell_sampling} is therefore straightforward: perform Bell sampling on two copies of $\rho$, estimate $p_\rho$ and $\eta$, and reconstruct the mitigated distribution using Eq.~\eqref{eq:algo1_correction}.

\subsection{Error-mitigated lower bound on the bridge degree}

We now show how to witness the bridge degree in the presence of global depolarizing noise via error mitigation. Under depolarizing noise, every eigenvalue sector acquires a nonzero contribution from $p_I(\lambda)$, including sectors that are unpopulated in the ideal distribution $p_\psi$. Nevertheless, the support of the noise-free distribution can still be identified from the measured distribution through the criterion
\begin{equation}\label{eq:sector_test}
p_\psi(\lambda)>0
\quad\Longleftrightarrow\quad
p_\rho(\lambda)>(1-\eta)\,p_I(\lambda).
\end{equation}
It follows that the bridge degree of the underlying noise-free state is given by
\begin{equation}\label{eq:noisy_bridge_degree}
    \nGMax(\psi)\;=\;\max\Bigl\{\tfrac{|\lambda|}{2}\ :\ p_\rho(\lambda)>(1-\eta)\,p_I(\lambda)\Bigr\}.
\end{equation}
This provides the error-mitigated analogue of the empirical bridge-degree estimator
$\widehat\Lambda_{\mathrm d}=\tfrac12\max_k|\lambda_k|$ [Eq.~\eqref{eq:bridge_degree_witness}]. Specifically, the occupied sectors of the
noise-free state can be identified by checking whether
$p_\rho(\lambda)>(1-\eta)p_I(\lambda)$. We introduce a statistically robust
finite-sample procedure, summarized in Algorithm~\ref{alg:noise_robust_maxlambda}, which yields a certified lower bound on the bridge degree with high probability.

The key quantity is the empirical residual
\begin{equation}\label{eq:emp_mtg}
    \widehat r(\lambda)
    =
    \widehat p_\rho(\lambda)
    -
    (1-\widehat\eta)p_I(\lambda),
\end{equation}
which estimates the ideal contribution
$r(\lambda)=p_\rho(\lambda)-(1-\eta)p_I(\lambda)=\eta p_\psi(\lambda)$.
A sector $\lambda\in\mathcal A$ is \emph{certified} as populated whenever
$\widehat r(\lambda)$ exceeds a confidence threshold. The error-mitigated bridge-degree estimator
$\widehat{\Lambda}^{\mathrm{mtg}}_{\mathrm d}$ is then defined as the largest
certified sector, which provides a lower bound on the true bridge degree. The following proposition proves the correctness of the procedure.

\begin{algorithm}[H]
\caption{Error-mitigated bridge degree lower bound under global depolarizing noise}
\label{alg:noise_robust_maxlambda}
\begin{flushleft}
\textbf{Input:} $N\geq 1$, access to $2N$ identical copies of an even state $\rho$ on
$n$ qubits, and confidence level $\alpha\in(0,1)$. \\
\textbf{Output:} $\widehat{\Lambda}^{\mathrm{mtg}}_{\mathrm{d}}$, which under the global
depolarizing model of Eq.~\eqref{eq:depol_model} satisfies
$\widehat{\Lambda}^{\mathrm{mtg}}_{\mathrm{d}}\leq\nGMax(\psi)$ with probability at
least $1-\alpha$.
\end{flushleft}
\begin{algorithmic}[1]
\State Run Algorithm~\ref{alg:bell_sampling} on $2N$ copies of $\rho$ to obtain
$(\lambda_1,\dots,\lambda_N)$, and form the empirical distribution
$\widehat p_\rho(\lambda)\gets N^{-1}\sum_{k=1}^N\mathbf 1\{\lambda_k=\lambda\}$.
\State Compute the maximally-mixed distribution $p_I(\lambda)\gets 4^{-n}\binom{2n}{n+\lambda/2}$
[Eq.~\eqref{eq:p_I}] for all even $\lambda\in[-2n,2n]$.
\State From the forbidden sectors form $\widehat Q_f\gets\sum_{\lambda\notin 8\mathbb Z}\widehat p_\rho(\lambda)$ and the known $Q_f^{0}\gets\sum_{\lambda\notin 8\mathbb Z} p_I(\lambda)$, and set $\widehat\eta\gets 1-\widehat Q_f/Q_f^{0}$.
\State Set $S\gets|\mathcal A|$.
\State Initialize the certified set $\mathcal C\gets\varnothing$.
\For{$\lambda\in\mathcal A$}
\State Set the certification margin
$\delta_\lambda\gets\bigl(1+p_I(\lambda)/Q_f^{0}\bigr)\sqrt{\log(S/\alpha)/(2N)}$.
    \If{$\widehat p_\rho(\lambda)-(1-\widehat\eta)\,p_I(\lambda) \;>\; \delta_\lambda$}
        \State $\mathcal C\gets\mathcal C\cup\{\lambda\}$.
    \EndIf
\EndFor
\State \Return $\widehat{\Lambda}^{\mathrm{mtg}}_{\mathrm{d}}\gets
\tfrac12\max\{|\lambda|:\lambda\in\mathcal C\}$, or $0$ if $\mathcal C=\varnothing$.
\end{algorithmic}
\end{algorithm}

\begin{prop}[Validity of the error-mitigated bridge degree lower bound]\label{prop:noise_robust_valid}
Under the global depolarizing model of Eq.~\eqref{eq:depol_model} with $p<1$, the output of
Algorithm~\ref{alg:noise_robust_maxlambda} satisfies
$\widehat{\Lambda}^{\mathrm{mtg}}_{\mathrm{d}}\leq\nGMax(\psi)$ with probability at
least $1-\alpha$.
\end{prop}

\begin{proof}
Define the residual
\begin{equation}
r(\lambda)
\coloneqq
p_\rho(\lambda)-(1-\eta)p_I(\lambda)
=
\eta\,p_\psi(\lambda),
\end{equation}
where the second equality follows from Eq.~\eqref{eq:depol_mixture}. Hence, $r(\lambda)\ge0$, with equality if and only if
$p_\psi(\lambda)=0$ (since $\eta=(1-p)^2>0$).

Applying the Bell-sampling primitive
(Algorithm~\ref{alg:bell_sampling}) $N$ times to two copies of $\rho$ produces eigenvalues
$\lambda_1,\dots,\lambda_N$, whose empirical distribution is
$\widehat p_\rho(\lambda)=N^{-1}\sum_{k=1}^N\mathbf 1\{\lambda_k=\lambda\}$. For each outcome $\lambda_k$, define
\begin{equation}
X_k^{(\lambda)}
=
\mathbf1\{\lambda_k=\lambda\}
-
\frac{p_I(\lambda)}{Q_f^0}
\mathbf1\{\lambda_k\notin8\mathbb Z\},
\end{equation}
where
$Q_f^0=\sum_{\lambda\notin8\mathbb Z}p_I(\lambda)$.
Then
$\widehat r(\lambda)=N^{-1}\sum_{k=1}^N X_k^{(\lambda)}$ is an unbiased estimator for $r(\lambda)$,
since
$\mathbb E[X_k^{(\lambda)}]
=p_\rho(\lambda)-p_I(\lambda)Q_f/Q_f^0
=r(\lambda)$,
where $Q_f=\sum_{\lambda\notin8\mathbb Z}p_\rho(\lambda')=(1-\eta)Q_f^0$.
Moreover,
$X_k^{(\lambda)}\in[-p_I(\lambda)/Q_f^0,\,1]$, so Hoeffding's inequality yields
\begin{equation}
\Pr\!\left(
\widehat r(\lambda)-r(\lambda)>\delta_\lambda
\right)
\le
\exp\!\left(
-\frac{2N\delta_\lambda^2}
{\left(1+p_I(\lambda)/Q_f^0\right)^2}
\right) = \frac{\alpha}{S},
\end{equation}
with
$\delta_\lambda=\left(1+p_I(\lambda)/Q_f^0\right)\sqrt{\log(S/\alpha)/(2N)}$. Applying a union bound over the $S=|\mathcal A|$ sectors, we obtain
$\widehat r(\lambda)\le r(\lambda)+\delta_\lambda$
simultaneously for all $\lambda\in\mathcal A$ with probability at least
$1-\alpha$.

Therefore, whenever
$\widehat r(\lambda)>\delta_\lambda$,
we have
$r(\lambda)\ge\widehat r(\lambda)-\delta_\lambda>0$,
implying $p_\psi(\lambda)>0$.
Hence every certified sector belongs to the support of $p_\psi$, and therefore
\begin{equation}
\widehat{\Lambda}^{\mathrm{mtg}}_{\mathrm d}
=
\frac12\max_{\lambda\in\mathcal C}|\lambda|
\le
\nGMax(\psi)
\end{equation}
with probability at least $1-\alpha$.
\end{proof}

The above confidence statement also implies a finite-sample complexity bound. Let $\lambda$ be a sector with $p_\psi(\lambda)>0$. Since certification requires $\widehat r(\lambda)>\delta_\lambda$, a sufficient condition is $r(\lambda)>2\delta_\lambda$. Using $r(\lambda)=\eta p_\psi(\lambda)$, this gives the requirement
\begin{equation}
N>
\frac{2\left(1+p_I(\lambda)/Q_f^0\right)^2}
{\eta^2p_\psi(\lambda)^2}
\log\frac{S}{\alpha}.
\end{equation}
Thus, resolving a sector with weight $p_\psi(\lambda)$ requires $N=
\mathcal{O}\!\left((1-p)^{-4}p_\psi(\lambda)^{-2}\log(S/\alpha)\right)$
Bell samples. For the purpose of estimating the bridge degree, it is sufficient to certify only the sector with the largest magnitude of \(\lambda\) among the occupied sectors. Hence, the relevant sample complexity is determined by the weight of this extremal value $\lambda_\star$, and scales as
\begin{equation}
N=
\mathcal{O}\!\left((1-p)^{-4}p_\psi(\lambda_\star)^{-2}
\log(S/\alpha)\right).
\end{equation}

The use of Hoeffding's inequality in Algorithm~\ref{alg:noise_robust_maxlambda} is for simplicity and
provides a distribution-independent confidence guarantee. The finite-sample bound can, in principle, be improved by employing sharper inequalities that
exploit the variance of the estimator, such as the Bernstein inequality, potentially reducing the required sample complexity.

\bibliography{bibliography}

\end{document}